\documentclass[12pt]{article}

\pdfoutput=1

\makeatletter
\def\cl@chapter{}
\makeatother

\usepackage{caption}
\usepackage[capposition=top]{floatrow}
\usepackage[dvipsnames]{xcolor}
\usepackage{adjustbox}
\usepackage{pdflscape}
\usepackage{blkarray}
\usepackage[indentfirst]{titlesec}
\usepackage{amsthm}
\usepackage{amsmath}
\usepackage{pst-all}
\usepackage{amssymb}
\usepackage{geometry}
\usepackage{graphicx}
\usepackage{natbib}
\usepackage{fixfoot,xspace}
\usepackage[bookmarks=false,hyperfootnotes=true]{hyperref}
\hypersetup{breaklinks=true,colorlinks=true,linkcolor=azul-pesc,citecolor=azul-pesc}
\usepackage{cleveref}
\usepackage{libertinus-type1}
\usepackage[cmintegrals,cmbraces]{newtxmath}
\usepackage{enumerate}
\usepackage{lipsum}
\usepackage[title,toc,titletoc]{appendix}
\usepackage{chngcntr}
\usepackage{apptools}
\usepackage{rotating}
\usepackage{float}
\usepackage{multirow}
\floatstyle{plain} 
\restylefloat{figure}
\usepackage{afterpage}
\usepackage{fancyhdr}
\usepackage{relsize}
\usepackage{etoolbox}
\usepackage{anyfontsize}
\usepackage{dsfont}
\usepackage{pgfplots}
\pgfplotsset{height=9cm,width=14cm,compat=1.9}
\allowdisplaybreaks
\usepackage{ragged2e}
\usepackage{mathtools}
\usepackage{fontawesome5}

\renewcommand\thesection{\Roman{section}} 
\renewcommand\thesubsection{\thesection.\arabic{subsection}} 
\titleformat{\section}[block]{\vspace{1em}\large\centering\scshape\bfseries}{\thesection.}{0.5em}{} 
\titleformat{\subsection}[block]{\slshape\centering}{\thesubsection}{0.5em}{}

\theoremstyle{definition}

\numberwithin{corollary}{section} 
\newtheorem{proposition}{Proposition}
\numberwithin{proposition}{section} 
\newtheorem{definition}{Definition}
\numberwithin{definition}{section} 
\newtheorem{assumption}{Assumption}
\numberwithin{assumption}{section} 
\newtheorem{lemma}{Lemma}
\numberwithin{lemma}{section} 
\newtheorem{theorem}{Theorem}
\numberwithin{theorem}{section}

\definecolor{azul-pesc}{RGB}{0, 52, 90}

\usepackage[protrusion=true,expansion=true]{microtype} 

\title{Portfolio Choice In Dynamic Thin Markets:\\ Merton Meets Cournot}

\author{\textsc{Puru Gupta}\footnote{\footnotesize Department of Economics, University of Warwick, Coventry CV$4$ $7$AL, United Kingdom.}\\ \footnotesize \scalebox{0.7}{\faIcon{envelope}} puru.gupta.1@warwick.ac.uk \and \textsc{Saul D. Jacka}\footnote{\footnotesize The Alan Turing Institute, London NW$1$ $2$DB, United Kingdom \& Department of Statistics, University of Warwick, Coventry CV$4$ $7$AL, United Kingdom. \newline  \footnotesize We are grateful to Ilan Kremer for numerous detailed insightful discussions. Puru wishes to thank Peter J. Hammond for pastoral support as well as Qianxue (Zoe) Zhang for helpful comments and suggestions. We also acknowledge feedback received from participants at the Stochastic Control and Financial Engineering Workshop, Princeton University, at the Financial Mathematics Workshop, Oxford\textendash Man Institute of Quantitative Finance, University of Oxford, and at the SIAM Conference on Financial Mathematics and Engineering (FM23).}\\ \footnotesize \scalebox{0.7}{\faIcon{envelope}} s.d.jacka@warwick.ac.uk}

\date{}

\numberwithin{equation}{section}
\begin{document}

\maketitle
\begin{abstract}
\par We consider an augmented version of Merton's portfolio choice problem, where trading by \textit{large} investors influences the price of underlying financial asset leading to strategic interaction among investors, with investors deciding their trading rates independently and simultaneously at each instant, in the spirit of dynamic Cournot competition, modelled here as a non-zero sum singular stochastic differential game. We establish an equivalence result for the value functions of an investor's best-response problem, which is a singular stochastic optimal control problem, and an auxiliary classical stochastic optimal control problem by exploiting the invariance of the value functions with respect to a diffeomorphic integral flow associated with the drift coefficient of the best-response problem. Under certain regularity conditions, we show that the optimal trajectories of the two control problems coincide, which permits analytical characterization of Markov\textendash Nash equilibrium portfolios. For the special case when asset price volatility is constant, we show that the unique Nash equilibrium is deterministic, and provide a closed-form solution which illuminates the role of imperfect competition in explaining the excessive trade puzzle.
\end{abstract}

\par \textbf{Keywords} - \small{Singular Stochastic Differential Game, Imperfect Competition, Institutional Investors, \par Portfolio Choice, Liquidity Frictions}\normalsize\\
\par \textbf{JEL Classification} - \small{C73, D43, G11, L13}\normalsize
\par \textbf{MS Classification} (2020) - \small{49N70, 91A05, 91A15, 91B54}\normalsize

\newpage

\section{Introduction}
\par The canonical portfolio choice theory in continuous time, which was pioneered by Robert Merton in a series of works written during the late sixties and early seventies (see for example \cite{merton1969lifetime} and \cite{merton1971opt}) relies crucially on the assumption of investors being \textit{small} in the sense that they are assumed to be price-takers. This theoretical framework serves as a reasonable approximation of financial markets which are populated primarily by individual households who trade for liquidity purposes and are unlikely to possess insider information and/or contribute much in the way of exercising effective control over management in order to generate appreciable price movements by their trading.

\par However, in their respective presidential addresses delivered to the American Finance Association in the wake of the Great Recession of 2007-08, \cite{french2008presidential} and \cite{stambaugh2014presidential} highlighted the prominent decline witnessed in the share of corporate equity holdings owned directly by households in the United States over a span of six decades, and discussed its potential consequences for investor behaviour in financial markets. Both of them surmised that the declining pool of retail investors has a direct bearing on competition for superior gains in equity trading. \textcolor{azul-pesc}{\Cref{fig:1}} below depicts the dynamics of U.S. corporate equity ownership from $1950$ to $2020$ where the reported estimates are based primarily on data sourced from the December $9$, $2022$ release of the flow of funds account issued by the U.S. Federal Reserve Board. Details on the methodology used to obtain the reported estimates\footnote{The estimates we obtain as well as the resulting trends in the share of corporate equity holdings depicted here, are qualitatively similar to \cite{french2008presidential} who reports these through $2007$, and \cite{stambaugh2014presidential} whose analysis extends through $2012$.} can be found in the Data Appendix.

\begin{figure}[h]
\centering
\begin{tikzpicture}
\begin{axis}[
    x tick label style={/pgf/number format/1000 sep=, legend pos=north east },
    ylabel={Share Of Corporate Equity Holdings},
    xlabel={Year},
    xmin=1950, xmax=2020,
    ymin=0, ymax=1,
    xtick={1950,1960,1970,1980,1990,2000,2010,2020},
    ytick={0.2,0.4,0.6,0.8,1},
    ymajorgrids=false,
    grid style=dashed,
]
\addplot[
    color=black,
    solid,
    mark=,
    smooth
    ]
    coordinates {
    (1950,0.93)(1951,0.93)(1952,0.92)(1953,0.91)(1954,0.90)(1955,0.89)(1956,0.89)(1957,0.88)(1958,0.88)(1959,0.87)(1960,0.86)(1961,0.85)(1962,0.85)(1963,0.83)(1964,0.82)(1965,0.82)(1966,0.81)(1967,0.80)(1968,0.79)(1969,0.78)(1970,0.75)(1971,0.73)(1972,0.73)(1973,0.71)(1974,0.65)(1975,0.62)(1976,0.64)(1977,0.62)(1978,0.59)(1979,0.59)(1980,0.59)(1981,0.59)(1982,0.53)(1983,0.52)(1984,0.48)(1985,0.45)(1986,0.45)(1987,0.44)(1988,0.47)(1989,0.49)(1990,0.48)(1991,0.49)(1992,0.50)(1993,0.49)(1994,0.47)(1995,0.47)(1996,0.35)(1997,0.33)(1998,0.32)(1999,0.31)(2000,0.32)(2001,0.29)(2002,0.28)(2003,0.26)(2004,0.27)(2005,0.26)(2006,0.29)(2007,0.28)(2008,0.26)(2009,0.22)(2010,0.24)(2011,0.25)(2012,0.24)(2013,0.24)(2014,0.24)(2015,0.24)(2016,0.25)(2017,0.25)(2018,0.26)(2019,0.26)(2020,0.24)
    };
    \addlegendentry{Households (Direct)}
\addplot[
    color=azul-pesc,
    solid,
    mark=,
    smooth
    ]
    coordinates {
    (1950,0.03)(1951,0.02)(1952,0.03)(1953,0.03)(1954,0.03)(1955,0.03)(1956,0.04)(1957,0.04)(1958,0.04)(1959,0.05)(1960,0.05)(1961,0.05)(1962,0.05)(1963,0.05)(1964,0.05)(1965,0.05)(1966,0.05)(1967,0.06)(1968,0.06)(1969,0.05)(1970,0.06)(1971,0.06)(1972,0.05)(1973,0.04)(1974,0.04)(1975,0.05)(1976,0.04)(1977,0.04)(1978,0.04)(1979,0.04)(1980,0.04)(1981,0.04)(1982,0.05)(1983,0.05)(1984,0.06)(1985,0.07)(1986,0.07)(1987,0.07)(1988,0.07)(1989,0.08)(1990,0.08)(1991,0.08)(1992,0.10)(1993,0.12)(1994,0.14)(1995,0.16)(1996,0.19)(1997,0.20)(1998,0.21)(1999,0.23)(2000,0.20)(2001,0.22)(2002,0.20)(2003,0.26)(2004,0.26)(2005,0.26)(2006,0.28)(2007,0.27)(2008,0.20)(2009,0.32)(2010,0.30)(2011,0.27)(2012,0.28)(2013,0.30)(2014,0.28)(2015,0.28)(2016,0.28)(2017,0.30)(2018,0.26)(2019,0.30)(2020,0.31)
    };
    \addlegendentry{Mutual Funds}
\end{axis}
\end{tikzpicture}
\floatfoot{Note: At the end of the second world war over ninety percent of corporate equity in the United States was owned directly by households. This value has witnessed a steady decline since the post-war period and achieved an all-time low of around thirty percent during the 2007-08 Great Recession. On the other hand, the share of mutual funds which stood at a lowly five percent in 1980, today accounts for roughly thirty percent of corporate equity holdings in the United States.}
\caption{Evolution of U.S. Corporate Equity Holdings (1950-2020)}
\label{fig:1}
\end{figure}
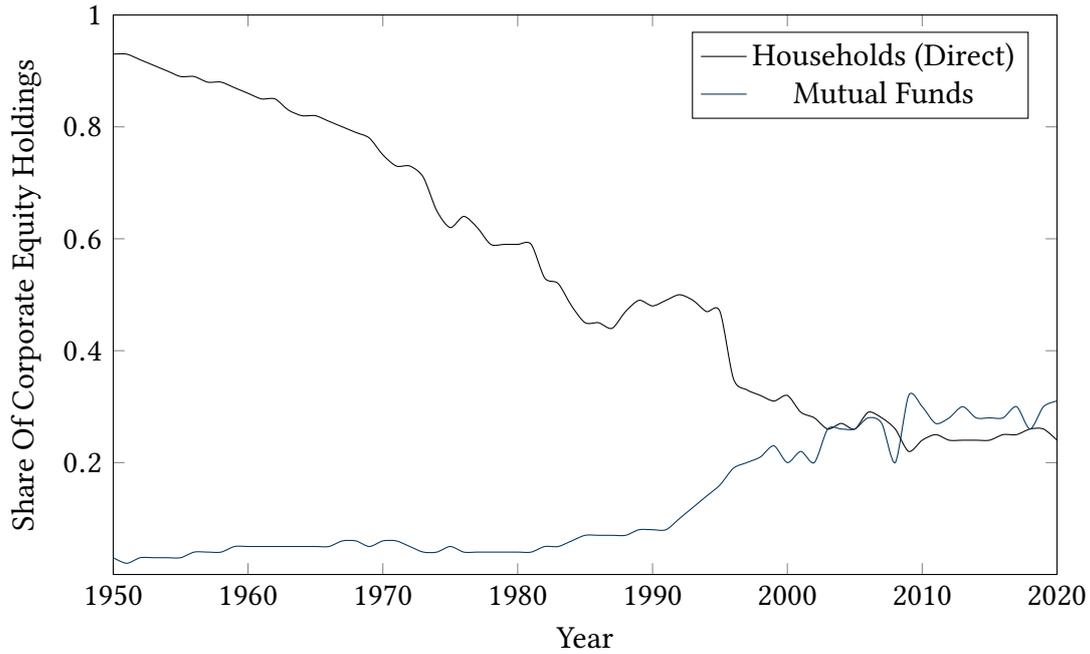

\par In their seminal work \cite{rostek2015dynamic} coined the term \textit{Dynamic Thin Markets} to describe contemporary equity markets which are characterized by the dominance of a relatively small number of institutional investors, with regard to holdings as well as trading; importantly, trades by these institutional investors have non-negligible influence on the price of the underlying traded asset. These institutional or \textit{large} investors, in turn, are aware of the impact their trades have on prices and do account for this \textit{price impact} while determining their optimum portfolio holdings. It is thus clear that the dynamic thin market paradigm represents an existential challenge to canonical portfolio choice theory, and in particular its foundational assumption of price-taking investors.

\par Our principal objective in the present work is to reconcile the empirical fact categorizing modern equity markets as dynamic thin markets, with classical portfolio choice theory which retains its appeal by virtue of its elegance and tractability. To this end, we consider an augmented version of Merton's portfolio choice problem with two large (institutional) investors who have non-zero price impact and who wish to maximize expected utility associated with their portfolio value at some finite terminal time. In line with the classical portfolio choice theory, the terminal portfolio value of an investor naturally will depend upon the price of the underlying financial asset. However, given that the two institutional investors now have non-zero price impact, the price of the asset in turn depends upon the portfolios of the two large investors which then leads to strategic interaction between them which we model and analyze as a non-cooperative stochastic differential game.

\par The augmented Merton setup we consider resembles a dynamic duopoly \textit{\`{a} la} \cite{cournot1838recherches}, in which at each instant, the investors decide the quantity of the financial asset they wish to trade simultaneously and independently of each other, and the price they pay/receive is then determined by a (stochastic) clearing rule. We thus refer to our setup as stochastic differential Merton\textendash Cournot game. It turns out that analyzing the Merton\textendash Cournot stochastic differential game is not straightforward as the best-response problem of an investor happens to be a singular optimal control problem. Heuristically, an optimal control problem is singular when its associated Hamiltonian can be extended real-valued on a subset of its domain. Interested readers are directed to \cite[Section 3.4.2]{pham2009continuous} for an accessible introduction. 

\par Singular optimal control problems are notoriously difficult to solve. The predominant approach to dealing with such problems is the equivalence approach pioneered by \cite{bather1967sequential}, which entails construction of an auxiliary classical optimal control problem that is equivalent to the original singular control problem in an appropriate sense and thus allows us to tackle the singular control problem by standard methods. In this vein, we note that the best-response problem of an investor in our setup is closely related to the class of singular optimal control problems considered in \cite{lasry2000classe} and we therefore extend and adapt their approach to our setup. For notable alternative approaches to optimal singular stochastic control, see \cite{karatzas1983class}, \cite{karatzas1984connections}, \cite{davis1994problem}, \cite{jacka2002avoiding}, \cite{dufour2004singular} and the references therein.

\par Informally, the Lasry\textendash Lions approach constructs an auxiliary control problem with the help of a diffeomorphic flow associated with the drift coefficient of the best-response control problem such that the value function of the original singular best-response problem of an investor remains invariant with respect to the flow, that is, substituting the flow value instead of the state leaves the best-response value function of an investor unchanged. 

\par This is because the construction of the auxiliary control problem employs the diffeomorphic flow to transport a given state to any desired value, instantaneously and frictionlessly, while, the flow parameter is chosen in a manner that eliminates the terms underlying the singular nature of the best-response problem of an investor from the dynamics of the transformed state process. The flow invariance property proves crucial in establishing the equivalence of the constructed auxiliary control problem and the singular best-response problem of an investor.

\par Since, the auxiliary control problem is a classical control problem by construction it permits us to derive the best-response value function of an investor by standard dynamic programming methods for the class of Markovian strategies. Moreover, under mild regularity conditions, we also prove that the optimal trajectories of the two control problems coincide, leading to an analytical characterization of investors' best-response mappings and the associated Markov\textendash Nash equilibrium as a system of coupled ordinary differential equations, which can be solved in closed-form in the special case when asset price volatility is constant, thereby leading us to the striking result that the unique Markov\textendash Nash equilibrium is deterministic.

\par The analysis of closed-form Markov\textendash Nash equilibrium has interesting economic implications, particularly with regard to the excessive trading puzzle first documented by \cite{odean1999investors}. In essence, the excessive trading puzzle refers to the fact that observed trading volumes in stock markets far exceed the predictions based upon standard asset pricing models such as the Capital Asset Pricing Model (CAPM). As is the case with other asset pricing puzzles, there have been numerous proposed explanations mostly rooted in behavioural economics for the excessive trading puzzle. However, as \citet*{liu2022taming} note there is a proliferation of behavioural explanations for asset pricing puzzles, which tend to be adhoc and lack cogency. 

\par Interestingly, \citet*{chen2019resolving} highlight the relevance of price impact in potentially explaining asset-pricing puzzles. In line with this hypothesis, we show that introduction of imperfect competition in the canonical Merton setup generates excessive trading volume compared to the classical price-taking setup, which points to the relevance of imperfect competition in explaining excessive trading puzzle, especially in view of the fact that standard asset pricing models too rely on the assumption of price-taking investors.

\par At first glance, this result represents a paradox since price impact is generally considered as one of the more significant forms of transaction costs faced by an investor in financial markets, \cite{gueant2016financial}. Thus, consideration of price impact costs should discourage investors from trading relative to a frictionless market benchmark. However, unlike a price-taking or a monopolistic framework, in a strategic framework the presence of more than one large investor imposes an externality on the investment opportunity set of a given large investor on account of the price impact of the other large investor(s), compelling an investor to optimally adjust portfolio holdings in response, generating excessive trading.

\par As a closing remark to the introduction, we note the organizational layout of the present work.
The following section discusses relation of our work to existing literature and underscores its
salient contributions. We proceed to the discussion of the specifics of the financial market model in \textcolor{azul-pesc}{\Cref{model}}. Subsequently, we overview the diffeomorphic flow approach to singular control problems in \textcolor{azul-pesc}{\Cref{intflo}}, which then leads us to \textcolor{azul-pesc}{\Cref{auxcp}} where we give the details of the construction of the auxiliary control problem. \textcolor{azul-pesc}{\Cref{vequiv}} presents the results on the equivalence
of the auxiliary and singular best-response problems, while \textcolor{azul-pesc}{\Cref{optpo}} deals with the characterization of Markov\textendash Nash equilibrium and its uniqueness for the case of constant asset price volatility case. Concluding remarks are presented in \textcolor{azul-pesc}{\Cref{conc}}, while proofs omitted in the main text are contained in the Appendix.

\section{Related Literature}

\par The present work adds to the extensive literature on trading in imperfectly competitive financial markets in which traders have non-zero price impact. Early seminal contributions include \cite{kyle1985continuous} and \cite{back1992insider} which examined optimal insider trading in a monopolistic framework where trades signal underlying asset value leading to non-zero price impact. Their setup was later generalised by \cite{collin2016insider} to the class of stochastic volatility models. For an alternative consumption based framework, see \cite{vayanos2001strategic}. 

\par Subsequently, \cite{bertsimas1998optimal} and \cite{almgren2001optimal} pioneered the literature on optimal execution which deals with optimal liquidation of a block of shares by a single large investor who seeks to minimize execution costs arising on account of price impact. Their work influenced several later generalizations such as \cite{ekren2019portfolio} which focuses on the case of transient price impact, and \cite{lions2007large} which extends the optimal execution framework to local volatility models. For a comprehensive survey of related works, see the excellent monograph by \cite{gueant2016financial}.

\par While the literature on monopolistic trading serves as an important precursor, our work shares greater affinity with the literature on \textit{execution games}, with notable contributions by \cite{brunnermeier2005predatory}, \citet*{carlin2007episodic}, and \cite{schied2019market} among others. The principal unifying feature of these works is strategic competition for liquidity among investors who wish to maximize proceeds from unwinding their respective initial portfolio positions, often to zero in the case of portfolio liquidation games, while simultaneously accounting for price impact resulting from trading by all investors. Importantly, in execution games the terminal portfolio target is exogenously specified, whereas in our setup the terminal portfolio is \textit{a priori} unconstrained and determined endogenously. 

\par Within the vast strategic trading literature, our work is closely related to the notable contribution by \citet*{back2000imperfect}, in which the authors generalize the framework considered in \cite{kyle1985continuous} to a multi-player setup. They focus attention on linear feedback Nash equilibrium by exploiting stochastic filtering techniques. Their approach does not extend to the case of symmetrically informed investors, on account of resulting singularity. This assumes importance in view of the empirical finding that institutional investors often fail to outperform a naive market benchmark, see for example \cite{vayanos1999strategic}, and hence seem unlikely to possess superior information due to which they should be considered as conceptually distinct from inside traders. Other notable contributions which consider strategic trading within an asymmetric information framework include \cite{vayanos1999strategic}, \citet*{gabaix2006institutional}, \cite{sannikov2016dynamic}, and \citet*{kyle2018smooth}. 

\par Recently, \citet*{chen2019resolving} and \citet*{micheli2021closed} have focused attention on strategic trading by institutional investors by relaxing the assumption of asymmetrically informed investors. However, the setup considered in these works is fundamentally different from ours. While, \citet*{micheli2021closed} focus on the case of \textit{joint, symmetric, temporary} price impact where investors use closed loop trading strategies which are a function of the trading rate of their opponents, \citet*{chen2019resolving} consider an equilibrium setup where an investor's trading affects only her \textit{perceived} price process, thereby abstracting from joint price impact, and the strategic interaction amongst investors arises on account of a market clearing condition. In contrast, our work examines strategic trading in a setup with \textit{joint, asymmetric, permanent} price impact where investors use Markovian trading strategies, with a particular focus on examining the resulting singular stochastic differential game at first hand.

\par The present work also relates broadly to the literature on trading in financial markets with liquidity frictions, such as \cite{brown2011dynamic}, \cite{garleanu2013dynamic}, \cite{garleanu2016dynamic}, \cite{olivares2018robust}, \cite{hautsch2019large} among others which examine the effect of transaction costs, \cite{garleanu2009portfolio}, and \citet*{ang2014portfolio} which focus on infrequent trading opportunities, \cite{cuoco1998optimal} which considers an alternate form of price impact for a single large investor, \cite{moallemi2017dynamic}, and \cite{longstaff2001optimal} which emphasize the role of trading constraints on portfolio choice, and \cite{obizhaeva2013optimal} which considers a limit-order book setup. In the present work, we focus attention on liquidity constraints arising on account of price impact of market orders in a strategic framework.

\par Equilibrium models of trading which endogenize market power and price impact have been studied in \cite{weretka2011endogenous}, and \cite{carvajal2012no} for a static and a two-period economy respectively. Continuous time equilibrium models of trading are considered in \citet*{anderson2008equilibrium}, \cite{martins2010equilibrium}, \citet*{hugonnier2012endogenous}, and \cite{ehling2015complete}. However, these tend to focus on characterization of market equilibrium and/or dynamic market completeness, whereas by limiting attention to a partial equilibrium setup, we obtain explicit characterization of optimal portfolios in the presence of market frictions. For an excellent overview of equilibrium models of imperfectly competitive financial markets, see the recent survey by \cite{rostek2020equilibrium}.

\par In addition, our paper contributes to the emerging literature on non-zero sum singular stochastic differential games. Notable contributions include \cite{kwon2015game} in which the authors combine games of optimal stopping with singular stochastic control to analyze competition for market share between two firms. This setup is generalized in \cite{kwon2022game} and \cite{kim2022investment} to allow for regular controls and mixed strategies in addition to singular controls. Further, \citet*{wangs2018maximum} and \citet*{wangw2018necessary} introduce a Malliavin calculus approach to analyze non-zero sum mixed regular-singular stochastic differential games. 

\par In contrast to the works cited in the preceding paragraph, we focus exclusively on a stochastic differential game with singular controls. This renders our work closer in spirit to the class of \textit{finite-fuel} games considered in \cite{de2018stochastic}, \cite{guo2019stochastic}, \cite{dianetti2020nonzero}, \citet*{guo2022class}. However, these works differ in their focus and methodology in important respects from ours. While not a zero-sum game, the strategic interaction in \cite{de2018stochastic} is nevertheless a tug-of-war with the motives of the two players aligned against each other, whereas in \cite{guo2019stochastic} players have rank-dependent payoffs. Likewise, \cite{dianetti2020nonzero} work under the additional assumption of submodular costs, while \citet*{guo2022class} focus on \textit{weak} solutions instead of strong solutions which is the focus of the present work.

\par Recently, \citet*{hu2017singular}, \citet*{cao2022stationary}, \citet*{dianetti2022unifying}, \citet*{campi2022mean}, and \citet*{cao2022mfgs} among others, focus on the class of \textit{mean-field} singular stochastic differential games, which exploits methodological advances from the vast literature on mean-field games which has grown extensively since the path-breaking contribution of \cite{lasry2007mean} and \citet*{huang2006large}. While a mean-field setup facilitates tractability by drastically reducing the dimensionality of the problem, it does so at the expense of strong symmetry assumptions. Our approach in contrast circumvents such symmetry considerations and allows explicit consideration of asymmetric equilibrium.

\section{Financial Market Model}\label{model}
\par In subsequent discussion, we maintain the standing assumption of an underlying complete filtered probability space, $\left(\Omega,\mathcal{F},\mathds{F},\mathds{P}\right)$ which satisfies the usual conditions; that is, the filtration $\mathds{F} = \{\mathcal{F}_{t}\}_{t\, \in\, [0,T]}$ is assumed to be right-continuous and $\mathds{P}$\textendash complete, with $\mathcal{F}_{T} = \mathcal{F}$. Additionally, we assume that the space supports a one-dimensional standard Brownian motion $B = \{B_{t}\}_{t\,\in\,[0,T]}$.

\par In this section, we extend the canonical continuous-time portfolio choice model pioneered by \cite{merton1969lifetime}, by incorporating the influence of trading by multiple large investors explicitly into the price dynamics of traded financial assets. We limit attention to permanent price impact, that is, we assume the impact of an institutional investor's trades on the market price of a financial asset is permanent. We do not allow for decay of the price impact or a (partial) reversion towards the pre-trade price following the instantaneous impact and we do not consider execution or slippage costs associated with trading, such as brokerage fees. 

\par As a result, in our framework the temporary price impact equals zero, as in \citet*{bouchard2016almost}. In view of this, our model should be interpreted as a scheduling model, that is, a model which determines the optimal trading schedule of a large investor in a market with multiple large investors. We direct the interested reader to \cite{gueant2016financial} for consideration of execution models, which focus on optimal routing and micromanagement of a given trading schedule, by examining optimal trading behaviour under execution costs.

\par We consider an environment with two institutional investors, indexed by $i \in \mathcal{I} = \left\{1,2\right\}$. The investors share a common investment horizon, $T \in \left(0,\infty\right)$. Each institutional investor can potentially invest in a risk-free asset (cash) and a single risky asset (stock), where the rate of return associated with the risk-free asset is denoted by $r_{t}$. We let $\pi^{i}_{t}$ denote the portfolio of investor $i$ at time $t$. Specifically, $\pi^{i}_{t}$ represents the number of stocks of the risky asset held by the $i$th investor. We assume that the portfolio dynamics of the $i$th investor are characterized as
\begin{equation}\label{control}
	d\pi^{i}_{t} \,=\, -x^{i}_{t}\,dt
\end{equation}

\par Note that $x^{i}_{t}$ above denotes the trading rate of the $i$th investor at time $t$. We assume that $\pi^{i}_{0}$, the initial level of stock holdings, is given for $i \in \mathcal{I}$. It follows from the equation above that when $x^{i}_{t}>0$, the $i$th investor holds an instantaneous selling position in the stock. While, on the other hand $x^{i}_{t}<0$ implies that the $i$th investor holds an instantaneous buying position in the stock. The dynamics specified by (\ref{control}) require that $\pi^{i}_{t}$ be absolutely continuous $\mathds{P}$-almost surely.

\par The instantaneous impact of trading by the $i$th investor on the price of the risky asset is modeled by the function $\kappa^{i}(x^{i}_{t})$. We denote the price of the risky asset at time $t$ by $S_{t}$, and assume that the stock price process follows dynamics governed by the following stochastic differential equation
\begin{equation}
	dS_{t} \,=\, \sigma(S_{t})\,dB_{t}\, +\, \left(\sum_{i} \kappa^{i}\!\left(x^{i}_{t}\right)dt\right)
\end{equation}

\par We assume that the initial price $S_{0}$ is given. The deterministic function $\sigma$ above represents the local volatility of the stock price. We impose the following assumption with regard to the local volatility process
\begin{assumption}\label{siglip}
\par The function $\sigma: \mathds{R}\rightarrow\mathds{R}_{+}$ is bounded, and Lipschitz continuous in its argument.
\end{assumption}

\par We defer our discussion of the properties of the price impact function and instead turn our attention to the problem of specifying a cogent and well-defined criterion associated with the portfolio of an investor, which the investor seeks to optimize. Note that, unlike the classical Merton paradigm, the product $\pi^{i}_{t}\,S_{t}$ no longer equals the liquidation value of the stock in general, and hence does not define the wealth accrued to the investor from investing in the risky asset. We shall, alternatively, think of $\pi^{i}_{t}\,S_{t}$ as describing the \textit{market valuation} of stocks held by the $i$th investor at time $t$. Correspondingly, we define the value of investor $i$'s portfolio at time $t$, denoted by $W^{i}_{t}$, as follows
\begin{equation}
W^{i}_{t}\, =\, \text{Cash Position}\, +\, \pi^{i}_{t}\,S_{t}
\end{equation}

\par Here again we assume that $W^{i}_{0}$, the portfolio value at the initial time, is given for $i \in \mathcal{I}$. This definition might seem arbitrary but as argued in \citet*{bouchard2016almost}, one may consider the pair $(W^{i}_{t}, \pi^{i}_{t})$ at any given time in order to obtain the exact composition of the $i$th investor's portfolio in terms of cash account and stock holdings. Moreover, in the event that the functions $\kappa^{1}, \kappa^{2}$ are identically zero, that is, in the conventional Merton paradigm, the definition above reduces to the canonical wealth equation. In what follows, we work under the simplifying assumption that the risk free rate $r_{t}$ is identically zero for all $t \in [0,T]$, or equivalently, that we are working with discounted values. Under this assumption, the dynamics of $W^{i}_{t}$ are described by the following stochastic differential equation
\begin{equation}\label{portval}
	dW^{i}_{t}\, =\, \pi^{i}_{t}\,dS_{t}
\end{equation}

\par The equation above embodies a conservation principle analogous to the \textit{self-financing condition} ubiquitous in continuous-time models of portfolio choice. Intuitively, it mandates that any change in the portfolio value should accrue on account of a change in the prices of underlying assets alone. A straightforward re-balancing of the portfolio, which does not involve external flows of fund or stock holdings, and which is not accompanied by a price change, should leave the composite portfolio value unchanged.

\par Empirical evidence with regard to the shape of the price impact function is equivocal. On the one hand works such as \cite{hasbrouck1991measuring}, \citet*{gabaix2006institutional} which analyze high frequency trading data in financial markets, document that the price impact follows a power law. On the other hand, estimates of price impact using low frequency trade data indicate that the price impact function is linear, see for example \citet*{plerou2002quantifying}. 

\par In the present work, we limit attention to a linear permanent price impact function. The choice of a linear price impact function seems natural on several accounts. First, assuming a linear price impact function allows us to abstract from consideration of optimal hedging frequency for a block order (\citet*{gatheral2012transient}). Further, in a continuous time framework such as ours, it suffices to consider a linear price impact function so long as the hypothesis only applies to the leading order following \cite{loeper2018option}.

\par The principal appeal of a linear price impact function however follows from \cite{huberman2004price}, and \cite{gatheral2010no}, in which it is established that such a price impact function guarantees the absence of profitable price manipulation for a large investor who is cognizant of her price impact. While suggestive, these results are established either in a discrete-time framework or in the context of a continuous-time liquidation problem with a single large investor. Nevertheless, we show that in the portfolio choice problem we consider, linearity of the price impact function is necessary and sufficient to rule out price manipulation by institutional investors. To this end, we recall below the definition of a dynamic arbitrage opportunity. In subsequent discussion, given $i \in \mathcal{I}$, we follow game-theoretic convention and let $-i \in \mathcal{I}\backslash\{i\}$ denote $i$'s opponent.

\begin{definition}
\par Fix $\tau$ such that $0 < \tau \leq T$, consider an adapted trading rate process $\left\{x^{-i}_{t}\right\}_{t\, \in\, [0,\,\tau]}$ for investor $-i$, and recall that $W^{i}_{\tau}$ denotes the portfolio value of investor $i$ at time $\tau$. The trading rate process $\left\{x^{i}_{t}\right\}_{t\, \in\, [0,\tau]}$ defines a \textit{dynamic arbitrage} for investor $i$ if we have
\begin{enumerate}[(i)]
\item $\left\{x^{i}_{t}\right\}_{t\, \in\, [0,\tau]}$ is a round-trip trade, that is, $\int^{\tau}_{0} x^{i}_{t}\,dt\, =\, 0$. 
\item $\mathds{E}\left[W^{i}_{\tau}\right]\, >\, W^{i}_{0}\,+\,\mathds{E}\left[\mathlarger{\int}^{\tau}_{0} \pi^{i}_{t}\,\kappa^{-i}\!\left(x^{-i}_{t}\right)dt\right] $
\end{enumerate}
\end{definition}

\par Note from the definition above that should a dynamic arbitrage opportunity exist for investor $i$, and provided the trading horizon is large, she can in principle replicate the dynamic arbitrage infinitely often to generate an almost sure net positive gain\footnote{This is just a consequence of the law of large numbers.} in her portfolio value, over and above the extrinsic change occurring on account of trading by her rival. This observation lends credibility to the definition as, under certain conditions, a dynamic arbitrage opportunity can approximate a classical arbitrage opportunity in the sense of a \textit{free lunch with vanishing risk}, see \cite{delbaen1994general}. 

\par The following lemma provides a characterization of the price impact function $\kappa^{i}$, which is both necessary and sufficient to rule out the existence of profitable price manipulation opportunities for investor $i$ in the market, in the sense of dynamic arbitrage defined above. The proof of the lemma relies on and extends the arguments of \cite{gatheral2010no} and \cite{gueant2016financial} to the case of optimal portfolio choice in a financial market with multiple large investors. For the sake of readability, the proof is relegated to the technical appendix.
\begin{lemma}\label{linpi}
\par The absence of dynamic arbitrage opportunities for investor $i$ is equivalent to the price impact function $\kappa^{i}$ being linear.
\end{lemma}

\par In view of the characterization of the price impact function furnished by the lemma above, which excludes profitable price manipulation (dynamic arbitrage) opportunities, we can rewrite the dynamics for the price process of the risky asset as follows
\begin{equation}\label{price}
	dS_{t}\, =\, \sigma(S_{t})\, dB_{t} \,-\, \left(\sum_{i} \theta^{i} x^{i}_{t} dt\right)
\end{equation}

\par The positive constant $\theta^{i}$ above captures the extent of permanent price impact associated with the trading activity of the $i$th investor. The sign restriction on $\theta^{i}$ is consistent with the observation that substantial buy orders drive up, while substantial sell orders drive down the price of any traded asset. Note that given the separable, additive form of price impact we consider, we can extend the setup and subsequent analysis to the case of $n>2$ large investors in a straightforward manner, but we retain the parsimonious two player setup to streamline the analysis without sacrificing economic import.

\par Subsequently, we shall treat the microstructural foundations of $\theta^{i}$ essentially as a black box. However, one may draw upon existing literature on market microstructure to grasp its underlying determinants. We cite one such example: the seminal work by \cite{shleifer1986large}, where the authors argue that large shareholders exercise effective control and monitoring, thereby reducing managerial rents and consequently improving expected returns. 

\par We suppose that each investor selects a trading strategy $\{x^{i}_{t}\}_{t\,\in\,[0,T]}$ so as to maximize the value of her portfolio at the terminal time $T$, and that the preferences of the $i$th investor over her terminal portfolio value are represented by means of a (Bernoulli) utility function, $u^{i}:\mathds{R} \rightarrow \mathds{R}$, which is twice continuously differentiable (that is, $u^{i}\in \mathds{C}^{2}$), strictly increasing and strictly concave in its argument. To formulate a precise statement of the optimization problem faced by a given institutional investor, we give the definition for the class of admissible controls in our setup;

\begin{definition}\label{admcon}
\par Given an initial time $s$ such that $0\leq s\leq T$, the class of admissible controls $\left\{X =\{x_{t}\}_{\,t\, \in\, [0,T]}\right\}$ with respect to $s$ is denoted by $\mathcal{A}_{s}$, and it consists of controls which satisfy
\begin{enumerate}[(i)]
	\item $x:\Omega\times\left[s,T\right]\rightarrow\,\mathds{R}$, is adapted with respect to $\mathds{F}$
	\item $x_{r} = 0$, for all $r \in [0,s)$
	\item $\mathds{E}\left[\exp{\left\{\,m\!\mathlarger{\int}_{s}^{T}\!\!\!\left\lvert x_{t}\right\rvert dt\mathlarger{\int}_{s}^{T}\!\!\!\left\lvert \hat{x}_{t}\right\rvert dt\right\}}\right] < \infty,\ $ for all $\,m \in \mathds{R},\, $ $\hat{X} = \left\{\hat{x}_{t}\right\}_{t\, \in\, [0,T]} \in \mathcal{A}_{s}$ 
\end{enumerate}
\end{definition}

\par Having defined the class of admissible controls, we characterize the dynamics of the associated state process. For ease of notation, we let $Y^{i}_{t}$ denote the state vector at time $t$, associated with the control problem of the $i$th investor, that is, $\left(Y^{i}_{t}\right)^{\!\mathrm{T}} = \left(S_{t},\pi^{i}_{t},\pi^{-i}_{t},W^{i}_{t},W^{-i}_{t}\right)$. In addition, we introduce and define the vector-valued functions $a^{i}$, $b^{i}$ and $\mathrm{v}^{i}$ as follows
\begin{align*}
    & a^{i}\!\left(Y^{i}_{t}\right)^{\!\mathrm{T}} = \left(-\theta^{-i},\,0,\,-1,\,-\theta^{-i}\pi^{i}_{t},\,-\theta^{-i}\pi^{-i}_{t}\right)\\
    & b^{i}\!\left(Y^{i}_{t}\right)^{\!\mathrm{T}} = \left(-\theta^{i},\,-1,\,0,\,-\theta^{i}\pi^{i}_{t},\,-\theta^{i}\pi^{-i}_{t}\right)\\
    & \mathrm{v}^{i}\!\left(Y^{i}_{t}\right)^{\mathrm{T}} = \sigma\!\left(S_{t}\right)\left(1,\,0,\,0,\,\pi^{i}_{t},\,\pi^{-i}_{t}\right)
\end{align*}

\par The notation introduced above allows a succinct characterization of the dynamics of $Y^{i}_{t}$, in terms of the following stochastic differential equation,
\begin{equation}\label{stated}
	\begin{split}
	dY^{i}_{t}\ =\ a^{i}\!\left(Y^{i}_{t}\right)x^{-i}_{t}dt\, +\, b^{i}\!\left(Y^{i}_{t}\right)x^{i}_{t} dt\, +\, \mathrm{v}^{i}\!\left(Y^{i}_{t}\right)\,dB_{t}
	\end{split}
\end{equation}

\par Given the characterization of the state process dynamics above, the proposition below establishes the existence of a unique, non-explosive strong solution to the system of stochastic differential equations defining the state process dynamics above (\ref{stated}). This proposition does not follow from standard existence results as the coefficients in (\ref{stated}) do not satisfy a Lipschitz type condition in general. The arguments presented in the proof below adapt and extend the ideas employed in \cite{applebaum2009levy} and \cite{lan2014new}, primarily through a localization approach. For the sake of readability, the proof is relegated to the technical appendix.

\begin{proposition}\label{exist1}
\par Let the local volatility function $\sigma$ satisfy {\color{azul-pesc}\Cref{siglip}}, and $X^{i}, X^{-i} \in \mathcal{A}_{s}$. Then, the system of stochastic differential equations (\ref{stated}) has a unique strong solution which is non-explosive, that is, the lifetime of the solution, $\liminf_{\,n\,\rightarrow\,\infty}\,\left\{t>0,\ \left\vert Y^{i}_{t}\right\vert\geq n \right\}\, >\, T$, $\mathds{P}$-almost surely.
\end{proposition}

\par We next turn our attention to defining the best-response value function for the $i$th investor. To this end, given $X^{-i}\,\in\,\mathcal{A}_{s}$, we define the (von-Neumann\textendash Morgenstern) payoff functional $U^{i}$ for the $i$th investor, presuming that the $i$th investor selects a control $X^{i} \in \mathcal{A}_{s}$, as the expected utility of the resulting terminal value of her portfolio,
\begin{equation*}
	U^{i}\!\left(X^{i};X^{-i}\right)\, =\, \mathds{E}\left[u^{i}\!\left(W^{i}_{T}\right)\right]
\end{equation*}

\par Following \citet*[Chapter 4]{dockner2000differential}, we shall focus attention on \textit{Markov\textendash Nash equilibrium} of the Merton\textendash Cournot game defined above. This is justified in view of the fact that in our setup the strategy of an institutional investor affects the payoff of her rival only through its influence on the state vector, particularly on the expected return of the risky asset. Recall that Markovian strategies for the $i$th investor are a class of closed-loop or feedback strategies, in which at each instant $0\leq t\leq T$, the history of the game, encapsulated by $\{Y^{i}_{s}: 0\leq s \leq t\}$, influences current trading rate $x^{i}_{t}$ only through its value at time $Y^{i}_{t}$. The formal definition of the class of Markovian strategies for an institutional investor in our setup is:
\begin{definition}
\par For an initial time $s \in [0,T]$, and $i \in \mathcal{I}$, $\mathcal{A}^{m}_{s}$ denotes the class of Markovian controls $\left\{X^{i} =\{x^{i}_{t}\}_{\,t\, \in\, [0,T]}\right\}$ with respect to $s$, which satisfy $X^{i} \in \mathcal{A}_{s}$, and $x^{i}_{t} = x^{i}\!\left(t, Y^{i}_{t}\right)$, for $t \in [0,T]$.
\end{definition}

\par In view of (\ref{stated}), it follows that if her opponent employs a Markovian strategy, then investor $i$ has a best-response which is Markovian\footnote{This is of course subject to the best-response problem of investor $i$ being well-defined in the sense of having a nonempty set of optimal controls.}. See \citet[Chapter 13]{fudenberg1991game} for a detailed discussion. We can therefore define the $i$th investor's optimal best-response problem on the set of admissible strategies, without loss of generality, provided we limit her opponent to Markovian strategies. Thus, the best-response value function $J^{i}$ for the $i$th investor, given a Markovian admissible control $X^{-i} \in \mathcal{A}^{m}_{s}$ for investor $-i$, where $i,-i \in \mathcal{I}$ and $i \neq -i$, is then defined to be the supremum of her payoff functional over the set of all admissible controls $X^{i} \in \mathcal{A}_{s}$,
\begin{equation*}
	J^{i}\!\left(s,y^{i};X^{-i}\right)\ =\ \sup_{X^{i}\,\in\,\mathcal{A}_{s}}U^{i}\!\left(X^{i};X^{-i}\right)
\end{equation*}

\par It should be noted that the class of closed-loop strategies we consider is fundamentally different from \citet*{micheli2021closed}, wherein an investor's trading rate is a function of the trading rate of other investors. However, unlike their framework which focuses on temporary price impact, we consider a setup with permanent price impact where the consideration of Markovian strategies is fitting. We conclude this section with a concise statement concerning the definition of the solution to the strategic investment stochastic differential Merton\textendash Cournot game outlined above
\begin{definition}\label{mneq}
A pair of control processes $(\hat{X}^{1}\!,\,\hat{X}^{2})$ is said to be a Markov\textendash Nash equilibrium for the game if for a given initial time $s$, we have $\hat{X}^{1},\,\hat{X}^{2} \in \mathcal{A}^{m}_{s}$ and 
	\begin{equation*}
		J^{i}\big(s,y^{i};\hat{X}^{-i}\big)\ =\ U^{i}\big(\hat{X}^{i};\,\hat{X}^{-i}\big),\quad \text{where}\ i\neq -i\ \text{and}\ i,-i \in \mathcal{I}
	\end{equation*}
\end{definition}

\par The definition above is weaker than the one employed in \citet*{chen2022learning}, where the authors focus on \textit{Subgame perfect\footnote{Informally, a subgame perfect equilibrium of a stochastic differential game is a strategy tuple which constitutes a Nash equilibrium of every proper subgame, and thus is in particular, a Nash equilibrium of the stochastic differential game itself.}(Nash) equilibrium}, wherein off-equilibrium strategies of other players may not necessarily be Markovian. The consideration of subgame perfection on their part is facilitated by the assumption of a tractable price impact framework, as in \cite{cuoco1998optimal}, and a linear quadratic utility function, which circumvents the dynamic programming approach by employing point-wise quadratic maximization instead. In contrast, by focusing on a Markov\textendash Nash equilibrium we can be more permissive regarding the class of investor preferences by allowing for (heterogeneous) risk-aversion among investors, while considering a price impact framework with greater degree of verisimilitude.  

\section{Diffeomorphic Integral Flow}\label{intflo}

\par We begin this section with a statement concerning notation we employ subsequently. The first derivative of a function with respect to its $m$th argument will be denoted by $D_{m}$. Additionally, we make use of $D^{n}_{m}$ to denote the $n$th partial derivative of a function with respect to its $m$th argument, for $n>1$. 

\par Returning to the Merton\textendash Cournot stochastic differential game described above, if $x^{i}_{t}$ were locally bounded, one could consider appealing to the notion of viscosity solutions \citep*{crandall1983viscosity} to obtain a characterization for the best-response value function of investor $i$, in terms of the viscosity solution of the following Hamilton\textendash Jacobi\textendash Bellman equation
\begin{multline}\label{hjbvisc}
	 D_{1}\,J^{i}\!\left(s,y^{i};X^{-i}\right)  \,+\, \left\{\sup\limits_{x^{i}} \, \left\langle D_{2}\,J^{i}\!\left(s,y^{i};X^{-i}\right),\ a^{i}(Y^{i})\,x^{-i} +\, b^{i}(Y^{i})\,x^{i} \right\rangle \right\}\\
	 +\, \frac{1}{2}\,\textrm{tr}\left(\,D^{2}_{2}\,J^{i}\!\left(s,y^{i};X^{-i}\right)\mathrm{v}^{i}(Y^{i})\mathrm{v}^{i}(Y^{i})^{\prime}\right)\, =\, 0, \ \ \text{with}\ \ J^{i}\vert_{\,t\,=\,T}\, =\, u^{i}\!\left(W^{i}_{T}\right)
\end{multline}

\par Note that in order for the above characterization of the best-response value function to be meaningful, the supremum with respect to $x^{i}$ must be well-defined. However, in our setup we wish to focus on liquidity frictions arising solely on account of price impact, and hence assuming \textit{a priori} that the trading rate of an investor is (locally) bounded seems overly restrictive\footnote{Interested readers are directed to \cite{longstaff2001optimal}, in which the author considers a setup where the trading rate is required to be bounded on account of liquidity constraints arising out of borrowing and short-sales constraints.}. In our setup, we work under the general premise that $x^{i}_{t} \in \mathds{R}$. As a result of this, the term associated with the supremum and hence the left hand side in the equation above may fail to be finite, yielding a singularity which renders the dynamic programming approach infeasible. 

\par It turns out that the best-response problem of investor $i$ in our model is related to the class of singular stochastic optimal control problems considered in \cite{lasry2000classe}, which suggests adapting and extending their approach to the Merton\textendash Cournot stochastic differential game. We provide here a heuristic argument along the lines of \cite{goldys2019class}, which extends the Lasry\textendash Lions approach to the case of jump-diffusions driven by L\'{e}vy processes, to illustrate the mechanics of this approach and its potential relevance to optimal portfolio choice in dynamic thin markets. To this end, note that in order to ensure that the above characterization (\ref{hjbvisc}) of the best-response value function remains valid we require the term corresponding to the controlled drift coefficient be zero, that is, the following condition must hold
\begin{equation}\label{finite}
\left\langle\, D_{2}\,J^{i}\!\left(s,y^{i};X^{-i}\right),\  b^{i}(Y^{i})\,\right\rangle\, =\, 0
\end{equation}

\par In case the above condition holds, then the supremum with respect to $x^{i}$ turns out to be well-defined in the Hamilton\textendash Jacobi\textendash Bellman equation above (\ref{hjbvisc}). This observation underlies the approach considered in \cite{lasry2000classe}, where the authors first derive an integral flow associated with the controlled drift coefficient of the state vector of the singular control problem. In a similar vein, we introduce the integral flow $\phi^{i}(q,y^{i}) \in \mathds{R}^{5}$, which is parametrized by the scalar $q \in \mathds{R}$, and is derived from the drift function $b^{i}(\cdot)$ by way of the following ordinary differential equation
\begin{equation}\label{flow}
\frac{\partial\phi^{i}(q)}{\partial q}\, =\, b^{i}\!\left(\phi^{i}\!\left(q\right)\right),\ \ \text{with}\ \ \phi^{i}\!\left(0,y^{i}\right)\, =\, y^{i}
\end{equation}

\par If one were to show, as indeed we aim to, that the best-response value function is invariant with respect to the integral flow computed above, that is, $J^{i}\!\left(s,\phi^{i}(q,y^{i});X^{-i}\right) = J^{i}\!\left(s,y^{i};X^{-i}\right)$, then it follows that by differentiating both sides of the invariance equation with respect to $q$ and invoking (\ref{flow}), we recover (\ref{finite}), which is precisely the condition required for the best-response problem of an investor in the Merton\textendash Cournot stochastic differential game to be non-trivial

\par For the class of singular control problems which satisfy an invariance property with respect to the integral flow defined in terms of the controlled drift coefficient of the state vector, the Lasry\textendash Lions approach constructs an auxiliary stochastic control problem, on the quotient space induced by orbits of this flow, which is equivalent to the original singular control problem in the sense that their respective value functions are equal\footnote{The careful reader will note the liberty we have taken in this statement, since the value functions are not defined on dimensionally equivalent domains. However, as we show later (see the proof of \textcolor{azul-pesc}{\Cref{vinv}}), this claim can be made rigorous with little effort.}. Further, the constructed auxiliary control problem is tractable by standard methods which allows us to solve the original singular best-response problem. 

\par As a coda to this section, we present the following proposition which establishes regularity properties for the integral flow defined above, and which we shall employ in proving the well-posedness of the auxiliary control problem defined in the next section. For the sake of readability, the proof is relegated to the technical appendix.

\begin{proposition}\label{reg}
	Given the vector-valued differential equation (\ref{flow}) which defines the integral flow $\phi^{i}(q,y^{i})$ for the singular best-response optimal control problem, we have  
	\begin{enumerate}[(i)]
		\item $\phi^{i}(q,y^{i})$ exists and is unique.
		\item The function $\varepsilon^{i}_{q}(\cdot\,) = \phi^{i}(q,\cdot\,)$ is twice continuously differentiable.
	\end{enumerate}
\end{proposition}

\section{Auxiliary Control Problem}\label{auxcp}
\par This section deals with the construction of an auxiliary stochastic control problem with the help of the integral flow defined in the previous section. The constructed auxiliary control problem will be shown to be equivalent to the best-response problem of an investor, in the sense that their associated value functions will coincide. It will also follow that the auxiliary control problem is tractable by standard dynamic programming methods by construction, which will in turn allow us to solve for the value function of the singular best-response stochastic control problem.

\par The principal notion guiding the construction of the auxiliary control problem is that one can employ the diffeomorphic flow defined in the preceding section to transport a given state to any desired value, instantaneously and frictionlessly. Moreover, a judicious choice of the flow parameter in conjunction with It\^{o}'s Lemma leads to the elimination of the terms underlying the singular nature of the best-response problem of an investor from the dynamics of the transformed state process. The only remaining step is to determine an appropriate transformation of the utility function so as to ensure the desired equivalence. 

\par We begin our description of the auxiliary control problem with the derivation of the auxiliary state process. To this end, we first derive the abridged process $\left\{\Gamma^{i}_{t}\right\}$, linked to the state process of the best-response problem of an investor (\ref{stated}), in three steps outlined below. 

\par First, given a fixed initial time $s \in [0,T]$, we substitute the zero control, $x^{i}_{t} = 0, \forall \ t \in \left(s,T\right]$, in the dynamics of the state process of the best-response problem of the $i$th investor (\ref{stated}). That is, for $s \leq t \leq T$, we consider the process $\left\{Y^{i,\,0}\right\}$, where the zero in the superscript indicates the zero control. From the discussion above, it follows that $Y^{i,\,0}_{s}\, =\, Y^{i}_{s}$, while for $t \in (s,T]$, the dynamics of $\left\{Y^{i,\,0}\right\}$ are governed by the following stochastic differential equation  
\begin{equation*}
     dY^{i,\,0}_{t}\, =\, a^{i}(Y^{i}_{t})\,x^{-i}_{t}dt\, +\, \mathrm{v}^{i}(Y^{i}_{t})\,dB_{t}
\end{equation*}

\par Note that as a consequence of the preceding step, the terms corresponding to $b^{i}$ in the dynamics of the state process $Y^{i}$ are eliminated. Next, we define the stochastic process $\left\{\Gamma^{i}\right\}$ as an abridgement of $\left\{Y^{i,\,0}\right\}$ by dropping $\pi^{i}_{t}$ as a state variable, that is, we have $\Gamma^{i}_{t} = Y^{i,\,0}_{t}\backslash\,\pi^{i,\,0}_{t} = (S^{0}_{t},\pi^{-i,\,0}_{t},W^{i,\,0}_{t},W^{-i,\,0}_{t})$ for $t \in \left(s, T\right]$. Additionally, note that the control variable for the auxiliary control problem is taken to be the degenerate/free state variable $\pi^{i}_{t}$.

\par Finally, in the last step, we note that in order to complete the specification of the dynamics of $\left\{\Gamma^{i}\right\}$, it remains to define the initial value $\Gamma^{i}_{s}$, which we affix as $\Gamma^{i}_{s} = \phi^{i}(-q,Y^{i,0}_{s})\backslash\, \phi^{i}(-q,\pi^{i,0}_{s}) = \phi^{i}(-q,Y^{i}_{s})\backslash\, \phi^{i}(-q,\pi^{i}_{s})$, where $q \in \mathds{R}$. That is, the abridged process $\Gamma^{i}$ corresponds to the truncated zero-control process $Y^{i,0}\backslash\,\pi^{i,0}$, modulo a translation of its initial value $Y^{i,0}_{s}\backslash\,{\pi^{i,0}_{s}}$ by the flow value $\phi^{i}(-q,Y^{i,0}_{s})\backslash\, \phi^{i}(-q,\pi^{i,0}_{s})$, for some fixed choice of scalar $q \in \mathds{R}$. 

\par Next, we define below two vector-valued functions $\hat{a}$ and $\hat{\mathrm{v}}$, in order to succinctly characterize the dynamics of the stochastic process $\left\{\Gamma^{i}\right\}$
\begin{equation*}
\setlength{\jot}{5pt}
     \begin{aligned}
		 & \hat{a}^{i}(\pi^{i}_{t},\Gamma^{i}_{t})^{\mathrm{T}} = \left(-\theta^{-i},-1, -\theta^{-i}\pi^{i}_{t},-\theta^{-i}\pi^{-i}_{t}\right) \\
		 & \hat{\mathrm{v}}^{i}(\pi^{i}_{t},\Gamma^{i}_{t})^{\mathrm{T}} = \sigma\!\left(S_{t}\right)\left(1,0,\pi^{i}_{t},\pi^{-i}_{t}\right)
	\end{aligned}
\end{equation*}

\par Thus, given $0\leq s\leq t\leq T$, the apparatus defined above leads us to the following stochastic differential equation governing the dynamics of the process $\left\{\Gamma^{i}\right\}$ 
\begin{equation}\label{abstate}
\setlength{\jot}{5pt}
   \begin{aligned}
     d\Gamma^{i}_{t}\, =\, &\, \hat{a}^{i}(\pi^{i}_{t},\Gamma^{i}_{t})\,x^{-i}_{t}dt +\, \hat{\mathrm{v}}^{i}(\pi^{i}_{t},\Gamma^{i}_{t})\,dB_{t}\\
     \Gamma^{i}_{s}\, =\, &\, \phi^{i}(-q,Y^{i}_{s})\backslash\, \phi^{i}(-q,\pi^{i}_{s})
   \end{aligned}
\end{equation}

\par Furthering our construction of the auxiliary control problem, we define the set of admissible control processes. To this end, we affix a compact set $\Delta \in \mathds{R}$, such that $0 \in \Delta$. Then, given an initial time $s \in [0,T]$, we can define the set of admissible control processes for the auxiliary control problem, denoted by $\mathcal{A}^{a}_{s}$, as follows
\begin{definition}\label{auxadmcon}
\par For a given an initial time $s \in \left[0,T\right]$, the class of admissible auxiliary controls $\left\{\{\pi_{t}\}_{t \in [0,T]}\right\}$ with respect to $s$, denoted by $\mathcal{A}^{a}_{s}$, satisfies
\begin{enumerate}[(i)]
	\item $\pi:\Omega\times\left[s,\,T\right]\rightarrow\,\Delta$, is adapted with respect to $\mathds{F}$
	\item $\pi\left(\omega,\cdot\,\right):\left[s,\,T\right]\rightarrow\,\Delta$, is left-continuous at $T$, for $\mathds{P}$-almost every $\omega$
	\item $\pi_{r} = 0$, for all $r \in [0,s)$
\end{enumerate}
\end{definition}

\par With the definition of the set of admissible controls and the abridged process $\Gamma^{i}$ at our disposal, we can define the controlled state process as well as the associated state dynamics for the auxiliary control problem. To this effect, Note that the (matrix) differential equation (\ref{flow}) which defines the integral flow $\phi^{i}(q,Y^{i}_{t})$ is autonomous in $q$. We can thus exploit this property to express a solution of the matrix differential equation (\ref{flow}), with an alternate initial condition, in terms of a simple $q$-shift of $\phi^{i}$ itself, see \citet[Section 6.2]{teschl2012ordinary} for details.

\par This observation in conjunction with the discussion above then implies that if we define the candidate auxiliary state process $\left\{Z^{i}\right\}$ by the relation $Z^{i}_{t} = \phi^{i}(q,\Gamma^{i}_{t})$, we can re-write it equivalently, by way of the $q$-shift property, as $\Gamma^{i}_{t}\, =\, \phi^{i}(-q, Z^{i}_{t})$, where we solve the matrix differential equation (\ref{flow}) through routine computation to obtain $\phi^{i}(q,\Gamma^{i}_{t})$ as
\begin{equation*}
\phi^{i}\!\left(q,\Gamma^{i}_{t}\right) = \phi^{i}\!\left(q,Y^{i}_{t}\right)\backslash\,\phi^{i}\!\left(q,\pi^{i}_{t}\right) = \left(S-\theta^{i}q, \pi^{-i}, W^{i}-\theta^{i}\pi^{i}q + \frac{\theta^{i}}{2}q^{2}, W^{-i}-\theta^{i}\pi^{-i}q\right)^{\mathrm{T}}
\end{equation*}

\par In order to define a suitable controlled state process for the auxiliary control problem, we next eliminate the abridged process $\left\{\Gamma^{i}\right\}$ in the relation defining $\left\{Z^{i}\right\}$ above, and introduce the control process $\left\{\pi^{i}\right\}$ in its dynamics. This final measure is achieved by exploiting the equivalence $\Gamma^{i}_{t} = \phi^{i}(-q, Z^{i}_{t})$, and by replacing the parameter $q$ with the $\mathds{F}$-adapted process $\left\{\pi^{i}\right\}$. For brevity we introduce the following notation which we use subsequently
\begin{equation}\label{auxnot}
P^{i} = S - \theta^{i}\pi^{i},\ \ \mathrm{w}^{i} = W^{i}-\frac{\theta^{i}}{2}\!\left(\pi^{i}\right)^{2},\ \ \mathrm{w}^{-i} = W^{-i}-\theta^{i}\pi^{-i}\pi^{i}
\end{equation}

\par Thus, we can re-write the controlled state process for the auxiliary control problem as $Z^{i}\, =\, \left(P^{i}, \pi^{-i}, \mathrm{w}^{i}, \mathrm{w}^{-i}\right)$. Strictly speaking, the auxiliary state variables defined above do not have an associated economic interpretation. One may, however, think of the state variable $P^{i}$ as describing the \textit{auxiliary price} faced by the $i$th investor, while $\mathrm{w}^{i}$, $\mathrm{w}^{-i}$ can be interpreted as the \textit{auxiliary portfolio values} of the two investors.

\par Next, we derive the dynamics of the controlled auxiliary state process $\left\{Z^{i}\right\}$. To this end, we consider the following stochastic differential equation which follows from a routine application of It\^{o}'s Lemma to the function $\phi^{i}(q,\Gamma^{i}_{t})$, where the applicability of It\^{o}'s Lemma is ensured by {\color{azul-pesc}\Cref{reg}} and routine computation shows that $D^{2}_{2}\phi^{i}\!\left(q,\Gamma^{i}_{t}\right) = \mathbf{0}_{4\times 4}$; 
\begin{equation}\label{auxflow}
	\begin{split}
		 d\phi^{i}(q,\Gamma^{i}_{t}) = &\,D_{2}\phi^{i}\!\left(q,\Gamma^{i}_{t}\right)\hat{a}^{i}\!\left(\pi^{i}_{t},\Gamma^{i}_{t}\right)x^{-i}_{t} dt + D_{2}\phi^{i}\!\left(q, \Gamma^{i}_{t}\right)\mathrm{\hat{v}}^{i}\!\left(\pi^{i}_{t},\Gamma^{i}_{t}\right)dB_{t},\quad s< t\leq T\\
    \end{split}
\end{equation}

\par In view of (\ref{auxnot}) and (\ref{auxflow}), the dynamics of the state process for the auxiliary control problem are then governed by the following system of stochastic differential equations for $s < t \leq T$
\begin{equation*}
\setlength{\jot}{5pt}
\begin{aligned}
& dP^{i}_{t}\,=\, -\theta^{-i} x^{-i}_{t} dt\,+\, \sigma\!\left(P^{i}_{t}+\theta^{i}\pi^{i}_{t}\right)dB_{t}\\ 
& d\pi^{-i}_{t}\, =\, -x^{-i}_{t} dt \\
& d\mathrm{w}^{i}_{t}\,=\, \pi^{i}_{t}\left(-\theta^{-i}x^{-i}_{t} dt\,+\, \sigma\!\left(P^{i}_{t}+\theta^{i}\pi^{i}_{t}\right)dB_{t}\right)\\
& d\mathrm{w}^{-i}_{t}\,=\,\pi^{-i}_{t}\left(-\theta^{-i}x^{-i}_{t} dt\,+\, \sigma\!\left(P^{i}_{t}+\theta^{i}\pi^{i}_{t}\right)dB_{t}\right) \,+\,\theta^{i}\pi^{i}_{t}\,x^{-i}_{t} dt
\end{aligned}
\end{equation*}

\par In order to derive the initial condition, we recall the $q$-shift property of the integral flow $\phi^{i}$ in conjunction with the fact that $Z^{i}_{s} = \phi^{i}(q, \Gamma^{i}_{s})$ and $\Gamma^{i}_{s} = \phi^{i}(-q, Z^{i}_{s})$ to obtain $Z^{i}_{s} = Y^{i}_{s}\backslash \pi^{i}_{s}$. Next, we define vector-valued deterministic functions $\beta^{i}$ and $\nu^{i}$ as follows
\begin{equation}\label{auxceq}
\setlength{\jot}{5pt}
\begin{aligned}
    \beta^{i}\!\left(\pi^{i}_{t},Z^{i}_{t}\right)^{\mathrm{T}} & = \left(-\theta^{-i},-1,-\theta^{-i}\pi^{i}_{t},\theta^{i}\pi^{i}_{t}-\theta^{-i}\pi^{-i}_{t}\right) =  \beta^{i}\!\left(\pi^{i}_{t},\phi^{i}(q, \Gamma^{i}_{t})\right)^{\mathrm{T}}\\ & = D_{2}\phi^{i}\!\left(\pi^{i}_{t},\Gamma^{i}_{t}\right)\hat{a}^{i}\!\left(\pi^{i}_{t},\Gamma^{i}_{t}\right)\\
    \nu^{i}(\pi^{i}_{t},Z^{i}_{t})^{\mathrm{T}} & =  \sigma\!\left(P^{i}_{t}+\theta^{i}\pi^{i}_{t}\right)\left(1,0,\pi^{i}_{t},\pi^{-i}_{t}\right) = \sigma\!\left(S_{t}\right)\left(1,0,\pi^{i}_{t},\pi^{-i}_{t}\right) = \nu^{i}(\pi^{i}_{t},\phi^{i}(q, \Gamma^{i}_{t}))^{\mathrm{T}} \\ & = D_{2}\phi^{i}\!\left(\pi^{i}_{t},\Gamma^{i}_{t}\right)\hat{v}^{i}\!\left(\pi^{i}_{t},\Gamma^{i}_{t}\right)
\end{aligned}
\end{equation}

\par We use the definition of the vector-valued functions $\beta^{i}$ and $\gamma^{i}$ to re-write the dynamics of the auxiliary state process $Z^{i}$ succinctly as follows
\begin{equation}\label{astate}
\setlength{\jot}{5pt}
	\begin{aligned}
		dZ^{i}_{t} =\, &\ \beta^{i}\!\left(\pi^{i}_{t}, Z^{i}_{t}\right)x^{-i}_{t}dt\, +\, \nu^{i}\!\left(\pi^{i}_{t}, Z^{i}_{t}\right)dB_{t}\\
		Z^{i}_{s} =\, &\ Y^{i}_{s},\quad 0\leq s\leq t \leq T
	\end{aligned}
\end{equation}

\par Having defined a suitable controlled state process for the auxiliary control problem, we turn our attention to the question of well-posedness of the system of stochastic differential equation (\ref{astate}) governing its dynamics. This is not immediate from standard existence results since the coefficients in the system of stochastic differential equations (\ref{astate}) do not satisfy a Lipschitz hypothesis in general. We instead extend the Euler approximation method employed in \cite{applebaum2009levy} and \cite{lan2014new} as earlier in order to prove the following proposition, the proof of which is relegated to the technical appendix for the sake of readability.

\begin{proposition}\label{ext} \par Suppose the local volatility function $\sigma$ satisfies {\color{azul-pesc}\Cref{siglip}}. Then, the system of stochastic differential equations defined by (\ref{astate}) has a unique strong solution for every $\pi^{i} \in \mathcal{A}^{a}_{s}$. Moreover, the solution is non-explosive, that is, the lifetime of the solution, $\inf\left\{t>0,\ \left\vert Z^{i}_{t}\right\vert\geq n \right\} > T$ as $n \rightarrow \infty$, almost surely.
\end{proposition}

\par We conclude the description of the auxiliary control problem by defining the auxiliary best-response value function for the $i$th large investor. To this end, given a fixed initial time $s \in [0,T)$, a deterministic initial auxiliary state $z^{i} \in \mathds{R}^{\lvert\,Z^{i}\,\rvert}$, and an admissible control $X^{-i} \in \mathcal{A}_{s}$ for investor $-i$, the auxiliary best-response value function of the $i$th investor, denoted by $V^{i}$, is defined as
\begin{equation}\label{auxvf}
\setlength{\jot}{7pt}
\begin{aligned}
   & V^{i}\!\left(s,z^{i};X^{-i}\right)\, = \sup\limits_{\pi^{i}\,\in\, \mathcal{A}^{a}_{s}}\mathds{E}\left[\,\sup\limits_{q\, \in\, \mathds{R}}\Big\{\,u^{i}\circ\phi^{i}\!\left(-q,Z^{i}_{T}\right)\!\Big\}\right] = \sup\limits_{\pi^{i}\,\in\, \mathcal{A}^{a}_{s}}\mathds{E}\left[u^{i}\!\left(\mathrm{w}^{i}_{T}\right)\right] \overset{\mathrm{def}}{=}\! \sup\limits_{\pi^{i}\,\in\,\mathcal{A}^{a}_{s}} H^{i}\!\left(\pi^{i};X^{-i}\right) \\
   & \quad \text{with}\ \ V^{i}\!\left(s,z^{i};X^{-i}\right)\!\big\vert_{\,s\,=\,T}\, = \sup\limits_{q\, \in\, \mathds{R}}\Big\{\,u^{i}\circ\phi^{i}\!\left(-q,Z^{i}_{T}\right)\!\Big\}
\end{aligned}    
\end{equation}

\section{Value Function Equivalence}\label{vequiv}
\par This section establishes the equivalence result which is central to the present work using a three step approach as in \cite{lasry2000classe}, and \cite{goldys2019class}. First, we establish an invariance result for the value function of the singular best-response control problem. In particular, we prove that the best-response value function of investor $i$ is invariant with respect to the diffeomorphic integral flow defined in {\color{azul-pesc}\Cref{intflo}}. 

\par Next, we show that a corresponding invariance result holds for the value function of the auxiliary optimal control problem, and finally, we establish the equivalence of the value functions of the singular and the auxiliary optimal control problems with the aid of the two invariance results. In subsequent discussion, we maintain the standing assumption that the utility function representing the preferences of the investors over the terminal portfolio value belong to the CARA or exponential class of utility functions. Specifically, we assume
\begin{equation*}
u^{i}\!\left(W^{i}_{T}\right)\, =\, -\,\frac{1}{\delta^{i}}\exp\!{\left(-\delta^{i}W^{i}_{T}\right)},\ \delta^{i}>0
\end{equation*}

\par The positive scalar $\delta^{i}$ above denotes the (constant) coefficient of absolute risk-aversion of the $i$th investor. Note that we permit heterogeneity in risk-aversion amongst institutional investors. It is straightforward to check that $u^{i}\in\mathds{C}^{2}$ for $i\in\mathcal{I}$, which implies that the exponential functional form assumed above for the utility function is compatible with the fundamental desiderata for the utility functions of investors outlined earlier in \textcolor{azul-pesc}{\Cref{model}}. 

\subsection{Singular Problem Invariance}

\par In order to prove the desired invariance result for the value function of the singular best-response control problem we first define the truncated best-response value function $J^{i,n}$, that is, given an initial time $s \in [0,T]$, an admissible strategy $X^{-i} \in \mathcal{A}_{s}$ for investor $-i$ and $n \in \mathds{N}$, we define 
\[{J}^{i,n}\!\left(s,y^{i};X^{-i}\right)\,=\,\sup_{X^{i}\,\in\,\mathcal{A}^{n}_{s}}\,U^{i}\!\left(X^{i},X^{-i}\right)\]

\par Evidently, the constraint space $\mathcal{A}^{n}_{s}$ in the definition above represents the subclass of the class of admissible controls that are bounded by $n$, and is formally defined as follows
\[\mathcal{A}^{n}_{s}\,=\,\big\{X \in \mathcal{A}_{s}: \left\lvert x_{t}\right\rvert\leq n,\, \text{for all}\ s\leq t\leq T\,\big\}\]

\par The following ancillary lemma collects certain regularity properties of the best-response value function $J^{i}$ as well as the truncated best-response value function $J^{i,n}$. Note that since $J^{i}$ corresponds to the value associated with a singular control problem, the fact that it is non-degenerate and has desirable continuity properties is not immediate. Thus, we call upon \textcolor{azul-pesc}{\Cref{admcon}}, as well as the convergence and continuity properties of elements of the sequence $\{J^{i,n}\}_{n \,\in \,\mathds{N}}$ defined above, to prove the following lemma, the proof of which is relegated to the technical appendix for the sake of readability.
\begin{lemma}\label{jconv}
\par Given an initial time $s\in[0,T]$, an admissible strategy $X^{-i} \in \mathcal{A}_{s}$ for investor $-i$, and a deterministic initial state $y^{i}$ in the domain of $Y^{i}$, we have
\begin{enumerate}[(i)]
\item The function $J^{i}\!\left(s,\,\cdot\,;\,X^{-i}\right):\mathds{R}^{\lvert\,Y^{i}\,\rvert}\,\rightarrow\,\mathds{R}$ is non-degenerate, that is, 
\[\left\lvert\,J^{i}\!\left(s,y^{i};\,X^{-i}\right)\right\rvert < \infty,\, \forall\, y^{i} \in \mathds{R}^{\lvert\,Y^{i}\,\rvert}\]
\item The sequence $\left\{J^{i,n}\!\left(s,y^{i};\,X^{-i}\right)\right\}_{n\,\in\, \mathds{N}}$ converges to $J^{i}\!\left(s,y^{i};\,X^{-i}\right)$, that is, 
\[\lim\limits_{n\,\rightarrow\,\infty}\,J^{i,n}\!\left(s,\,y^{i};\,X^{-i}\right)\, =\, J^{i}\!\left(s,\,y^{i};\,X^{-i}\right)\]
\item The function $J^{i,n}\!\left(s,\,\cdot\,;\,X^{-i}\right) : \mathds{R}^{\lvert\, Y^{i}\,\rvert}\,\rightarrow\, \mathds{R}$ is continuous.
\item The function $J^{i}\!\left(s,\,\cdot\,;\,X^{-i}\right) : \mathds{R}^{\lvert\, Y^{i}\,\rvert}\,\rightarrow\, \mathds{R}$ is lower semi-continuous.
\item The function $J^{i}\!\left(\,\cdot\,,y^{i};\,X^{-i}\right) : [0,T]\,\rightarrow\, \mathds{R}$ is lower semi-continuous at $T$.
\end{enumerate}
\end{lemma}

\par Next, we establish an equicontinuity like property for the controlled state process $Y^{i}$ in the following lemma, which is called upon in proving the desired invariance result. To this end, given a fixed initial time $s \in [0,T]$, we consider a decreasing sequence of deterministic times $\{\epsilon_{n}\}_{n\,\in\,\mathds{N}} \subseteq (s,T]$, such that $\epsilon_{n}\downarrow s$. Moreover, for a given a scalar $q \in \mathds{R}$, we define a control process $X^{\epsilon_{n}}=\left\{x^{\epsilon_{n}}_{t}\right\}_{t\,\in\,[s,\,T]}$, corresponding to the element $\epsilon_{n}$ of the sequence above, as follows
\begin{equation}\label{uep}
    x^{\epsilon_{n}}_{t}\ =\ \begin{cases} q/\left(\epsilon_{n}-s\right),\,& t\, \in\, \left[s,\epsilon_{n}\right]\\
    0,\,& t \in \left[0,s\right)\,\bigcup\,\left(\epsilon_{n},T\right]
    \end{cases}  
\end{equation}

\par It follows from the definition above that for all $k\in\mathds{N}$, such that $k\geq\left\lvert \,q/(\epsilon_{n}-s)\,\right\rvert$, we have $X^{\epsilon_{n}}\in\mathcal{A}^{k}_{s}$. For a given $X^{-i}\,\in\,\mathcal{A}_{s}$, we then consider the sequence of processes $\left\{Y^{i,\,\epsilon_n}\right\}_{n\,\in\,\mathds{N}}$ where $Y^{i,\,\epsilon_n}$ is defined as the solution of the system of stochastic differential equations (\ref{stated}), with initial condition $Y^{i,\,\epsilon_n}_{s}\,=\,y^{i} \in \mathds{R}^{\,\left\lvert\, Y^{i}\,\right\rvert}$, and $X^{i}\,=\, X^{\epsilon_n}$. With the help of the apparatus introduced here, we state and prove the following ancillary lemma

\begin{lemma}\label{flowconv}
\par Fix an $n\in\mathds{N}$, and consider the corresponding random variable $Y^{i,\,\epsilon_{n}}_{\epsilon_{n}}$. The random variable $Y^{i,\epsilon_n}_{\epsilon_{n}}$ converges to $\phi^{i}\!\left(q,y^{i}\right)$ $\mathds{P}$-almost surely as $\epsilon_{n}\downarrow s$.
\end{lemma}

\begin{proof}
\par First, we establish that for a given $q\in\mathds{R}$, the matrix-valued functions $D_{2}\phi^{i}\!\left(q,y\right)$ and $D^{2}_{2}\phi^{i}\!\left(q,y\right)$ are bounded with respect to the Hilbert\textendash Schmidt norm, where $y$ takes values in $\mathds{R}^{\left\vert\, Y^{i}\right\vert}$. To this end, recall from {\textcolor{azul-pesc}{\Cref{reg}}} that $\phi^{i}\!\left(q,y\right)$ is continuously differentiable with respect to its second argument. Also, from (\ref{lipphi}) it follows that for a given $q\in\mathds{R}$, $\phi^{i}\!\left(q,\cdot\,\right)$ is Lipschitz continuous. It is then immediate from the preceding statements that $D_{2}\phi^{i}\!\left(q,y\right)$ is bounded in the Hilbert\textendash Schmidt norm.

\par Further, recall that $D_{2}\phi^{i}\!\left(q,y\right)$ is defined to be the solution to (\ref{dphi}), where $J_{b^{i}}$ is a matrix of constant scalar functions, and hence trivially Lipschitz continuous. We can then repeat the argument employed in the proof to {\textcolor{azul-pesc}{\Cref{reg} (ii)}} to establish that $D_{2}\phi^{i}\!\left(q,y\right)$ is Lipschitz continuous in its second argument for a given $q\in\mathds{R}$. The claim that $D^{2}_{2}\phi^{i}\!\left(q,y\right)$ is bounded in the Hilbert\textendash Schmidt norm is then immediate in view of continuous differentiability of $D_{2}\phi^{i}\!(q,y)$ with respect to its second argument.

\par Next, we prove the convergence result advanced in the claim. To this end, we consider a fixed $k\in\mathds{N}$, and $X^{-i}\in\mathcal{A}_{s}$ and define the $\mathds{F}$-stopping time $\tau_{k,\,x^{-i}}$ as follows
\begin{equation*}
\tau_{k,\,x^{-i}} \,=\, \inf\left\{t>s:\int^{t}_{s}\!\!\left\lvert\,x^{-i}_{u}\,\right\rvert du > k\right\}
\end{equation*}

\par Additionally, for $t \in [s,\epsilon_{n}]$, we consider the function $\phi^{i}\!\left(-\frac{q}{\left(\epsilon_{n}-s\right)}\left(t-s\right),Y^{i,\,\epsilon_{n}}_{t}\right)$. Then, by way of It\^{o}'s Lemma we obtain the following for some $C>0$
\begin{multline}\label{phiconv}
    \mathds{E}\left[\mathds{1}_{\left\{\epsilon_{n}\,\leq\, \tau_{k,x^{-i}}\right\}}\left\lvert\,\phi^{i}\!\left(-q,Y^{i,\,\epsilon_n}_{\epsilon_n}\right)-y^{i}\,\right\rvert\,\right]\\ 
    \leq\, C\mathds{E}\left[\mathlarger{\int}_{s}^{\epsilon_{n}}\!\!\mathds{1}_{\left\{u\,\leq\, \tau_{k,x^{-i}}\right\}}\left\vert\,\left\langle D_{2}\phi^{i}\!\left(-\frac{q}{\epsilon_{n}-s}\left(u-s\right),Y^{i,\,\epsilon_{n}}_{u}\right),\,a^{i}\!\left(Y^{i,\,\epsilon_{n}}_{u}\right)x^{-i}_{u}\right\rangle\,\right\vert du\right] \\ 
    +\, C\mathds{E}\left[\mathlarger{\int}_{s}^{\epsilon_{n}}\!\!\mathds{1}_{\left\{u\,\leq\, \tau_{k,x^{-i}}\right\}}\left\vert\,\text{tr}\left( D^{2}_{2}\phi^{i}\!\left(-\frac{q}{\epsilon_{n}-s}\left(u-s\right),Y^{i,\,\epsilon_{n}}_{u}\right) \mathrm{v}^{i}\!\left(Y^{i,\,\epsilon_{n}}_{u}\right) \mathrm{v}^{i}\!\left(Y^{i,\,\epsilon_{n}}_{u}\right)^{\!\mathrm{T}}\right)\,\right\vert du\right]\\
    +\, C\mathds{E}\left[\,\sup\limits_{t\,\in\,[s,\,\epsilon_{n}]}\left\lvert\,\mathlarger{\int}_{s}^{t}\!\!\mathds{1}_{\left\{u\,\leq\, \tau_{k,x^{-i}}\right\}}\left\langle D_{2}\phi^{i}\!\left(-\frac{q}{\epsilon_{n}-s}\left(u-s\right),Y^{i,\epsilon_{n}}_{u}\right),\,\mathrm{v}^{i}\!\left(Y^{i,\,\epsilon_{n}}_{u}\right)\right\rangle dB_{u}\,\right\rvert\,\right] \\
    +\,C\mathds{E}\,\Bigg\lvert\,\mathlarger{\int}_{s}^{\epsilon_{n}}\!\!\mathds{1}_{\left\{u\,\leq\, \tau_{k,x^{-i}}\right\}}\frac{q}{\epsilon_{n}-s}\Bigg(\!\left\langle D_{2}\phi^{i}\!\left(-\frac{q}{\epsilon_{n}-s}\left(u-s\right),Y^{i,\,\epsilon_{n}}_{u}\right),\,b^{i}\!\left(Y^{i,\,\epsilon_{n}}_{u}\right)\right\rangle\\-D_{1}\phi^{i}\!\left(-\frac{q}{\epsilon_{n}-s}\left(u-s\right),Y^{i,\epsilon_{n}}_{u}\right)\!\Bigg)\,du\,\Bigg\rvert 
\end{multline}

\par Consider the fourth term on the right-hand side of (\ref{phiconv}), and note that for all $u \in [s,\epsilon_{n}]$, the following equality can be derived from (\ref{flow}) via routine calculus (See \cite[Exercise 6.1.12]{applebaum2009levy}, where the candidate function equals the identity map) 
\begin{equation*}
\left\langle D_{2}\phi^{i}\!\left(-\frac{q}{\epsilon_{n}-s}\left(u-s\right),Y^{i,\,\epsilon_{n}}_{u}\right),\,b^{i}\!\left(Y^{i,\,\epsilon_{n}}_{u}\right)\right\rangle\ =\ D_{1}\phi^{i}\!\left(-\frac{q}{\epsilon_{n}-s}\left(u-s\right),Y^{i,\,\epsilon_{n}}_{u}\right)
\end{equation*}

\par It then follows from the above that the fourth term on the right-hand side of (\ref{phiconv}) is identically zero. Further, note that for $u\in [s,\epsilon_{n}]$ we have by definition 
\begin{equation*}
\left\lvert\,\pi^{i,\,\epsilon_{n}}_{u}\,\right\rvert\,\leq\,\left\lvert\,q\,\right\rvert,\ \mathds{1}_{\left\{u\,\leq\,\tau_{k,\,x^{-i}}\right\}}\left\lvert\,\pi^{-i}_{u}\,\right\rvert\,\leq\,k
\end{equation*} 

\par In view of the above, it then follows that for $u\in[s,T]$, there exists a positive constant $C$ such that we have 
\begin{equation*}
    \mathds{1}_{\left\{u\,\leq\,\tau_{k,\,x^{-i}}\right\}}\Big(\left\lvert\, a^{i}\!\left(Y^{i,\,\epsilon_{n}}_{u}\right)\,\right\rvert+\left\lvert\, \mathrm{v}^{i}\!\left(Y^{i,\,\epsilon_{n}}_{u}\right)\,\right\rvert\Big)\,\leq\,C
\end{equation*}

\par Moreover, as established earlier the functions $D_{2}\phi^{i}\!\left(\,\cdot\,\right)$ and $D^{2}_{2}\phi^{i}\!\left(\,\cdot\,\right)$ are bounded. Thus, we can find a positive constant $C$ such that the following holds
\begin{multline*}
    \mathds{E}\left[\mathlarger{\int}_{s}^{\epsilon_{n}}\!\!\mathds{1}_{\left\{u\,\leq\, \tau_{k,x^{-i}}\right\}}\left\vert\,\left\langle D_{2}\phi^{i}\!\left(-\frac{q}{\epsilon_{n}-s}\left(u-s\right),Y^{i,\,\epsilon_{n}}_{u}\right),\,a^{i}\!\left(Y^{i,\,\epsilon_{n}}_{u}\right)x^{-i}_{u}\right\rangle\,\right\vert du\right]\\ 
    +\,\mathds{E}\left[\mathlarger{\int}_{s}^{\epsilon_{n}}\!\!\mathds{1}_{\left\{u\,\leq\, \tau_{k,x^{-i}}\right\}}\left\vert\,\text{tr}\left( D^{2}_{2}\phi^{i}\!\left(-\frac{q}{\epsilon_{n}-s}\left(u-s\right),Y^{i,\,\epsilon_{n}}_{u}\right) \mathrm{v}^{i}\!\left(Y^{i,\,\epsilon_{n}}_{u}\right) \mathrm{v}^{i}\!\left(Y^{i,\,\epsilon_{n}}_{u}\right)^{\!\mathrm{T}}\right)\,\right\vert du\right]\\
    \leq\,C\mathds{E}\left[\mathlarger{\int}_{s}^{\epsilon_{n}}\!\!\mathds{1}_{\left\{u\,\leq\, \tau_{k,x^{-i}}\right\}}\left\lvert\,x^{-i}_{u}\,\right\vert du\right]+C\left(\epsilon_{n} - s\right) 
\end{multline*}

\par Finally, we consider the third term on the right-hand side of (\ref{phiconv}), and note that we can invoke Burkholder\textendash Davis\textendash Gundy inequality to ascertain the existence of a positive constant $C$ such that
\begin{equation*}
    \mathds{E}\left[\sup\limits_{t\,\in\,[s,\,\epsilon_{n}]}\,\left\lvert\,\mathlarger{\int}_{s}^{t}\!\!\mathds{1}_{\left\{u\,\leq\,\tau_{k,x^{-i}}\right\}}\left\langle D_{2}\phi^{i}\!\left(-\frac{q}{\epsilon_{n}-s}\left(u-s\right),Y^{i,\epsilon_{n}}_{u}\right),\,\mathrm{v}^{i}\!\left(Y^{i,\,\epsilon_{n}}_{u}\right)\right\rangle dB_{u}\, \right\rvert\,\right]\leq\,C\sqrt{\left(\epsilon_{n} - s\right)}
\end{equation*}

\par Further, since $\mathlarger{\int}_{s}^{\epsilon_{n}}\!\mathds{1}_{\left\{u\,\leq\, \tau_{k,x^{-i}}\right\}}\left\lvert\,x^{-i}_{u}\,\right\vert du < k$, we can invoke (reverse) Fatou's Lemma to ascertain 
\begin{equation*}
    \limsup\limits_{\epsilon_{n}\,\downarrow\,s}\,\mathds{E}\left[\mathds{1}_{\left\{\epsilon_{n}\,\leq\, \tau_{k,x^{-i}}\right\}}\left\vert\,\phi^{i}\!\left(-q,Y^{i,\,\epsilon_n}_{\epsilon_n}\right)-y^{i}\,\right\vert\,\right]\leq\,\mathds{E}\left[\limsup\limits_{\epsilon_{n}\,\downarrow\,s}\mathlarger{\int}_{s}^{\epsilon_{n}}\!\!\mathds{1}_{\left\{u\,\leq\, \tau_{k,x^{-i}}\right\}}\left\lvert\,x^{-i}_{u}\,\right\vert du\right]\,=\,0
\end{equation*}

\par As the choice of $k$ above was arbitrary, and given that $X^{-i}\in\mathcal{A}_{s}$, we can find $k\in\mathds{N}$ such that $T\leq\tau_{k,\,x^{-i}}$. Hence, it follows that for $\{\epsilon_{n}\}_{n\,\in\,\mathds{N}} \subseteq (s,T]$, such that $\epsilon_{n}\downarrow s$, we have $\phi^{i}(-q,Y^{i,\,\epsilon_n}_{\epsilon_n})\rightarrow y^{i}$, $\mathds{P}$-almost surely as $\epsilon_{n}\downarrow s$. The claim is then immediate in view of the time shift property of the integral flow $\phi^{i}$.
\end{proof}

\par Next, we again consider a decreasing sequence of deterministic times $\{\epsilon_{n}\}_{n\,\in\,\mathds{N}} \subseteq (s,T]$, such that $\epsilon_{n}\downarrow s$, and let $q\in\mathds{R}$ be as above. As earlier, we define a control process $\hat{X}^{\epsilon_{n}}=\left\{\hat{x}^{\epsilon_{n}}_{t}\right\}$ associated with an element $\epsilon_{n}$ of the sequence as follows
\begin{equation}\label{huep}
     \hat{x}^{\epsilon_{n}}_{t}\ =\ \begin{cases} -q/\left(\epsilon_{n}-s\right),&\ t\,\in\, \left[s,\epsilon_{n}\right]\\ 0,&\ t\,\in\, \left[0,s\right)\,\bigcup\,\left(\epsilon_{n},T\right]\end{cases}
\end{equation}

\par It is immediate from the definition above that for all $k\in\mathds{N}$, such that $k\geq\left\lvert\,-q/\left(\epsilon_{n}-s\right)\,\right\rvert$, we have $\hat{X}^{\epsilon_{n}} \in \mathcal{A}^{k}_{s}$. Next, for a given $X^{-i}\in\mathcal{A}_{s}$, we consider the sequence of processes $\{\hat{Y}^{i,\,\epsilon_n}\}_{n\,\in\,\mathds{N}}$ where $\hat{Y}^{i,\,\epsilon_n}$ is defined as the solution of the system of stochastic differential equations (\ref{stated}), with initial condition $\hat{Y}^{i,\,\epsilon_n}_{s}=\phi^{i}\!\left(q,y^{i}\right)$, where $y^{i}$ is as before, and $\hat{X}^{i}=\hat{X}^{\epsilon_n}$. With the aid of the notation introduced above, we state the following lemma which serves as a corollary to {\textcolor{azul-pesc}{\Cref{flowconv}}}, and can be proved in an identical manner, hence we omit its proof.

\begin{lemma}\label{hflowconv}
\par Fix an $n\in\mathds{N}$, and consider the corresponding random variable $\hat{Y}^{i,\,\epsilon_{n}}_{\epsilon_{n}}$. The random variable $\hat{Y}^{i,\epsilon_n}_{\epsilon_{n}}$ converges to $y^{i}$ $\mathds{P}$-almost surely as $\epsilon_{n}\downarrow s$.
\end{lemma}

\par We are now in a position to state and prove our first fundamental result. The said result is formalized in the following theorem which establishes the invariance of the value function of the original singular best-response stochastic control problem, with respect to the integral flow. Note that unlike \cite{lasry2000classe} and \cite{goldys2019class}, the utility function $u^{i}$ we consider is not bounded and hence our proof employs an alternate approach rooted in equicontinuity arguments.

\begin{theorem}[\textbf{\textsc{Invariance - I}}]\label{inv1}
\par Given an initial time $s \in [0,T]$, a fixed initial state $y^{i} \in \mathds{R}^{\,\left\lvert Y^{i}\right\rvert}$, and a scalar $q \in \mathds{R}$, the best-response value function of the $i$th investor is invariant with respect to the integral flow $\phi^{i}\!\left(q,y^{i}\right)$, that is, we have $J^{i}\!\left(s,y^{i};X^{-i}\right) = J^{i}\!\left(s,\phi^{i}\!\left(q,y^{i}\right);X^{-i}\right)$.
\end{theorem}

\begin{proof}
\par First, we establish that $J^{i}(s,y^{i};X^{-i})\geq J^{i}(s,\phi^{i}(q,y^{i});X^{-i})$. To this end, given $s \in [0,T]$, we consider a decreasing sequence of deterministic times $\{\epsilon_{n}\} \in (s,T]$ such that $\epsilon_{n}\downarrow s$. Further, for each element $\epsilon_{n}$ of the sequence, we define an associated control process $X^{\epsilon_{n}}$ as in (\ref{uep}). Then, for $k \geq \lvert\,q/(\epsilon_{n}-s)\,\rvert$, we have
\begin{equation*}
  J^{i}\!\left(s,y^{i};X^{-i}\right)\,\geq\, J^{i,\,k}\!\left(s,y^{i};X^{-i}\right)
\end{equation*}

\par Note that since $X^{\epsilon_{n}}$ and $X^{-i}$ are admissible, it follows from \cite[D.5, Appendix D]{fleming2006controlled} that for a given $n \in \mathds{N}$ and $s \leq t \leq \epsilon_{n}$, $Y^{i,\epsilon_{n}}_{t}$ converges to $y^{i}$ $\mathds{P}$-almost surely as $t\downarrow s$. This fact in conjunction with \textcolor{azul-pesc}{\Cref{jconv} (iii)} then implies that 
\begin{equation*}
J^{i,\,k}\!\left(s,y^{i};X^{-i}\right)\,\geq\, \limsup\limits_{t\,\downarrow\, s}\,J^{i,\,k}\!\left(s,Y^{i,\epsilon_{n}}_{t};X^{-i}\right)
\end{equation*}

\par Given that the inequality above also holds in particular when $t = \epsilon_{n}$, we can appeal to \textcolor{azul-pesc}{\Cref{jconv} (ii)}, \textcolor{azul-pesc}{\Cref{jconv} (iii)}, and \textcolor{azul-pesc}{\Cref{flowconv}} to obtain the desired relation as follows
\begin{multline*}
J^{i}\!\left(s,y^{i};X^{-i}\right)\,=\, \limsup\limits_{k\,\rightarrow\,\infty}\,J^{i,\,k}\!\left(s,y^{i};X^{-i}\right)\,\geq\,\limsup\limits_{k\,\rightarrow\,\infty}\left(\limsup\limits_{\epsilon_{n}\,\downarrow\, s}\,J^{i,\,k}\!\left(s,Y^{i,\epsilon_{n}}_{\epsilon_{n}};X^{-i}\right)\right)\\ \geq\, \limsup\limits_{k\,\rightarrow\,\infty}\,J^{i,\,k}\!\left(s,\phi^{i}(q,y^{i});X^{-i}\right)\,=\, J^{i}\!\left(s,\phi^{i}(q,y^{i});X^{-i}\right)
\end{multline*}

\par To prove the claim, it only remains to show that $J^{i}(s,\phi^{i}(q,y^{i});X^{-i}) \geq J^{i}(s,y^{i};X^{-i})$. To this end, given $s \in [0,T]$, we again consider a decreasing sequence $\{\epsilon_{n}\} \in (s,T]$ such that $\epsilon_{n}\downarrow s$, and corresponding to each element $\epsilon_{n}$ of the sequence, we consider a control process $\hat{X}^{\epsilon_{n}}$ defined as in (\ref{huep}). In view of the above, for $k \geq \lvert\,-q/(\epsilon_{n}-s)\,\rvert$ we then have 
\begin{equation*}
J^{i}\!\left(s,\phi^{i}(q,y^{i});X^{-i}\right)\,\geq\, J^{i,k}\!\left(s,\phi^{i}(q,y^{i});X^{-i}\right)
\end{equation*}

\par Further, since $\hat{X}^{\epsilon_{n}}$ and $X^{-i}$ are admissible, it again follows from \cite[D.5, Appendix D]{fleming2006controlled} that for a given $n \in \mathds{N}$ and $s\leq t\leq \epsilon_{n}$, $\hat{Y}^{i,\epsilon_{n}}_{t}$ converges to $\phi^{i}(q,y^{i})$ $\mathds{P}$-almost surely as $t\downarrow s$. This in conjunction with \textcolor{azul-pesc}{\Cref{jconv} (iii)} then implies that 
\begin{equation*}
J^{i,k}\!\left(s,\phi^{i}(q,y^{i});X^{-i}\right)\,\geq\, \limsup\limits_{t\,\downarrow\,s}\,J^{i,k}\!\left(s,\hat{Y}^{i,\epsilon_{n}}_{t};X^{-i}\right)
\end{equation*}

\par Note that the inequality above holds in particular when $t = \epsilon_{n}$. Thus, in view of \textcolor{azul-pesc}{\Cref{jconv} (ii)}, \textcolor{azul-pesc}{\Cref{jconv} (iii)}, and \textcolor{azul-pesc}{\Cref{hflowconv}} we can obtain the desired relation  as follows
\begin{multline*}
J^{i}\!\left(s,\phi^{i}(q,y^{i});X^{-i}\right)\,=\, \limsup\limits_{k\,\rightarrow\,\infty}\,J^{i,\,k}\!\left(s,\phi^{i}(q,y^{i});X^{-i}\right)\,\geq\,\limsup\limits_{k\,\rightarrow\,\infty}\left(\limsup\limits_{\epsilon_{n}\,\downarrow\, s}\,J^{i,\,k}\!\left(s,\hat{Y}^{i,\epsilon_{n}}_{\epsilon_{n}};X^{-i}\right)\right)\\ \geq\, \limsup\limits_{k\,\rightarrow\,\infty}\,J^{i,\,k}\!\left(s,y^{i};X^{-i}\right)\,=\, J^{i}\!\left(s,y^{i};X^{-i}\right)
\end{multline*}

\par It is then immediate from the above that the best-response value function of the $i$th investor is invariant with respect to the integral flow $\phi^{i}$.
\end{proof}

\subsection{Auxiliary Problem Invariance}

\par We turn our attention to establishing an invariance result for the auxiliary control problem. To prove the invariance of value function of $i$th investor's auxiliary control problem $V^{i}$, with respect to the diffeomorphic integral flow $\phi^{i}$, as well as the subsequent equivalence result for the value functions of the two control problems, we recollect below certain definitions and results related to the theory of piece-wise constant controls. The interested reader is directed to \cite{krylov2008controlled} for a comprehensive treatment.

\par Recall that the state space for the auxiliary control process is $\Delta$, which is compact in the topology induced by the Euclidean metric. Hence, we can select a countable subset $\mathds{S}\,=\,\{s_{m},\,m\in\mathds{N}\} \subset \Delta$, which is dense everywhere in $\Delta$. Further, we define $\mathds{S}_{N}=\{s_{m}\},\,m\in\{1,2,\!...,N\}$. For a given initial time $s\,\in\,[0,T]$, we let $I_{n}=\left\{s\,=\,t_{0},\,t_{1},\,...\,,\,t_{n}\,=\,T\right\}$ denote a partition of the time interval $[s,\,T]$, with $\mathcal{A}^{a,\,pc}_{s}\!\left(I_{n},N\right)$ denoting the collection of piece-wise constant controls for a given integer $N$ and partition $I_{n}$. We say that $\Pi = \pi:\Omega\times[s,T]\rightarrow \Delta\in\mathcal{A}^{a,\,pc}_{s}\!\!\left(I_{n},N\right)$ if, $(i)$ for every $(\omega,t)\in\Omega\times[s,T]$, we have $\pi_{t}(\omega)\in\mathds{S}_{N}$, and $(ii)$ for every $\omega\in\Omega$, we have $\pi_{t}(\omega)=\pi_{t_{k+1}}(\omega)$ for $t\in(t_{k},t_{k+1}]$, where $k\,=\,0,...,\,n-1$. Further, we also define  
\begin{equation*}
    \mathcal{A}^{a,\,pc}_{s}\!\!\left(I_{n}\right)=\,\bigcup\limits_{N}\ \mathcal{A}^{a,\,pc}_{s}\!\!\left(I_{n},N\right)\quad \text{and}\quad\mathcal{A}^{a,\,pc}_{s}=\,\bigcup\limits_{I_{n}}\ \mathcal{A}^{a,\,pc}_{s}\!\!\left(I_{n}\right)
\end{equation*}

\par Next, we define a notion of convergence for stochastic processes which is relevant to the present context . To this end, we first equip the space of admissible auxiliary controls $\mathcal{A}^{a}_{s}$ with a metric $\mathrm{d}$, where given $\Pi_{1} = \{\pi_{1,\,t}\}_{t\,\in\,\left[s,\,T\right]}$, $\Pi_{2} = \{\pi_{2,\,t}\}_{t\,\in\,\left[s,\,T\right]} \in \mathcal{A}^{a}_{s}$, we define $\mathrm{d}\left(\Pi_{1},\Pi_{2}\right)$ as follows
\begin{equation*}
    \mathrm{d}\left(\Pi_{1},\Pi_{2}\right)\, =\, \mathds{E}\left[\,\int\limits^{T}_{s}\!\!\left\lvert\,\pi_{1,\,t}-\pi_{2,\,t}\,\right\rvert dt\,\right]
\end{equation*}

\par Given a sequence $\{\Pi_{n}\}_{n\,\in\,\mathds{N}}\subseteq\mathcal{A}^{a}_{s}$, we say that $\Pi_{n}$ converges to $\Pi$, where $\Pi\in\mathcal{A}^{a}_{s}$, if we have $\mathrm{d}\left(\Pi_{n},\Pi\right)\rightarrow 0$, as $n\rightarrow\infty$. Further, given a partition $I_{n}=\left\{s=t_{0},t_{1},\!...,t_{n}=T\right\}$ of the time interval $[s,T]$, we define the diameter of $I_{n}$ as $\max(t_{i+1} - t_{i})$, where $i \in \{0,1,\!..., n - 1\}$. The following lemma establishes the useful result that we can characterize any admissible control $\Pi \in \mathcal{A}^{a}_{s}$ as the limit of a sequence of piece-wise constant controls. For the sake of readability, the proof is relegated to the technical appendix

\begin{lemma}\label{pcapp}
\par Consider a sequence of nested partitions $\left\{I_{n}\right\}_{n\,\in\,\mathds{N}}$ of the time interval $[s,T]$ such that the diameter of $I_{n}$ converges to zero as $n\rightarrow\infty$. Then, for each admissible control $\Pi\in\mathcal{A}^{a}_{s}$, there exists a sequence of controls $\left\{\Pi_{n}\in \mathcal{A}^{a,\,pc}_{s}\!\left(I_{n}\right)\right\}_{n\,\in\,\mathds{N}}$ converging to $\Pi$, where convergence is defined in terms of the metric $\mathrm{d}$.
\end{lemma}

\par With the characterization result for the auxiliary controls at our disposal, we next consider an analogous characterization result for the auxiliary state process. To this end, given $X^{-i}\in\mathcal{A}_{s}$, we consider an arbitrary sequence $\{\Pi_{n}\}_{n\,\in\,\mathds{N}}$, where $\Pi_{n}\in\mathcal{A}^{a,\,pc}_{s}\!\left(I_{n}\right)$, and we let $Z^{i,\,\Pi_{n}}$ denote the solution of the system of stochastic differential equations (\ref{astate}), with $\pi^{i}_{t} = \pi_{n,\,t}$ for all $t\in[s,T]$. Further, given $\Pi\in\mathcal{A}^{a}_{s}$, we let $Z^{i,\,\Pi}$ denote the solution of the system of stochastic differential equations (\ref{astate}), with $\pi^{i}_{t}=\pi_{t}$, for all $t \in [s,T]$. The following lemma then establishes a handy characterization result for the controlled state process $Z^{i,\,\Pi}$ as the limit of a sequence of controlled processes $\left\{Z^{i,\,\Pi_{n}}\right\}_{n\,\in\,\mathds{N}}$. For the sake of readability, the proof is relegated to the technical appendix

\begin{lemma}\label{pcsapp}
\par Consider a deterministic initial time $s\in[0,T]$, and let $\left\{\Pi_{n}\in\mathcal{A}^{a,\,pc}_{s}\!\left(I_{n}\right)\right\}_{n\,\in\,\mathds{N}}$ denote a sequence of piece-wise constant controls, which converges in terms of the metric $\mathrm{d}$, to the admissible control $\Pi\in\mathcal{A}^{a}_{s}$. Given the corresponding state processes $\left\{Z^{i,\,\Pi_{n}}\right\}_{\,n\,\in\,\mathds{N}}$ and $Z^{i,\,\Pi}$ respectively, we can find a subsequence $\{n_{m}\}_{m\,\in\,\mathds{N}}$, such that along the subsequence we have
\begin{equation*}
\lim\limits_{m\,\rightarrow\,\infty}\mathds{E}\left[\sup_{t\,\in\,[s,T]}\,\left\lvert Z^{i,\,\Pi_{n_{\,m}}}_{t}\,-\,Z^{i,\Pi}_{t}\right\rvert^{\,2}\right]\,=\,0    
\end{equation*}
\end{lemma}

\par The preceding results suggest that one may reasonably approximate $V^{i}$, the value function of the auxiliary optimal control problem for investor $i$, by restricting attention to the class of piece-wise constant controls as the diameter of the partition of the time interval $[s,T]$, employed in defining the piece-wise constant control, becomes smaller. We formalize this intuition in the following lemma, the proof of which is relegated to the technical appendix for the sake of readability.

\begin{lemma}\label{pcval}
\par Fix a deterministic initial time $s \in [0,T]$, and a deterministic initial state $z^{i} \in \mathds{R}^{\lvert Z^{i}\rvert}\,$. Given a sequence of nested partitions $\{I_{n}\}_{n\,\in\,\mathds{N}}$ of the time interval $[s,T]$, such that the diameter of the partitions converges to zero, we have 
\begin{equation*}
\lim\limits_{n\,\rightarrow\,\infty} \left\{\sup\limits_{\Pi_{n}\, \in\, \mathcal{A}^{a,\,pc}_{s}\!\left(I_{n}\right)} H^{i}\!\left(\Pi_{n};X^{-i}\right)\right\}\, =\, V^{i}\!\left(s,z^{i};X^{-i}\right)
\end{equation*}
\end{lemma}

\par The apparatus of the auxiliary control problem as well as the formalism related to the theory of piece-wise constant controls developed thus far in this section is employed to prove the second principal assertion of our work, which concerns the invariance of the value function of the auxiliary control problem of investor $i$, formalized in the theorem below. The arguments used in the proving the theorem rely extensively on the construction of the auxiliary control problem introduced in \textcolor{azul-pesc}{\Cref{auxcp}} along with ancillary results related to the theory of piece-wise controls introduced in this section.

\begin{theorem}[\textbf{\textsc{Invariance - II}}]\label{inv2}
Consider a fixed initial time $s \in [0,T]$, as well as a deterministic initial state $Z^{i}_{s} = z^{i} \in \mathds{R}^{\left\lvert Z^{i} \right\rvert}$. Given a scalar $q \in \Delta$, and $X^{-i}\in\mathcal{A}_{s}$, the auxiliary best-response value function of investor $i$ is invariant with respect to the integral flow $\phi^{i}\!\left(q,\,z^{i}\right)$, that is, we have $V^{i}\!\left(s,z^{i};X^{-i}\right)\, =\, V^{i}\!\left(s,\phi^{i}\!\left(q,z^{i}\right);X^{-i}\right)$.
\end{theorem}

\begin{proof}
\par The proof of the claim is straightforward in the particular instance when $s=T$. To see this, note that when $s=T$, we have $Z^{i}_{T} = z^{i}$, and so in view of the definition of $V^{i}$ it is immediate that

\begin{equation*}
V^{i}\!\left(T,z^{i};X^{-i}\right)\, =\, \sup\limits_{m\, \in\, \mathds{R}}\,u^{i}\circ\phi^{i}\!\left(-m,Z^{i}_{T}\right)\, =\, \sup\limits_{m\, \in\, \mathds{R}}\,u^{i}\circ\phi^{i}\!\left(-m,z^{i}\right)
\end{equation*}

\par Moreover, since the Bernoulli utility function $u^{i}$ is strictly increasing, we note that the supremum in the definition above remains invariant with respect to translation of the flow $\phi^{i}$, which then implies that given a scalar $q \in \Delta$ we have
\begin{equation*}
V^{i}\!\left(T,z^{i};X^{-i}\,\right)\, =\, \sup\limits_{m\, \in\, \mathds{R}}\,u^{i}\circ\phi^{i}\!\left(-m + q,z^{i}\right)\, =\, \sup\limits_{m \,\in\, \mathds{R}}\,u^{i}\circ\phi^{i}\!\left(-m,\phi^{i}\!\left(q,z^{i}\right)\right)
\end{equation*}

\par The second equality in the equation above is a straightforward consequence of the flow property. The claim is then immediate from the above. 

\par Next, we establish that $V^{i}\!\left(s,z^{i};X^{-i}\right)\geq V^{i}\!\left(s,\phi^{i}(q,z^{i});X^{-i}\right)$ for a given initial time $s \in [0,T)$. To this end, we consider the deterministic strategy $\Pi_{0} = \{\pi_{0,\,t}\}_{\,t\,\in\,[s,\,T]}$, where $\pi_{0,\,t} = 0$, for all $t \in [s,T]$. It is immediate in view of \textcolor{azul-pesc}{\Cref{auxadmcon}} that $\Pi_{0} \in \mathcal{A}^{a}_{s}$. Further, for a given $X^{-i} \in \mathcal{A}_{s}$, we let $Z^{\,i,\,0}_{T}$ denote the terminal value of the controlled auxiliary state process corresponding to the control $\Pi_{0}$ for investor $i$, with $Z^{\,i,\,0}_{s} = z^{i}$. Recalling the definition of $V^{i}$, we then have
\begin{equation*}
V^{i}\!\left(s,z^{i};X^{-i}\right)\, \geq\, \mathds{E}\left[\,\sup_{m\,\in\,\mathds{R}}u^{i}\circ\phi^{i}\big(\!-m,Z^{\,i,\,0}_{T}\big)\right]
\end{equation*}

\par Further, it follows by definition that for $t \in [0,T]$ we have $Z^{\,i,\,0}_{t} = \phi^{i}(0,\Gamma^{\,i,\,0}_{t})$, where recall that $\Gamma^{\,i,\,0}_{t}$ denotes the value at the time $t$ of the abridged process $\Gamma^{\,i,\,0}$, which is defined as the solution of the system of stochastic differential equations (\ref{abstate}), with $\Gamma^{\,i,\,0}_{s} = z^{i}$. Thus, we can rewrite the equation above as
\begin{equation*}
V^{i}\!\left(s,z^{i};X^{-i}\right)\, \geq\, \mathds{E}\left[\,\sup_{m\,\in\,\mathds{R}}u^{i}\circ\phi^{i}\big(\!-m,\Gamma^{\,i,\,0}_{T}\big)\right]
\end{equation*}

\par Next, for a given partition $I_{n} = \left\{s=t_{0},\,t_{1},\!...,\,t_{n}=T\right\}$ of the time interval $[s,T]$, we consider a piece-wise constant control $\Pi_{n} \in \mathcal{A}^{a,\,pc}_{s}\!\left(I_{n}\right)$ such that $\pi_{n,\,t_{0}} = q \in \Delta$. Further, given $X^{-i} \in \mathcal{A}_{s}$, we let ${Z}^{\,i,\,\Pi_{n}}$ denote the controlled auxiliary state process corresponding to auxiliary control $\Pi_{n}$ for investor $i$, with initial condition ${Z}^{\,i,\,\Pi_{n}}_{s} = \phi^{i}(q,z^{i})$. Recalling the definition of the auxiliary state process we obtain ${Z}^{\,i,\,\Pi_{n}}_{T} = \phi^{i}\big(\pi_{n,\,t_{n}},\Gamma^{i,\,0}_{T}\big)$ which in turn implies that
\begin{equation*}
V^{i}\!\left(s,z^{i};X^{-i}\right)\, \geq\, \mathds{E}\left[\, \sup_{m\,\in\,\mathds{R}}u^{i}\circ\phi^{i}\big(\!-m,\phi^{i}\big(\!-\pi_{n,\,t_{n}},{Z}^{\,i,\,\Pi_{n}}_{T}\big)\big)\right]
\end{equation*}

\par Note that the supremum in the equation above is again translation invariant on account of strict monotonicity of the Bernoulli utility function $u^{i}$, which in conjunction with the flow property of $\phi^{i}$ and the definition of $V^{i}$, then implies that we have  
\begin{equation*}
V^{i}\!\left(s,z^{i};X^{-i}\right)\ \geq\, \sup\limits_{\Pi_{n}\, \in\, \mathcal{A}^{a,\,pc}_{s}\!\left(I_{n}\right)}\mathds{E}\left[\, \sup_{m\,\in\,\mathds{R}}u^{i}\circ\phi^{i}\big(\!-m,{Z}^{\,i,\,\Pi_{n}}_{T}\big)\right]\ \geq\, \sup\limits_{\Pi_{n}\, \in\, \mathcal{A}^{a,\,pc}_{s}\!\left(I_{n}\right)} \mathds{E}\left[\,u^{i}\big(\mathrm{w}^{\,i,\,\Pi_{n}}_{T}\big)\right]
\end{equation*}

\par Note that since the choice of the partition $I_n$ above was arbitrary, it then follows in view of \textcolor{azul-pesc}{\Cref{pcval}} that we have $V^{i}\!\left(s,\,z^{i};\,X^{-i}\right)\geq V^{i}\!\left(s,\,\phi^{i}\!(q,\,z^{i});\,X^{-i}\right)$. 

\par In order to prove the claim, it remains to show that $V^{i}\!\left(s,\phi^{i}(q,z^{i});X^{-i}\right)\geq V^{i}\!\left(s,z^{i};X^{-i}\right)$, for a given initial time $s\in[0,T)$. To this end, we fix $q \in \Delta$ and consider the deterministic strategy $\Pi_{q} = \{\pi_{q,\,t} \}_{t\,\in\,[s,T]}$ for investor $i$, where $\pi_{q,\,t} = q$, for $t\in[s,\,T]$. It is immediate in view of \textcolor{azul-pesc}{\Cref{auxadmcon}} that $\Pi_{q} \in \mathcal{A}^{a}_{s}$. Further, given $X^{-i} \in \mathcal{A}_{s}$, we let $Z^{\,i,\,q}_{T}$ denote the value of the corresponding controlled auxiliary state process at the terminal time $T$, with initial condition $Z^{\,i,\,q}_{s} = \phi^{i}(q,z^{i})$. Recalling the definition of $V^{i}$, it then follows that we have
\begin{equation*}
    V^{i}\!\left(s,\phi^{i}(q,z^{i});X^{-i}\right)\, \geq\, \mathds{E}\left[\, \sup_{m\,\in\,\mathds{R}}u^{i}\circ\phi^{i}\big(\!-m,Z^{\,i,\,q}_{T}\big)\right]
\end{equation*}

\par Further, as earlier we let $\Gamma^{\,i,\,0}_{t}$ denote the value at the time $t$ of the abridged process $\Gamma^{\,i,\,0}$, which is defined as the solution of the system of stochastic differential equations (\ref{abstate}), with $\Gamma^{\,i,\,0}_{s} = z^{i}$. It follows from the definition of the abridged process $\Gamma^{\,i}$ that $\Gamma^{\,i,\,0}_{t} = \phi^{i}(-q,Z^{\,i,\,q}_{t})$ for all $t \in [s,T]$, which in turn implies that
\begin{equation*}
    V^{i}\!\left(s,\phi^{i}(q,z^{i});X^{-i}\right)\, \geq\, \mathds{E}\left[\, \sup_{m\,\in\,\mathds{R}}u^{i}\circ\phi^{i}\big(\!-m,\phi^{i}\big(q,\Gamma^{\,i,\,0}_{T}\big)\big)\right]
\end{equation*}

\par Next, we consider a partition $I_{n} = \left\{s=t_{0},\,t_{1},...,\,t_{n}=T\right\}$ of the time interval $[s,T]$, and let $\widetilde{\Pi}_{n} \in \mathcal{A}^{a,\,pc}_{s}\!\left(I_{n}\right)$ be an admissible piece-wise constant control such that $\tilde{\pi}_{n,\,t_{0}} = 0 \in \Delta$. Further, given $X^{-i}\in\mathcal{A}_{s}$, we let $\widetilde{Z}^{i}$ denote the controlled auxiliary state process corresponding to the control $\widetilde{\Pi}_{n}$ for investor $i$. Note that by definition of the auxiliary state process the initial state value equals $\widetilde{Z}^{i}_{s} = \phi^{i}(0,\Gamma^{\,i,\,0}_{s}) = z^{i}$, and we have $\widetilde{Z}^{i}_{T} = \phi^{i}(\tilde{\pi}_{n,\,t_{n}},\Gamma^{i,\,0}_{T})$ which in turn implies that the equation above can be rewritten as
\begin{equation*}
    V^{i}\!\left(s,\phi^{i}(q,z^{i});X^{-i}\right)\, \geq\ \mathds{E}\left[\,\sup_{m\,\in\,\mathds{R}}u^{i}\circ\phi^{i}\big(\!-m,\phi^{i}\big(q,\phi^{i}\big(\!-\tilde{\pi}_{n,\,t_{n}},\widetilde{Z}^{i}_{T}\big)\big)\big)\right]
\end{equation*}

\par Again, we appeal to the flow property of $\phi^{i}$, and the translation invariance of the supremum in the equation above that follows on account of the strict monotonicity of the Bernoulli utility function $u^{i}$, in conjunction with the definition of $V^{i}$ which then leads us to  
\begin{equation*}
V^{i}\!\left(s,\phi^{i}(q,z^{i});X^{-i}\right)\, \geq\, \sup\limits_{\widetilde{\Pi}_{n}\, \in\, \mathcal{A}^{a,\,pc}_{s}\!\left(I_{n}\right)}\mathds{E}\left[\, \sup_{m\,\in\,\mathds{R}}u^{i}\circ\phi^{i}\big(\!-m,\widetilde{Z}^{i}_{T}\big)\right]\, \geq\, \sup\limits_{\widetilde{\Pi}_{n}\, \in\, \mathcal{A}^{a,\,pc}_{s}\!\left(I_{n}\right)} \mathds{E}\left[u^{i}\big(\tilde{\mathrm{w}}^{i}_{T}\big)\right]
\end{equation*}

\par Finally, since the choice of the partition $I_n$ above was arbitrary, we consider the limit as $n\rightarrow \infty$ and invoke \textcolor{azul-pesc}{\Cref{pcval}}, which then implies that $V^{i}\!\left(s,z^{i};X^{-i}\right)\geq V^{i}\!\left(s,\phi^{i}(q,z^{i});X^{-i}\right)$, whereupon the claim in the statement of the lemma follows as a consequence.
\end{proof}

\subsection{Value Function Equivalence}

\par This subsection is devoted to the statement and proof of the value function equivalence result. The theorem below, which documents this principal result, serves as a game-theoretic counterpart to \cite[Theorem 1]{lasry2000classe} in which the authors limit attention to a single agent framework. Nevertheless, given our focus on Markovian strategies in the present work, the proof of the theorem proceeds along similar lines. 

\par We first show that $V^{i}$ is bounded above by $J^{i}$ by appealing to the piece-wise constant control framework for the auxiliary control problem outlined earlier, which in conjunction with the standard nature of the auxiliary control problem leads us to the desired result. To show the converse, we first show that $V^{i}$ serves as an upper bound for the truncated best-response value function $J^{i,n}$ using a viscosity solution approach, from which we can obtain the desired result by appealing to \textcolor{azul-pesc}{\Cref{jconv}}.

\begin{theorem}[\textbf{\textsc{Equivalence}}]\label{vinv}
\par Consider a fixed initial time $s \in [0,T)$, a fixed initial state for the singular control problem $Y^{i}_{s} = y^{i} \in \mathds{R}^{\lvert\, Y^{i}\,\rvert}$, a scalar $q \in \Delta$, and a fixed initial state for the auxiliary control problem $z^{i} \in \mathds{R}^{\left\lvert Z^{i}\right\rvert}$, such that $z^{i} = \big(y^{i}\backslash \pi^{i}_{s}\big)$. Given $X^{-i}\in\mathcal{A}_{s}$, the best-response value function of the of the $i$th investor is equivalent to the value function of the corresponding auxiliary control problem, that is we have $J^{i}\!\left(s,y^{i};X^{-i}\right)\ =\ V^{i}\!\left(s,z^{i};X^{-i}\right)$
\end{theorem}

\begin{proof}
\par First, we establish that $J^{i}\!\left(s,y^{i};X^{-i}\right) \geq V^{i}\!\left(s,z^{i};X^{-i}\right)$. To this end, given an initial time $s \in [0,T)$ such that $Y^{i}_{s} = y^{i}$, we consider the deterministic strategy $\widehat{X}^{i} = \{\hat{x}^{i}_{t}\}_{t\, \in\, [s,T]}$, where $\hat{x}^{i}_{t} = 0$, for $t \in [s,T]$. Given $X^{-i} \in \mathcal{A}_{s}$, it is immediate that if the $i$th investor were to follow the strategy $\widehat{X}^{i}$ from initial time $s$ up to time $T-\epsilon$, for $\epsilon > 0$, then we would have $Y^{i}_{T-\epsilon}\, =\, (\Gamma^{i,\,0}_{T-\epsilon},\pi^{i}_{s})$, where $\Gamma^{i,\,0}$ denotes the abridged state process defined by (\ref{abstate}), with initial condition $\Gamma^{i,\,0}_{s} = Y^{i}_{s}\backslash \pi^{i}_{s}$. Given that $J^{i}$ is defined to be the value function of $i$th investor's singular best-response problem, and in view of the fact that $\hat{X}^{i} \in \mathcal{A}_{s}$, it follows by way of the invariance result established in \textcolor{azul-pesc}{\Cref{inv1}} that for a given scalar $q \in \mathds{R}$ we have
\begin{equation*}
    J^{i}\!\left(s,y^{i};X^{-i}\right)\, \geq\, \mathds{E}\left[J^{i}\big(T-\epsilon,\big(\Gamma^{i,\,0}_{T-\epsilon},\pi^{i}_{s}\big);X^{-i}\big)\right]\, \geq\, \mathds{E}\left[J^{i}\big(T-\epsilon,\phi^{i}\big(q,\big(\Gamma^{i,\,0}_{T-\epsilon},\pi^{i}_{s}\big)\big);X^{-i}\big)\right]
\end{equation*}

\par Recall from \textcolor{azul-pesc}{\Cref{jconv}(v)} that the function $J^{i}$ is lower semi-continuous with respect to its first argument at the terminal time $T$, and from \textcolor{azul-pesc}{\Cref{jconv}(iv)} that $J^{i}$ is lower semi-continuous with respect to its second argument, which along with the continuity of the integral flow $\phi^{i}$ then implies that we have
\begin{equation*}
    J^{i}\!\left(s,y^{i};X^{-i}\right)\, \geq\, \liminf\limits_{\epsilon\,\downarrow\, 0}\, \mathds{E}\left[J^{i}\big(T-\epsilon,\phi^{i}\big(q,\big(\Gamma^{i,\,0}_{T-\epsilon},\pi^{i}_{s}\big)\big);X^{-i}\big)\right]
\end{equation*}

\par Moreover, since $\hat{X}^{i} \in \mathcal{A}_{s}$ for $s \in [0,T]$, it follows from the above that we have
\begin{equation*}
    \inf\limits_{u\,\in\,[T-\epsilon,\,T]}\,J^{i}\big(u,\phi^{i}\big(q,\big(\Gamma^{i,\,0}_{T-\epsilon},\pi^{i}_{s}\big)\big);X^{-i}\big)\, \geq\, -\frac{1}{\delta^{i}}\,\mathds{E}\left[\exp{\!\left(-\delta^{i}W^{i}_{s} -\delta^{i}\pi^{i}_{s}\int_{s}^{T}\!\!\!\!dS_{t}\right)} \right]
\end{equation*}

\par Note that the right-hand side above represents the expected utility of the $i$th investor from following the strategy $\hat{X}^{i}$ from initial time $s$ up to time $T$, given $X^{-i} \in \mathcal{A}_{s}$. From the proof of \textcolor{azul-pesc}{\Cref{jconv}}, we know that the right-hand side of the equation above is finite. Thus, we may invoke Fatou's Lemma to establish that
\begin{equation}\label{eq1}
    J^{i}\!\left(s,y^{i};X^{-i}\right)\, \geq\, \mathds{E}\left[\liminf\limits_{\epsilon\,\downarrow\, 0}\,J^{i}\big(T-\epsilon,\phi^{i}\big(q,\big(\Gamma^{i,\,0}_{T-\epsilon},\pi^{i}_{s}\big)\big);X^{-i}\big)\right]
\end{equation}

\par Appealing again to the lower semi-continuity of the function $J^{i}$  with respect to its first argument at the terminal time $T$, as well as with respect to its second argument, see \textcolor{azul-pesc}{\Cref{jconv}}, along with the continuity of the integral flow $\phi^{i}$ and the definition of the function $J^{i}$ at time $T$ we have\begin{equation*}
    \liminf\limits_{\epsilon\,\downarrow\, 0}\,J^{i}\big(T-\epsilon,\phi^{i}\big(q,\big(\Gamma^{i,\,0}_{T-\epsilon},\pi^{i}_{s}\big)\big);X^{-i}\big)\, \geq\, u^{i}\circ\phi^{i}\big(q,\big(\Gamma^{i,\,0}_{T}\!,\pi^{i}_{s}\big)\big)
\end{equation*}

\par Since the choice of $q$ above was arbitrary, it then follows that we have
\begin{equation*}
    \liminf\limits_{\epsilon\,\downarrow\, 0}\,J^{i}\big(T-\epsilon,\phi^{i}\big(q,\big(\Gamma^{i,\,0}_{T-\epsilon},\pi^{i}_{s}\big)\big);X^{-i}\big)\, \geq\, \sup\limits_{q\,\in\,\mathds{R}}\, u^{i}\circ\phi^{i}\big(q,\big(\Gamma^{i,\,0}_{T},\pi^{i}_{s}\big)\big)
\end{equation*}

\par Next, we consider a partition $I_{n} = \left\{s=t_{0},\,t_{1},...,\,t_{n}=T\right\}$ of the time interval $[s,T]$, and let $\Pi_{n} \in \mathcal{A}^{a,\,pc}_{s}\!\left(I_{n}\right)$ be an admissible piece-wise constant control such that $\pi_{n,\,t_{0}} = 0 \in \Delta$. Given $X^{-i}\in\mathcal{A}_{s}$, we let $Z^{\,i,\,\Pi_{n}}$ denote the auxiliary state process corresponding to the strategy $\Pi_{n}$ for the $i$th investor. Note from the definition of the auxiliary state process that the initial state value is given by $Z^{\,i,\,\Pi_{n}}_{s} = \phi^{i}\big(0,\Gamma^{\,i,\,0}_{s}\big) = z^{i}$. Also, given the definition of the auxiliary state process we have $Z^{\,i,\,\Pi_{n}}_{T} = \phi^{i}\big(\pi_{n,\,t_{n}},\Gamma^{i,\,0}_{T}\big)$, which in conjunction with the fact that the Bernoulli utility function $u^{i}$ does not depend on $\pi^{i}_{T}$ then implies that
\begin{equation*}
    \liminf\limits_{\epsilon\,\downarrow\, 0}\,J^{i}\big(T-\epsilon,\phi^{i}\big(q,\big(\Gamma^{i,\,0}_{T-\epsilon},\pi^{i}_{s}\big)\big);X^{-i}\big)\, \geq\, \sup\limits_{q\,\in\,\mathds{R}}\, u^{i}\circ\phi^{i}\big(q,\phi^{i}\big(-\pi_{n,\,t_{n}},Z^{\,i,\,\Pi_{n}}_{T}\big)\big)
\end{equation*}

\par We recall the flow property of $\phi^{i}$, along with the translation invariance of the supremum in the equation above which follows on account of the strict monotonicity of the Bernoulli utility function $u^{i}$, to note that (\ref{eq1}) leads us to
\begin{equation*}
   J^{i}\!\left(s,y^{i};X^{-i}\right)\, \geq\, \sup\limits_{\Pi_{n}\, \in\, \mathcal{A}^{a,\,pc}_{s}\!\left(I_{n}\right)}\mathds{E}\left[\,\sup\limits_{q\,\in\,\mathds{R}}\, u^{i}\circ\phi^{i}\big(q,Z^{\,i,\,\Pi_{n}}_{T}\big)\right]
\end{equation*}

\par Since the choice of the partition $I_n$ above was arbitrary, we consider the limit as $n\rightarrow \infty$ and invoke \textcolor{azul-pesc}{\Cref{pcval}}, whereupon it follows that $J^{i}\!\left(s,y^{i};X^{-i}\right)\geq V^{i}\!\left(s,\,z^{i};X^{-i}\right)$.

\par It remains to prove the converse, that is, we show that $V^{i}\!\left(s,z^{i};X^{-i}\right)\geq J^{i}\!\left(s,y^{i};X^{-i}\right)$. Recall that $V^{i}$ denotes the value function of the auxiliary best-response problem of investor $i$, and hence it represents a viscosity supersolution of the associated Hamilton\textendash Jacobi\textendash Bellman equation. Thus, if we let $\underline{V}^{i}$ denote the lower semi-continuous envelope of $V^{i}$, that is, $\underline{V}^{i}$ denotes the largest lower semi-continuous minorant of $V^{i}$ (the existence of which follows in view of the fact that $V^{i}$ is bounded from below, as shown above), it follows that given $g \in \mathds{C}^{1,2}\big((0,T]\times\mathds{R}^{\lvert Z^{i}\rvert};\mathds{R}\big)$ such that $\underline{V}^{i} - g$ attains a global minimum at $\left(s^{*}\!,z^{*}\right)\in (0,T)\times\mathds{R}^{\lvert Z^{i}\rvert}$, we must have
  \begin{equation*}
  -D_{1}g\!\left(s^{*}\!,z^{*}\right)+\sup\limits_{q\,\in\, \Delta}\left\{\left\langle -D_{2}g\!\left(s^{*}\!,z^{*}\right),\beta^{i}\!\left(q,z^{*}\right)x^{-i}\right\rangle- \frac{1}{2}\mathrm{tr}\left(D_{2}^{2}g\!\left(s^{*}\!,z^{*}\right)\nu^{i}\!\left(q,z^{*}\right)\nu^{i}\!\left(q,z^{*}\right)^{\mathrm{T}}\right)\right\} \geq 0 
  \end{equation*}

\par Moreover, as the functions $\beta^{i}$, $\nu^{i}$ are continuous in their first argument and the set $\Delta$ is compact, we can find a $q^{*} \in \Delta$ such that the equation above can be re-written equivalently as
\begin{equation*}
  -D_{1}g\!\left(s^{*}\!,z^{*}\right)-\left\langle D_{2}g\!\left(s^{*}\!,z^{*}\right),\beta^{i}\!\left(q^{*}\!,z^{*}\right)x^{-i}\right\rangle- \frac{1}{2}\mathrm{tr}\left(D_{2}^{2}g\!\left(s^{*}\!,z^{*}\right)\nu^{i}\!\left(q^{*}\!,z^{*}\right)\nu^{i}\!\left(q^{*}\!,z^{*}\right)^{\mathrm{T}}\right) \geq 0 
  \end{equation*}
  
 \par Further, given an initial auxiliary state $z^{*}$, and an auxiliary control value $q^{*}$, we let $\gamma^{*}$ denote the corresponding initial value of the associated abridged process $\Gamma^{i}$, defined as in \textcolor{azul-pesc}{\Cref{auxcp}}. Recalling the definition of the auxiliary state process $Z^{i}$, it follows that we have $z^{*} = \phi^{i}\!\left(q^{*},\gamma^{*}\right)$, using which we can re-write the equation above equivalently as 
\begin{multline*}
  -D_{1}g\!\left(s^{*}\!,\phi^{i}\!\left(q^{*}\!,\gamma^{*}\right)\right)-\left\langle D_{2}g\!\left(s^{*}\!,\phi^{i}\!\left(q^{*}\!,\gamma^{*}\right)\right),\beta^{i}\!\left(q^{*}\!,\phi^{i}\!\left(q^{*}\!,\gamma^{*}\right)\right)x^{-i}\right\rangle\\ -\, \frac{1}{2}\mathrm{tr}\left(D_{2}^{2}g\!\left(s^{*}\!,\phi^{i}\!\left(q^{*}\!,\gamma^{*}\right)\right)\nu^{i}\!\left(q^{*}\!,\phi^{i}\!\left(q^{*}\!,\gamma^{*}\right)\right)\nu^{i}\!\left(q^{*}\!,\phi^{i}\!\left(q^{*}\!,\gamma^{*}\right)\right)^{\mathrm{T}}\right)\geq 0 
 \end{multline*}

\par Recalling the definitions of $\beta^{i}, \nu^{i}$ from \textcolor{azul-pesc}{\Cref{auxcp}} (\ref{auxceq}) and the fact that $D^{2}_{2}\phi^{i}\!\left(q,\gamma\right) = \mathbf{0}_{4\times4}$ for all $(q,\gamma)$, we can re-write the equation above equivalently as
\begin{multline}\label{visc1}
  -D_{1}g\!\left(s^{*}\!,\phi^{i}\!\left(q^{*}\!,\gamma^{*}\right)\right)-\left\langle D_{2}g\!\left(s^{*}\!,\phi^{i}\!\left(q^{*}\!, \gamma^{*}\right)\right),D_{2}\phi^{i}\!\left(q^{*}\!,\gamma^{*}\right)\hat{a}^{i}\!\left(\gamma^{*}\right)x^{-i}\right\rangle\\ -\, \frac{1}{2}\mathrm{tr}\left(D_{2}\phi^{i}\!\left(q^{*}\!,\gamma^{*}\right)^{\mathrm{T}}D_{2}^{2}g\!\left(s^{*}\!,\phi^{i}\!\left(q^{*}\!, \gamma^{*}\right)\right)D_{2}\phi^{i}\!\left(q^{*}\!,\gamma^{*}\right)\mathrm{\hat{v}}^{i}\!\left(\gamma^{*}\right)\mathrm{\hat{v}}^{i}\!\left(\gamma^{*}\right)^{\mathrm{T}}\right) \geq 0 
  \end{multline}

\par Next, corresponding to $g \in \mathds{C}^{1,2}\big((0,T]\times\mathds{R}^{\lvert Z^{i} \rvert};\mathds{R}\big)$ we define $\hat{g} \in \mathds{C}^{1,2}\big((0,T]\times\mathds{R}^{\lvert Z^{i}\rvert};\mathds{R}\big)$ as $\hat{g}\!\left(s,\gamma\right) = g\!\left(s,\phi^{i}\!\left(q^{*}\!,\gamma\right)\right)$, and further note that in view of the definition of $\hat{g}$ we have 
\begin{align*}
     D_{1}\hat{g}\!\left(s,\gamma\right)\, =\, &\, D_{1}g\!\left(s,\phi^{i}\!\left(q^{*}\!,\gamma\right)\right)\\
     D_{2}\hat{g}\!\left(s,\gamma\right)\, =\, &\, D_{2}g\!\left(s,\phi^{i}\!\left(q^{*}\!,\gamma\right)\right)D_{2}\phi^{i}\!\left(q^{*}\!,\gamma\right)\\
     D^{2}_{2}\hat{g}\!\left(s,\gamma\right)\, =\, &\,  D_{2}\phi^{i}\!\left(q^{*}\!,\gamma\right)^{\mathrm{T}}D_{2}^{2}g\!\left(s,\phi^{i}\!\left(q^{*}\!,\gamma\right)\right)D_{2}\phi^{i}\!\left(q^{*}\!,\gamma\right)+D_{2}g\!\left(s,\phi^{i}\!\left(q^{*}\!,\gamma\right)\right)D^{2}_{2}\phi^{i}\!\left(q^{*}\!,\gamma\right)\\
     =\, &\,  D_{2}\phi^{i}\!\left(q^{*}\!,\gamma\right)^{\mathrm{T}}D_{2}^{2}g\!\left(s,\phi^{i}\!\left(q^{*}\!,\gamma\right)\right)D_{2}\phi^{i}\!\left(q^{*}\!,\gamma\right)
\end{align*}

\par Note that the last equality above follows from the fact that we have $D^{2}_{2}\phi^{i}\!\left(q,\gamma\right) = \mathbf{0}_{4\times4}$ for all $(q,\gamma)$. Recalling the definition of $\hat{g}$ together with the invariance result for the auxiliary control problem established in \textcolor{azul-pesc}{\Cref{inv2}}, we ascertain that whenever $\underline{V}^{i} - g$ attains a global minimum at $\left(s^{*},\phi^{i}(q^{*},\gamma^{*})\right)\in (0,T)\times\mathds{R}^{\lvert Z^{i}\rvert}$, $\underline{V}^{i}(\,\cdot\,,\phi^{i}(q^{*}\!,\cdot\,)) - \hat{g}$ attains a global minimum at $\left(s^{*}\!,\gamma^{*}\right)$. This fact, in conjunction with the definition of $\hat{g}$ implies that we can re-write (\ref{visc1}) equivalently as 
\begin{equation*}
  -D_{1}\hat{g}\!\left(s^{*}\!,\gamma^{*}\right)-\left\langle D_{2}\hat{g}\!\left(s^{*}\!,\gamma^{*}\right),\hat{a}^{i}\!\left(\gamma^{*}\right)x^{-i}\right\rangle - \frac{1}{2}\mathrm{tr}\left(D_{2}^{2}\hat{g}\!\left(s^{*}\!,\gamma^{*}\right)\mathrm{\hat{v}}^{i}\!\left(\gamma^{*}\right)\mathrm{\hat{v}}^{i}\!\left(\gamma^{*}\right)^{\mathrm{T}}\right) \geq 0 
  \end{equation*}
  
\par Next, we consider the function $\mathfrak{U}^{i}\!\left(\,\cdot\,,\phi^{i}(q^{*}\!,\cdot\,);X^{-i}\right): [0,T) \times \mathds{R}^{\lvert Y^{i}\rvert}\rightarrow\mathds{R}$ defined as follows
\begin{equation*}
\mathfrak{U}^{i}\!\left(s,\phi^{i}(q^{*}\!,y);X^{-i}\right) \, =\, V^{i}\!\left(s,\phi^{i}(q^{*}\!,\gamma);X^{-i}\right)
\end{equation*}

\par It is immediate in view of the definition of the abridged state process $\Gamma^{i}$ that in the definition above we have $\gamma = y\backslash \pi^{i}_{s}$ or equivalently $y = (\gamma,\pi^{i}_{s})$. Further, given $f \in \mathds{C}^{1,2}\big((0,T]\times\mathds{R}^{\lvert Y^{i}\rvert};\mathds{R}\big)$ we define $\hat{f} \in \mathds{C}^{1,2}\big((0,T]\times\mathds{R}^{\lvert Y^{i}\rvert};\mathds{R}\big)$ as $\hat{f}\!\left(s,y\right) = f\!\left(s,\phi^{i}\!\left(q^{*}\!,y\right)\right)$. It then follows in view of the definition of $\mathfrak{U}^{i}$ that whenever $\underline{\mathfrak{U}}^{i}(\,\cdot\,,\phi^{i}(q^{*}\!,\cdot\,)) - \hat{f}(\,\cdot\,,\cdot\,)$ attains a global minimum at $\left(s^{*}\!,y^{*}\right)\in (0,T)\times\mathds{R}^{\lvert Y^{i}\rvert}$, $\underline{V}^{i}(\,\cdot\,,\phi^{i}(q^{*}\!,\cdot\,)) - \hat{f}\big(\,\cdot\,,\big(\,\cdot\,,\pi^{i,*}_{s^{*}}\big)\big)$ attains a global minimum at $\left(s^{*}\!,\gamma^{*}\right)$, and therefore the function $\tilde{f} \in \mathds{C}^{1,2}\big((0,T]\times\mathds{R}^{\lvert Z^{i}\rvert};\mathds{R}\big)$ defined as $\tilde{f}\!\left(s,\gamma\right) = \hat{f}\big(s,\big(\gamma,\pi^{i,*}_{s^{*}}\big)\big)$ satisfies
\begin{equation*}
  -D_{1}\tilde{f}\!\left(s^{*}\!,\gamma^{*}\right)-\left\langle D_{2}\tilde{f}\!\left(s^{*}\!,\gamma^{*}\right),\hat{a}^{i}\!\left(\gamma^{*}\right)x^{-i}\right\rangle - \frac{1}{2}\mathrm{tr}\left(D_{2}^{2}\tilde{f}\!\left(s^{*}\!,\gamma^{*}\right)\mathrm{\hat{v}}^{i}\!\left(\gamma^{*}\right)\mathrm{\hat{v}}^{i}\!\left(\gamma^{*}\right)^{\mathrm{T}}\right)\geq 0 
\end{equation*}

\par Recalling the definition of the vector-valued functions $a^{i}, \mathrm{v}^{i}$ along with the fact that in both functions the element corresponding to the state $\pi^{i}$ is identically zero, we can re-write the equation above equivalently as 
\begin{equation*}
  -D_{1}\hat{f}\!\left(s^{*}\!,y^{*}\right)-\left\langle D_{2}\hat{f}\!\left(s^{*}\!,y^{*}\right),a^{i}\!\left(y^{*}\right)x^{-i}\right\rangle- \frac{1}{2}\mathrm{tr}\left(D_{2}^{2}\hat{f}\!\left(s^{*}\!,{y}^{*}\right)\mathrm{{v}}^{i}\!\left({y}^{*}\right)\mathrm{{v}}^{i}\!\left({y}^{*}\right)^{\mathrm{T}}\right)\geq 0 
\end{equation*}
  
\par In view of the above, it follows that the function $\mathfrak{U}^{i}\!\left(s,\phi^{i}(q^{*}\!,\gamma);X^{-i}\right)$, and hence the function $V^{i}\!\left(s,\phi^{i}(q^{*}\!,\gamma);X^{-i}\right) = V^{i}\!\left(s,z;X^{-i}\right)$ is a viscosity supersolution to the following equation for a given $n\in \mathds{N}$
\begin{equation}\label{visc2}
\setlength{\jot}{7pt}
  -D_{1}\psi\!\left(s,y\right)+ n\left\lvert\left\langle D_{2}\psi\!\left(s,y\right),b^{i}\!\left(y\right)\right\rangle\right\rvert-\left\langle D_{2}\psi\!\left(s,y\right),a^{i}\!\left(y\right)x^{-i}\right\rangle -\frac{1}{2}\mathrm{tr}\left(D_{2}^{2}\psi\!\left(s,y\right)\mathrm{v}^{i}\!\left(y\right)\mathrm{v}^{i}\!\left(y\right)^{\mathrm{T}}\right)=0
\end{equation}

\par Additionally, we note that the terminal condition for the equation above is given as 
\begin{equation*}
V^{i}\!\left(T,\phi^{i}(q^{*}\!,\gamma);X^{-i}\right)\, =\, \sup_{m\, \in\, \mathds{R}}\, u^{i}\circ\phi^{i}\left(-m,\gamma\right)\, \geq\, u^{i}\!\left(W^{i}_{T}\right)
\end{equation*}

\par Next, we recall that $J^{i,n}$ is defined to be the value function of $i$th investor's truncated best-response problem, that is $J^{i,n}$ denotes the value function of $i$th investor's best-response problem when her control process $X^{i}$ is constrained to lie in $\mathcal{A}^{n}_{s}$. In view of this, it follows that $J^{i,n}$ must satisfy the viscosity subsolution property for the associated Hamilton\textendash Jacobi\textendash Bellman equation. Thus, if we let $\bar{J}^{i,\,n}$ denote the upper semi-continuous envelope of $J^{i,n}$, that is, $\bar{J}^{i,\,n}$ denotes the smallest upper semi-continuous majorant of $V^{i}$ (the existence of which follows in view of the fact that $J^{i,n}$ is bounded, as shown in \textcolor{azul-pesc}{\Cref{jconv}}, it follows that given $f \in \mathds{C}^{1,2}\big((0,T]\times\mathds{R}^{\lvert Y^{i}\rvert};\mathds{R}\big)$, such that $\bar{J}^{i,\,n} - f$ attains a global maximum at $\left(s^{*}\!,y^{*}\right)\in (0,T)\times\mathds{R}^{\lvert Y^{i}\rvert}$, we must have
\begin{align*}
  -D_{1} f\!\left(s,y\right)+\sup_{x\,\in\,\overline{\mathcal{B}_{n}(0)}}\Big\{\!\left\langle -D_{2} f\!\left(s,y\right),b^{i}\!\left(y\right) x\right\rangle\!\Big\} &-\left\langle D_{2} f\!\left(s,y\right),a^{i}\!\left(y\right)x^{-i}\right\rangle\\ &- \frac{1}{2}\mathrm{tr}\left(D_{2}^{2}f\!\left(s,y\right)\mathrm{v}^{i}\!\left(y\right)\mathrm{v}^{i}\!\left(y\right)^{\mathrm{T}}\right)\leq 0
\end{align*}

\par It is immediate in view of the above that $J^{i,n}$ is a viscosity subsolution to (\ref{visc2}), with terminal condition given by $J^{i,n}\!\left(T,y;X^{-i}\right) = u^{i}\!\left(W^{i}_{T}\right)$. Thus, recalling the comparison principle for viscosity solutions of second-order partial differential equations \cite[Theorem V.9.1]{fleming2006controlled} we conclude that given a scalar $q \in \Delta$, an initial time $s \in [0,T)$, and initial states $y^{i} \in \mathds{R}^{\lvert\,Y^{i}\rvert}$, $z^{i} \in \mathds{R}^{\lvert\,Z^{i}\rvert}$, such that $z^{i} = \big(y^{i}\backslash \pi^{i}_{s}\big)$ we have $J^{i,n}\!\left(s,y^{i};X^{-i}\right)\leq V^{i}\!\left(s,\phi^{i}(q,z^{i});X^{-i}\right)$. Further, by way of \textcolor{azul-pesc}{\Cref{jconv}}, and the invariance result for the auxiliary best-response value function established in \textcolor{azul-pesc}{\Cref{inv2}}, it is immediate from the above that we have $J^{i}\!\left(s,y^{i};X^{-i}\right)\leq V^{i}\!\left(s,z^{i};X^{-i}\right)$, whereupon the claim follows immediately.
\end{proof}

\section{Markov\textendash Nash Equilibrium}\label{optpo}

\par The first principal objective of this section is to show that under certain regularity conditions the optimal trajectories of the an investor's best-response problem and the associated auxiliary control problem coincide which facilitates analytical characterization of the best-response of an investor in the Merton\textendash Cournot stochastic differential game given that the auxiliary control problem is tractable by standard methods. 

\par Secondly, in the particular instance when the asset price volatility is constant, we obtain a closed-form expression for the Markov\textendash Nash equilibrium portfolios, which turn out to be deterministic. We also prove that the resulting deterministic Markov\textendash Nash equilibrium portfolios are unique if we limit attention to the class of c\`{a}dl\`{a}g strategies and discuss the role of imperfect competition in potentially explaining the excessive trading puzzle.

\par In the subsequent discussion, we work under an additional simplifying assumption, which entails little loss of generality, that given a deterministic initial time $s \in [0,T]$, we have $\pi^{i}_{t} = 0$, for $t < s$, $i \in \mathcal{I}$. This assumption serves to ensure that the price of the risky asset as well as the portfolio values of the two investors are independent of trading history prior to the initial time, conditional on initial holdings.

\subsection{Best\textendash Response Characterization}

\par In order to characterize the best-response of an investor, we first look to derive an optimal auxiliary control for investor $i$. Since the auxiliary control problem of investor $i$ is a standard optimal control problem by construction, we can in principle solve for the optimal auxiliary control by using either the Hamilton\textendash Jacobi\textendash Bellman approach, or the Stochastic Maximum Principle approach. However, in the following lemma we derive an optimal auxiliary control for investor $i$ through a direct verification argument, which is technically less onerous and serves to shorten the proof considerably.    

\begin{proposition}\label{opaux}
\par Let $\underline{\pi}$ and $\overline{\pi}$ denote the lower and upper bound of the state space $\Delta$ of the auxiliary control process $\pi^{i}$, respectively. Given a fixed initial time $0\leq s<T$, a deterministic initial auxiliary state $z^{i} \in \mathds{R}^{\lvert Z^{i}\rvert}$, an admissible strategy $X^{-i} \in \mathcal{A}_{s}$ for investor $-i$ we can characterize an optimal auxiliary strategy for investor $i$ which is denoted by $\Pi^{i,*}$ as
\begin{equation*}
    \begin{cases} 
    \pi^{i,*}_{t}\, =\,  \underline{\pi}\,\bigvee\Big(\!-\theta^{-i}x^{-i}_{t}/\delta^{i}\sigma^{2}\!\left(P^{i}_{t}+\theta^{i}\pi^{i,*}_{t}\right)\,\bigwedge\,\overline{\pi}\Big),\ t \in \left[s,T\right]\\
    \pi^{i,*}_{s^{-}}\, =\, 0 
    \end{cases}
\end{equation*}
\end{proposition}

\begin{proof}
\par It is immediate in view of \textcolor{azul-pesc}{\Cref{admcon}} and \textcolor{azul-pesc}{\Cref{auxadmcon}} that $\Pi^{i,*}$ as defined in the statement of the theorem is admissible , that is, we have $\Pi^{i,*} \in \mathcal{A}^{a}_{s}$. 

\par Let $\mathrm{w}^{i,*}_{T}$ denote the terminal value of $i$th investor's auxiliary portfolio value process corresponding to auxiliary control $\Pi^{i,*}$. It then follows as a consequence of the definition of the auxiliary best-response value function $V^{i}$, see (\ref{auxvf}), as well as the auxiliary portfolio value process $\mathrm{w}^{i}$, see (\ref{astate}), that $V^{i}\!\left(s,z^{i};X^{-i}\right)\geq\mathds{E}\left[u^{i}(\mathrm{w}^{i,*}_{T})\right]$.

\par Thus, in order to prove the claim, it suffices to show that $V^{i}\!\left(s,z^{i};X^{-i}\right)\leq \mathds{E}\left[u^{i}(\mathrm{w}^{i,*}_{T})\right]$. To this end, note that given an admissible strategy $X^{-i} \in \mathcal{A}_{s}$ for investor $-i$ and the definition of $V^{i}$, see (\ref{auxvf}), it follows that we have
\begin{equation*}\label{vinf}
V^{i}\!\left(s,z^{i};X^{-i}\right)=-\inf\limits_{\Pi^{i}\,\in\,\mathcal{A}^{a}_{s}}\mathds{E}\left[\frac{1}{\delta^{i}}\exp\left(\delta^{i}\mathrm{w}^{i}_{s}+\delta^{i}\theta^{-i}\!\!\!\int^{T}_{s}\!\!\!\pi^{i}_{t}\,x^{-i}_{t}dt-\delta^{i}\!\!\int^{T}_{s}\!\!\!\pi^{i}_{t}\,\sigma\!\left(P^{i}_{t}+\theta^{i}\pi^{i}_{t}\right)dB_{t}\right)\right]
\end{equation*}

\par Next, we define the non-negative valued stochastic process $\left\{\vartheta_{t}\right\}_{t\,\in\,\left[s,T\right]}$ as
\begin{equation*}
\vartheta_{t}\,=\,\exp\left(-\delta^{i}\!\!\int^{t}_{s}\!\!\pi^{i}_{u}\,\sigma\!\left(P^{i}_{u}+\theta^{i}\pi^{i}_{u}\right)dB_{u}\,-\,\frac{\,(\delta^{i})^{2}}{2}\!\!\!\int^{t}_{s}\!\!\left(\pi^{i}_{u}\right)^{2}\!\sigma^{2}\!\left(P^{i}_{u}+\theta^{i}\pi^{i}_{u}\right)du\right)
\end{equation*}

\par Given that $\Pi^{i} \in \mathcal{A}^{a}_{s}$, we have $\pi^{i}_{t} \in \Delta$ by definition for $s\leq t\leq T$. It then follows that $\Pi^{i}$ is uniformly bounded $\mathds{P}$-almost surely, which in conjunction with the fact that the local volatility function $\sigma$ satisfies \textcolor{azul-pesc}{\Cref{siglip}} in turn implies that 
\begin{equation*}
\mathds{E}\left[\exp\left(\,\frac{\,(\delta^{i})^{2}}{2}\!\!\!\int^{T}_{s}\!\!\left(\pi^{i}_{u}\right)^{2}\!\sigma^{2}\!\left(P^{i}_{u}+\theta^{i}\pi^{i}_{u}\right)du\right)\right]<\infty
\end{equation*}

\par It follows in view of the above and Novikov's criterion \cite[Proposition VIII.1.15, Page 332]{revuz2013continuous} that the stochastic process $\left\{\vartheta_{t}\right\}_{t\,\in\,\left[s,T\right]}$ is a uniformly integrable $(\mathds{P},\mathds{F})$-martingale with $\mathds{E}[\vartheta_{T}] = \vartheta_{s} = 1$, and hence we can define a probability measure $\mathds{Q}\approx\mathds{P}$ by way of its Radon\textendash Nikodym derivative $d\mathds{Q}/d\mathds{P}$ as
\begin{equation*}
\frac{d\mathds{Q}}{d\mathds{P}}\,=\,\vartheta_{T}\,=\,\exp\left(-\delta^{i}\!\!\int^{T}_{s}\!\!\pi^{i}_{t}\,\sigma\!\left(P^{i}_{t}+\theta^{i}\pi^{i}_{t}\right)dB_{t}\,-\,\frac{\,(\delta^{i})^{2}}{2}\!\!\!\int^{T}_{s}\!\!\left(\pi^{i}_{t}\right)^{2}\!\sigma^{2}\!\left(P^{i}_{t}+\theta^{i}\pi^{i}_{t}\right)dt\right)
\end{equation*}

\par In view of the equation above and as a consequence of Radon\textendash Nikodym theorem it then follows that we have 
\begin{equation*}
-V^{i}\!\left(s,z^{i};X^{-i}\right)=\inf\limits_{\Pi^{i}\,\in\,\mathcal{A}^{a}_{s}}\mathds{E}^{\mathds{Q}}\left[\frac{1}{\delta^{i}}\exp\left(\delta^{i}\mathrm{w}^{i}_{s}+\delta^{i}\theta^{-i}\!\!\!\int^{T}_{s}\!\!\!\pi^{i}_{t}\,x^{-i}_{t}dt +\frac{(\delta^{i})^{2}}{2}\!\!\!\int^{T}_{s}\!\!\!\!\left(\pi^{i}_{t}\right)^{2}\!\sigma^{2}\!\left(P^{i}_{t}+\theta^{i}\pi^{i}_{t}\right)dt\right)\right]
\end{equation*}

\par Note that $\mathds{E}^{\mathds{Q}}$ in the equation above denotes expectation with respect to the probability measure $\mathds{Q}\approx \mathds{P}$ defined earlier. Further, we note that given the strict monotonicity of the exponential function the equation above leads us to
\begin{equation*}
-V^{i}\!\left(s,z^{i};X^{-i}\right)\geq\frac{\exp\left(\delta^{i}\mathrm{w}^{i}_{s}\right)}{\delta^{i}}\,\mathds{E}^{\mathds{Q}}\left[\exp\left(\,\inf\limits_{\Pi^{i}\,\in\,\mathcal{A}^{a}_{s}}\delta^{i}\theta^{-i}\!\!\!\int^{T}_{s}\!\!\!\pi^{i}_{t}\,x^{-i}_{t}dt +\frac{\,(\delta^{i})^{2}}{2}\!\!\!\int^{T}_{s}\!\!\left(\pi^{i}_{t}\right)^{2}\!\sigma^{2}\!\left(P^{i}_{t}+\theta^{i}\pi^{i}_{t}\right)dt\right)\right]
\end{equation*}

\par Given the equation above, we recall the definition of $\Pi^{i,*}$ in the statement of the theorem so as to obtain
\begin{equation*}
-V^{i}\!\left(s,z^{i};X^{-i}\right)\geq\frac{\exp\left(\delta^{i}\mathrm{w}^{i}_{s}\right)}{\delta^{i}}\,\mathds{E}^{\mathds{Q}}\left[\exp\left(\,\delta^{i}\theta^{-i}\!\!\!\int^{T}_{s}\!\!\!\pi^{i,*}_{t}\,x^{-i}_{t}dt +\frac{\,(\delta^{i})^{2}}{2}\!\!\!\int^{T}_{s}\!\!\big(\pi^{i,*}_{t}\big)^{2}\sigma^{2}\big(P^{i}_{t}+\theta^{i}\pi^{i,*}_{t}\big)dt\right)\right]
\end{equation*}

\par Given the definition of the auxiliary portfolio value, see $\mathrm{w}^{i}$ (\ref{astate}), and from Radon\textendash Nikodym theorem it then follows that the right-hand side of the equation above equals $-\mathds{E}\left[u^{i}(\mathrm{w}^{i,*}_{T})\right]$ whereupon the claim follows immediately.
\end{proof}

\par We next turn our attention to the key issue of analyzing the relationship between the optimal auxiliary control of investor $i$, characterized in the preceding proposition, and investor $i$'s best-response. In the following theorem, we prove that given an admissible strategy of investor $-i$, under certain regularity conditions the optimal auxiliary control of investor $i$ characterized in \textcolor{azul-pesc}{\Cref{opaux}} corresponds to investor $i$'s best-response.

\par Informally, if the optimal auxiliary control of investor $i$ corresponding to a given admissible strategy of investor $-i$ is sufficiently \textit{regular}, we can in principle derive the optimal portfolio holding as well as the optimal trading rate of investor $i$ through investor $i$'s optimal auxiliary control. The proof of the theorem relies on arguments underlying the construction of investor $i$'s auxiliary control problem, outlined in \textcolor{azul-pesc}{\Cref{auxcp}}, in addition to the exponential functional form of investor $i$'s utility function.

\begin{theorem}\label{trajeq1}
\par Consider a fixed initial time $s \in [0,T)$, a deterministic initial portfolio value $W^{i}_{s}$ and a deterministic initial portfolio holding $\pi^{i}_{s}$ for investor $i$. Suppose that $\Pi^{i,\,*} \in \mathcal{A}^{a}_{s}$ represents an optimal auxiliary control for investor $i$ for a given admissible strategy $X^{-i}\in \mathcal{A}_{s}$ of investor $i$ such that for $\mathds{P}$-almost every $\omega$ we have
\begin{enumerate}[(i)]
\item $\pi^{i,*}_{s^{-}}(\omega) = 0$, with $\pi^{i,*}_{s}(\omega) = \pi^{i}_{s}$
\item $\Pi^{i,\,*}\!\left(\omega\right)$ is absolutely continuous with respect to the Lebesgue measure on $(s,T]$ 
\item The derivative process $X^{i,\,*}\!\left(\omega\right)$ of $\Pi^{i,\,*}\!\left(\omega\right)$, defined as $x^{i,\,*}_{t}(\omega) = -\partial_{+} \pi^{i,\,*}_{t}(\omega)/\partial t$, for $s\leq t \leq T$ satisfies $$\mathds{E}\left[\exp{\left\{m\!\int_{s}^{T}\!\!\!\!\left\lvert x^{i,\,*}_{t}\right\rvert dt\int_{s}^{T}\!\!\!\!\left\lvert\hat{x}_{t}\right\rvert dt\right\}}\right] < \infty,\  \text{for all}\ \ m \in \mathds{R},\, \hat{X} = \left\{\hat{x}_{t}\right\}_{t\, \in\, [s,\,T]} \in \mathcal{A}_{s}$$
\end{enumerate}
\par Then, the portfolio value process $W^{i,\,\Pi^{i,\,*}}$, defined as the controlled portfolio value corresponding to portfolio process $\pi^{i}_{t} = \pi^{i,\,*}_{t}\!,\,$ for all $s\leq t \leq T\,$ represents an optimal trajectory for the best-response problem of investor $i$, that is, we have $J^{i}(s,y^{i};X^{-i}) = \mathds{E}\big[u^{i}(W^{i,\,\Pi^{i,\,*}}_{T})\big]$.
\end{theorem}

\begin{proof}
\par Suppose $\Pi^{i,\,*} \in \mathcal{A}^{a}_{s}$ denotes an optimal auxiliary control for investor $i$, for a given admissible strategy $X^{-i} \in \mathcal{A}_{s}$ of investor $-i$, satisfying the desiderata in the statement of the theorem, and let $W^{i,\,\Pi^{i,\,*}}\!$ denote the corresponding portfolio value process as defined in the statement of the theorem. 

\par It then follows as a consequence of \textcolor{azul-pesc}{\Cref{admcon}} that $X^{i,\, *} = \{x^{i,\,*}_{t}\}_{t\, \in\, [s,\,T]}$ as defined in the statement of the theorem is an admissible strategy for investor $i$, that is, we have $X^{i,\, *} \in \mathcal{A}_{s}$. Also, it is immediate from the above that $W^{i,\,X^{i,\,*}}_{T}\!\!=W^{i,\,\Pi^{i,\,*}}_{T}\!\!$ which in conjunction with the definition of $J^{i}$ then implies that 
\[J^{i}\big(s,y^{i};X^{-i}\big)\geq \mathds{E}\big[u^{i}(W^{i,\,X^{i,\,*}}_{T})\big]=\mathds{E}\big[u^{i}(W^{i,\,\Pi^{i,\,*}}_{T})\big]\]

\par To prove the converse, we recall (\ref{stated}), from which it immediately follows that given an optimal auxiliary control $\Pi^{i,*} \in \mathcal{A}^{a}_{s}$ and an admissible strategy $X^{-i} \in \mathcal{A}^{a}_{s}$ for investor $-i$, the terminal value of the stochastic process $W^{i,\,\Pi^{i,*}}\!$ as defined in the statement of the theorem can be written as
\begin{equation*}
W^{i,\,\Pi^{i,\,*}}_{T}\!=W^{i}_{s}+\frac{\,\theta^{i}}{2}\left[\big(\pi^{i,\,*}_{T}\big)^{2}-\big(\pi^{i,\,*}_{s}\big)^{2}\right]-\theta^{-i}\!\!\!\int^{T}_{s}\!\!\!\pi^{i,\,*}_{t}x^{-i}_{t}dt+\,\!\!\int^{T}_{s}\!\!\!\pi^{i,\,*}_{t}\sigma\!\left(S_{t}\right)dB_{t}
\end{equation*}

\par Also, given that the utility function $u^{i}$ of investor $i$ is strictly increasing it follows from the equation above that we have 
\begin{equation*}
\mathds{E}\big[u^{i}(W^{i,\,\Pi^{i,\,*}}_{T})\big]\geq\mathds{E}\left[u^{i}\!\left(W^{i}_{s}-\frac{\,\theta^{i}}{2}\big(\pi^{i,\,*}_{s}\big)^{2}-\theta^{-i}\!\!\!\int^{T}_{s}\!\!\!\pi^{i,\,*}_{t}x^{-i}_{t}dt+\,\!\!\int^{T}_{s}\!\!\!\pi^{i,\,*}_{t}\sigma\!\left(S_{t}\right)dB_{t}\right)\right]
\end{equation*}

\par Further, in view of the assumption regarding the functional form of the utility function $u^{i}$ we can rewrite the equation above as
\begin{equation*}
\mathds{E}\big[u^{i}(W^{i,\,\Pi^{i,\,*}}_{T})\big]\geq-\frac{1}{\delta^{i}}\,\mathds{E}\left[\exp\left(-\delta^{i}\left(W^{i}_{s}-\frac{\,\theta^{i}}{2}\big(\pi^{i,\,*}_{s}\big)^{2}-\theta^{-i}\!\!\!\int^{T}_{s}\!\!\!\pi^{i,\,*}_{t}x^{-i}_{t}dt+\,\!\!\int^{T}_{s}\!\!\!\pi^{i,\,*}_{t}\sigma\!\left(S_{t}\right)dB_{t}\right)\right)\right]
\end{equation*}

\par Next, recall from \textcolor{azul-pesc}{\Cref{auxcp}} that the initial value of $i$th investor's auxiliary portfolio process $\mathrm{w}^{i}_{s}$ equals $W^{i}_{s}$ by construction, and that by definition $S_{t} = P^{i}_{t}+\theta^{i}\pi^{i}_{t}$. Thus, given that $\Pi^{i,\,*}\!$ represents an optimal auxiliary control, given $X^{-i} \in \mathcal{A}_{s}$, it follows in view of (\ref{astate}) that the optimal terminal auxiliary portfolio value denoted by $\mathrm{w}^{i,\,\Pi^{i,\,*}}_{T}\!$ equals
\begin{equation*}
\mathrm{w}^{i,\,\Pi^{i,\,*}}_{T}\!=\,W^{i}_{s}\,-\,\theta^{-i}\!\!\!\int^{T}_{s}\!\!\!\pi^{i,\,*}_{t}x^{-i}_{t}dt+\!\int^{T}_{s}\!\!\!\pi^{i,\,*}_{t}\sigma\!\left(P^{i}_{t}+\theta^{i}\pi^{i,\,*}_{t}\right)dB_{t}
\end{equation*}

\par Therefore, given the equation above and in view of the fact that $\big(\pi^{i,\,*}_{s}\big)^{2}\geq 0$ we have
\begin{equation*}
\mathds{E}\big[u^{i}(W^{i,\,\Pi^{i,\,*}}_{T})\big]\geq-\frac{1}{\delta^{i}}\,\mathds{E}\left[\exp\!\left(\frac{\delta^{i}\theta^{i}}{2}\big(\pi^{i,\,*}_{s}\big)^{2}\right)\exp\!\left(-\delta^{i}\mathrm{w}^{i,\,\Pi^{i,\,*}}_{T}\right)\right]\geq-\frac{1}{\delta^{i}}\,\mathds{E}\left[\exp\!\left(-\delta^{i}\mathrm{w}^{i,\,\Pi^{i,\,*}}_{T}\right)\right]
\end{equation*}

\par Further, since $\Pi^{i,\,*} \in \mathcal{A}^{a}_{s}$ defines an optimal auxiliary control for investor $i$, it follows from the second inequality above that we have
\begin{equation*}
\mathds{E}\big[u^{i}(W^{i,\,\Pi^{i,\,*}}_{T})\big]\geq V^{i}\!\left(s,z^{i};X^{-i}\right)
\end{equation*}

\par Given the equation above as well as the equivalence result established in \textcolor{azul-pesc}{\Cref{vinv}}, it follows that we have $J^{i}(s,y^{i};X^{-i}) \leq \mathds{E}\big[u^{i}(W^{i,\,\Pi^{i,\,*}}_{T})\big]$ which completes the proof.
\end{proof}

\par From the preceding results, we see that the best-response of investor $i$, given an admissible strategy $X^{-i} \in \mathcal{A}_{s}$ of investor $-i$ can be characterized in terms of two distinct \textit{continuation} regions and a single \textit{control} region. The first continuation region is defined by the inequality $-\theta^{-i}x^{-i}_{t}/\delta^{i}\sigma^{2}\!\left(P^{i}_{t}+\theta^{i}\pi^{i}_{t}\right)\geq\overline{\pi}$ which corresponds to aggressive buying on part of investor $-i$. Due to this, there is an upward pressure on the drift of the price of the risky asset on account of investor $-i$'s price impact, which in turn drives up expected returns associated with the risky asset and acts as a \textit{positive externality} for investor $i$'s investment opportunity set, thereby inducing investor $i$ to optimally hold a long position of $\overline{\pi}$ in the risky asset. 

\par As a result, when the trading of investor $-i$ corresponds to the first continuation region, the portfolio process of investor $i$ jumps instantaneously to the upper bound $\overline{\pi}$, and subsequently the trading rate of investor $i$ drops down to zero. Thus, the best-response of investor $i$ is given by a \textit{barrier strategy} in the sense of \cite{kwon2015game}, with $\overline{\pi}$ acting as a barrier for the portfolio process, which can be absorbing or reflecting depending upon investor $-i$'s trading rate process $X^{-i} \in \mathcal{A}_{s}$.

\par The second continuation region is defined by the inequality $-\theta^{-i}x^{-i}_{t}/\delta^{i}\sigma^{2}\!\left(P^{i}_{t}+\theta^{i}\pi^{i}_{t}\right)\leq\underline{\pi}$ that corresponds to aggressive selling of the risky asset on part of investor $-i$, which leads to a downward pressure on the drift of the price of the risky asset due to investor $-i$'s price impact. This effect is responsible for driving down expected returns associated with the risky asset and generating a \textit{negative externality} on investor $i$'s investment opportunity set, which in turn leads to investor $i$ optimally holding a short position of $\underline{\pi}$ in the risky asset. 

\par Consequently, if the trading rate of investor $-i$ corresponds to the second continuation region, the portfolio process of investor $i$ jumps instantaneously to the lower bound $\underline{\pi}$, and as in the case of the first continuation region the trading rate of investor $i$ drops down to zero. The best-response of investor $i$ is again given by a barrier strategy in this case, with $\underline{\pi}$ serving as either an absorbing or a reflecting barrier for the portfolio process, depending upon investor $-i$'s trading rate process $X^{-i} \in \mathcal{A}_{s}$.

\par The lone control region for the best-response of investor $i$, given an admissible strategy $X^{-i} \in \mathcal{A}_{s}$ of investor $-i$ is defined by the inequality $\underline{\pi}<-\theta^{-i}x^{-i}_{t}/\delta^{i}\sigma^{2}\!\left(P^{i}_{t}+\theta^{i}\pi^{i}_{t}\right)<\overline{\pi}$ which signifies moderate trading behaviour by investor $-i$. In the control region, the nature of externality on investor $i$'s investment opportunity set on account of investor $-i$'s price impact is determined by the sign of $x^{-i}_{t}$.

\par In other words, the distinguishing feature of investor $i$'s control region is that the externality imposed on the investor $i$'s investment opportunity set due to investor $-i$'s trading can be positive or negative. Hence, if the trading rate of investor $-i$ corresponds to the control region, investor $i$'s \textit{best-response} portfolio process jumps instantaneously to $-\theta^{-i}x^{-i}_{s}/\delta^{i}\sigma^{2}\!\left(S_{s}\right)$ at the initial time $s$ and follows the trajectory defined by $\pi^{i,*}_{t}=-\theta^{-i}x^{-i}_{t}/\delta^{i}\sigma^{2}\!\left(S_{t}\right)$ thereafter, with investor $i$'s trading rate at time $t \in [s,T]$ given by $d\pi^{i,*}_{t}$.

\par Two remarks are in order here - first, note that although \textcolor{azul-pesc}{\Cref{opaux}} specifies an optimal auxiliary control for investor $i$, it does not imply that the specified optimal auxiliary control is unique. Second, while \textcolor{azul-pesc}{\Cref{trajeq1}} does provide seemingly mild set of criteria under which an optimal auxiliary control corresponds to a best-response, it is agnostic about potential multiplicity of best-responses for a given $X^{-i} \in \mathcal{A}_{s}$. Thus, in view of these limitations, the characterization of investor $i$'s best-response furnished by \textcolor{azul-pesc}{\Cref{opaux}} and \textcolor{azul-pesc}{\Cref{trajeq1}} remains partial at best. 

\par The following theorem, which serves as a partial converse to \textcolor{azul-pesc}{\Cref{trajeq1}}, addresses this issue by proving that under certain conditions investor $i$'s best-response to a given admissible strategy $X^{-i} \in \mathcal{A}_{s}$ of investor $-i$, if it exists, corresponds to an optimal auxiliary control for investor $i$. Thus, uniqueness of investor $i$'s optimal auxiliary control for a given admissible strategy $X^{-i} \in \mathcal{A}_{s}$ of investor $-i$ should in turn imply the uniqueness of investor $i$'s best-response to $X^{-i}$ and vice-versa, so long as the hypotheses of \textcolor{azul-pesc}{\Cref{trajeq1}} and \textcolor{azul-pesc}{\Cref{trajeq2}} are satisfied.

\begin{theorem}\label{trajeq2}
\par Consider a fixed initial time $s \in [0,T)$, a deterministic initial portfolio value $W^{i}_{s}$ and a deterministic initial portfolio holding $\pi^{i}_{s} \in \Delta$ for investor $i$. Suppose that $X^{i,\,*} \in \mathcal{A}_{s}$ represents a best-response strategy for $i$th investor, given an admissible strategy $X^{-i}\in \mathcal{A}_{s}$ for investor $-i$ such that for $s\leq t \leq T$,
\[\pi^{i,\,*}_{t}\, =\,\pi^{i}_{s}-\mathlarger{\int}_{s}^{t}\!\!\! x^{i,\,*}_{u} du\, \in\, \Delta,\ \mathds{P}-\text{almost surely}\]
\par Then, the auxiliary portfolio value process $\mathrm{w}^{i,\,\Pi^{i,\,*}}$ which is defined as the controlled process corresponding to auxiliary control $\pi^{i}_{t} = \pi^{i,\,*}_{t}\!$, for all $s\leq t \leq T$, determines an optimal trajectory for the auxiliary control problem of investor $i$, that is, we have $V^{i}(s,z^{i};X^{-i}) = \mathds{E}\big[u^{i}(\mathrm{w}^{i,\,\Pi^{i,\,*}}_{T})\big]$.
\end{theorem}

\begin{proof}
\par Suppose $X^{i,\,*} \in \mathcal{A}_{s}$ denotes an optimal best-response for investor $i$, given an admissible strategy $X^{-i} \in \mathcal{A}_{s}$ for investor $-i$, satisfying the desiderata in the statement of the theorem, and let $\mathrm{w}^{i,\,\Pi^{i,\,*}}\!$ denote the corresponding auxiliary portfolio value process as defined in the statement of the theorem. 

\par It follows as a consequence of \textcolor{azul-pesc}{\Cref{auxadmcon}} that $\Pi^{i,\, *}\! = \{\pi^{i,\,*}_{t}\}_{t\, \in\, [s,\,T]}$ as defined in the statement of the theorem, is an admissible control for the auxiliary control problem of investor $i$, that is, we have $\Pi^{i,\, *} \in \mathcal{A}^{a}_{s}$ which in conjunction with the definition of $V^{i}$ then implies that 
\[V^{i}\!\left(s,z^{i};X^{-i}\right)\geq\mathds{E}\big[u^{i}(\mathrm{w}^{i,\,\Pi^{i,\,*}}_{T})\big]\]

\par To prove the converse, we recall (\ref{astate}), which implies that given a best-response strategy $X^{i,\,*} \in \mathcal{A}_{s}$ or equivalently $\Pi^{i,\,*}$ for investor $i$ given an admissible strategy $X^{-i} \in \mathcal{A}^{a}_{s}$ for investor $-i$, the terminal value of the stochastic process $\mathrm{w}^{i,\,\Pi^{i,*}}\!$ as defined in the statement of the theorem can be written as
\begin{align*}
\mathrm{w}^{i,\,\Pi^{i,\,*}}_{T}\!=\,&\,\mathrm{w}^{i}_{s}\,-\,\theta^{-i}\!\!\!\int^{T}_{s}\!\!\!\pi^{i,\,*}_{t}x^{-i}_{t}dt+\,\!\!\int^{T}_{s}\!\!\!\pi^{i,\,*}_{t}\sigma\!\left(P^{i}_{t}+\theta^{i}\pi^{i,\,*}_{t}\right)dB_{t}\\
=\,&\,\mathrm{w}^{i}_{s}\,-\,\theta^{-i}\!\!\!\int^{T}_{s}\!\!\!\pi^{i,\,*}_{t}x^{-i}_{t}dt+\,\!\!\int^{T}_{s}\!\!\!\pi^{i,\,*}_{t}\sigma\!\left(P^{i}_{t}+\theta^{i}\pi^{i,\,*}_{t}\right)dB_{t}\,+\,\frac{\,\theta^{i}}{2}\left[\big(\pi^{i,\,*}_{T}\big)^{2}-\,\big(\pi^{i,\,*}_{s}\big)^{2}\right]\\
&\,-\,\frac{\,\theta^{i}}{2}\left[\big(\pi^{i,\,*}_{T}\big)^{2}-\,\big(\pi^{i,\,*}_{s}\big)^{2}\right]\,
\end{align*}

\par Given that the utility function $u^{i}$ of investor $i$ is strictly increasing, it follows from the equation above that we have 
\begin{equation}\label{traj1}
\mathds{E}\big[u^{i}(\mathrm{w}^{i,\,\Pi^{i,\,*}}_{T})\big]\geq\mathds{E}\left[u^{i}\!\left(\mathrm{w}^{i,\,\Pi^{i,\,*}}_{T}\!+\frac{\,\theta^{i}}{2}\left[\big(\pi^{i,\,*}_{T}\big)^{2}-\big(\pi^{i,\,*}_{s}\big)^{2}\right]-\frac{\,\theta^{i}}{2}\big(\pi^{i,\,*}_{T}\big)^{2}\right)\right]
\end{equation}

\par Recall from \textcolor{azul-pesc}{\Cref{auxcp}} that $\mathrm{w}^{i}_{s}$ the initial value of $i$th investor's auxiliary portfolio process equals $W^{i}_{s}$ by construction and that by definition we have $S_{t} = P^{i}_{t}+\theta^{i}\pi^{i}_{t}$. Further, given that $X^{i,\,*}\!$ represents an optimal best-response for investor $i$, given an admissible strategy $X^{-i} \in \mathcal{A}_{s}$ for investor $-i$ it follows in view of (\ref{stated}) that the optimal terminal portfolio value denoted by $W^{i,\,\Pi^{i,\,*}}_{T}\!$ equals
\begin{align*}
W^{i,\,\Pi^{i,\,*}}_{T}\!=\,&\,W^{i}_{s}\,+\frac{\,\theta^{i}}{2}\!\left[\big(\pi^{i,\,*}_{T}\big)^{2}\!-\big(\pi^{i,\,*}_{s}\big)^{2}\right]-\theta^{-i}\!\!\!\int^{T}_{s}\!\!\!\pi^{i,\,*}_{t}x^{-i}_{t}dt+\!\int^{T}_{s}\!\!\!\pi^{i,\,*}_{t}\sigma\!\left(S_{t}\right)dB_{t}\\
=\,&\,\mathrm{w}^{i,\,\Pi^{i,\,*}}_{T}\!+\frac{\,\theta^{i}}{2}\!\left[\big(\pi^{i,\,*}_{T}\big)^{2}\!-\big(\pi^{i,\,*}_{s}\big)^{2}\right]
\end{align*}

\par In view of the assumption regarding the functional form of the utility function $u^{i}$ we can thus rewrite (\ref{traj1}) as
\begin{equation*}
\mathds{E}\big[u^{i}(\mathrm{w}^{i,\,\Pi^{i,\,*}}_{T})\big]\geq-\frac{1}{\delta^{i}}\,\mathds{E}\left[\exp\left(-\delta^{i}\left(W^{i,\,\Pi^{i,\,*}}_{T}\!-\frac{\,\theta^{i}}{2}\big(\pi^{i,\,*}_{T}\big)^{2}\right)\right)\right]
\end{equation*}

\par Given the equation above and the fact that $\big(\pi^{i,\,*}_{T}\big)^{2}\geq 0$ it follows that
\begin{equation*}
\mathds{E}\big[u^{i}(\mathrm{w}^{i,\,\Pi^{i,\,*}}_{T})\big]\geq-\frac{1}{\delta^{i}}\,\mathds{E}\left[\exp\!\left(\frac{\delta^{i}\theta^{i}}{2}\big(\pi^{i,\,*}_{T}\big)^{2}\right)\exp\!\left(-\delta^{i}W^{i,\,\Pi^{i,\,*}}_{T}\right)\right]\geq-\frac{1}{\delta^{i}}\,\mathds{E}\left[\exp\!\left(-\delta^{i}W^{i,\,\Pi^{i,\,*}}_{T}\right)\right]
\end{equation*}

\par Since $X^{i,\,*} \in \mathcal{A}_{s}$ is a best-response for investor $i$, it follows from the second inequality above that we have
\begin{equation*}
\mathds{E}\big[u^{i}(\mathrm{w}^{i,\,\Pi^{i,\,*}}_{T})\big]\geq J^{i}\!\left(s,y^{i};X^{-i}\right)
\end{equation*}

\par Given the equation above and the equivalence result established in \textcolor{azul-pesc}{\Cref{vinv}}, it follows that we have $V^{i}(s,z^{i};X^{-i}) \leq \mathds{E}\big[u^{i}(\mathrm{w}^{i,\,\Pi^{i,\,*}}_{T})\big]$, which completes the proof.
\end{proof}

\subsection{Constant Volatility Paradigm}

\par In order to pin down a unique best-response and analytically solve for the resulting Markov\textendash Nash equilibrium, we require further structure on the local volatility function $\sigma$, as in general one may obtain optimal strategies (portfolios) which are processes with unbounded variation, $\mathds{P}$-almost surely, owing to the singular nature of an investor's best-response problem, see \cite[Remark 2.4]{lions2007large} for further detail. In subsequent discussion, we therefore focus on the particular case when the local volatility function $\sigma(\cdot)$ is a (positive) constant function, as in \cite{bachelier1900theorie}. Note that a (positive) constant function vacuously satisfies the hypotheses of \textcolor{azul-pesc}{\Cref{siglip}}, and hence the preceding analysis carries over verbatim.

\begin{assumption}\label{sigcon}
\par The function $\sigma: \mathds{R}\rightarrow\mathds{R}_{+}$ is constant, that is $\sigma\!\left(S_{t}\right) = \sigma > 0$.
\end{assumption}

\par In the following lemma we establish the uniqueness of the optimal auxiliary control of investor $i$, given an admissible strategy $X^{-i} \in \mathcal{A}_{s}$ for investor $-i$, under a constant volatility paradigm and the assumption of c\`{a}dl\`{a}g strategies. The proof of the lemma relies on standard arguments which crucially depend on the fact that, for a given initial time $s \in [0,T]$, the class of admissible auxiliary controls $\mathcal{A}^{a}_{s}$, defined as in \textcolor{azul-pesc}{\Cref{auxadmcon}}, is non-redundant (in an appropriate sense; see below for a precise statement). This fact is instrumental in proving that the function $H^{i}$ as defined in (\ref{auxvf}) is strictly concave with respect to $\pi^{i}$, which in turn implies the uniqueness of the optimal auxiliary control.

\begin{lemma}\label{uniq}
\par Consider a fixed initial time $s \in [0,T]$, and suppose that the local volatility function satisfies \textcolor{azul-pesc}{\Cref{sigcon}}. Given a deterministic initial auxiliary state $z^{i} \in \mathds{R}^{\lvert Z^{i} \rvert}$, and an admissible strategy $X^{-i} \in \mathcal{A}_{s}$ for investor $-i$, the optimal auxiliary policy of investor $i$ is unique within the class of c\`{a}dl\`{a}g strategies, that is, if there exist $\Pi^{i,\,*},\tilde{\Pi}^{i,\,*} \in \mathcal{A}^{a}_{s}$ such that
\begin{enumerate}[(i)]
\item $\Pi^{i,\,*},\tilde{\Pi}^{i,\,*} \in \mathcal{A}^{a}_{s}$ are c\`{a}dl\`{a}g processes.
\item $V^{i}\!\left(s,z^{i};X^{-i}\right) = \mathds{E}\big[u^{i}(\mathrm{w}^{i,\,\Pi^{i,\,*}}_{T})\big]=\mathds{E}\big[u^{i}(\mathrm{w}^{i,\,\tilde{\Pi}^{i,\,*}}_{T})\big]$
\end{enumerate}
\par Then, the stochastic processes $\Pi^{i,\,*}$ and $\tilde{\Pi}^{i,\,*}$ are modifications of each other.
\end{lemma}

\begin{proof}
\par We first show that the subset of the class of admissible auxiliary controls which satisfy c\`{a}dl\`{a}g property is non-redundant, conditional on the local volatility function $\sigma$ satisfying \textcolor{azul-pesc}{\Cref{sigcon}}. That is, given two admissible auxiliary controls $\Pi^{i}$ and $\tilde{\Pi}^{i} \in \mathcal{A}^{a}_{s}$, we prove that $\mathrm{w}^{i,\,\Pi^{i}}_{T} = \mathrm{w}^{i,\,\tilde{\Pi}^{i}}_{T}$, $\mathds{P}$-almost surely, if and only if the control processes $\Pi^{i}, \tilde{\Pi}^{i}$ are modifications of each other. To this end, recall the dynamics of $\mathrm{w}^{i}$ given by (\ref{astate}), and note that given $z^{i} \in \mathds{R}^{\lvert Z^{i}\rvert} $, and an admissible strategy $X^{-i} \in \mathcal{A}_{s}$ for investor $-i$ we have
\begin{equation}\label{uni1}
\mathrm{w}^{i,\,\Pi^{i}}_{T} - \mathrm{w}^{i,\,\tilde{\Pi}^{i}}_{T} = -\,\theta^{-i}\!\!\!\int^{T}_{s}\!\!\!\left(\pi^{i}_{t}-\tilde{\pi}^{i}_{t}\right)x^{-i}_{t}dt\,+\,\sigma\!\!\int^{T}_{s}\!\!\!\left(\pi^{i}_{t}-\tilde{\pi}^{i}_{t}\right)dB_{t} = 0, \ \mathds{P}-\text{almost surely}
\end{equation}

\par It is immediate from the above that if the processes $\Pi^{i}, \tilde{\Pi}^{i}$ are modifications of each other then $\mathrm{w}^{i,\,\Pi^{i}}_{T} = \mathrm{w}^{i,\,\tilde{\Pi}^{i}}_{T}$, $\mathds{P}$-almost surely. Further, if we denote the quadratic variation process of a continuous semimartingale $M$ as $\langle M\rangle$, it follows from (\ref{uni1}) that we have
\begin{equation*}
\left\langle\mathrm{w}^{i,\,\Pi^{i}} - \mathrm{w}^{i,\,\tilde{\Pi}^{i}}\right\rangle_{T}= \sigma^{2}\!\!\int^{T}_{s}\!\!\!\left(\pi^{i}_{t}-\tilde{\pi}^{i}_{t}\right)^{2}\!dt = 0, \ \mathds{P}-\text{almost surely}
\end{equation*}

\par Given the equation above, it follows that if we have $\mathrm{w}^{i,\,\Pi^{i}}_{T} = \mathrm{w}^{i,\,\tilde{\Pi}^{i}}_{T}$, $\mathds{P}$-almost surely, then the control processes $\Pi^{i}$ and $\tilde{\Pi}^{i}$ coincide $\mathds{P} \otimes \lambda_{[s,\,T]}$ almost everywhere, where $\lambda_{[s,\,T]}$ denotes the restriction of Lebesgue measure defined on $\mathds{R}$ over the interval $[s,T]$. Thus, for every $t \in [s,T)$, and given $\epsilon > 0$, we can find a decreasing sequence $\{t_{n}\}_{n\, \in\, \mathds{N}} \in (t,t+\epsilon\wedge T]$, with $\lim_{n\rightarrow \infty}t_{n} = t$, such that $\pi(\omega,t_{n}) = \tilde{\pi}(\omega,t_{n})$, for every element of the sequence, $\mathds{P}$-almost surely. Moreover, since the admissible control processes $\Pi^{i}$ and $\tilde{\Pi}^{i}$ are c\'{a}dl\'{a}g, this in turn implies that $\pi(\omega,t) = \tilde{\pi}(\omega,t)$, $\mathds{P}$-almost surely. 

\par Also, in view of \textcolor{azul-pesc}{\Cref{auxadmcon}}, the control processes $\Pi^{i}$ and $\tilde{\Pi}^{i}$ are left-continuous at $T$. We can then follow a similar argument as above to deduce that $\pi(\omega,T) = \tilde{\pi}(\omega,T)$, $\mathds{P}$-almost surely, and hence, if we have $\mathrm{w}^{i,\,\Pi^{i}}_{T} = \mathrm{w}^{i,\,\tilde{\Pi}^{i}}_{T}$, $\mathds{P}$-almost surely, it follows that the auxiliary control processes $\Pi^{i}$ and $\tilde{\Pi}^{i}$ are modifications of each other. Thus, we conclude that the class of admissible auxiliary controls $\mathcal{A}^{a}_{s}$ is non-redundant.

\par The non-redundancy of the class of admissible auxiliary controls $\mathcal{A}^{a}_{s}$ is a sufficient condition for the strict concavity of the function $H^{i}$, as defined in (\ref{auxvf}). To see this, consider two admissible auxiliary controls $\Pi^{i}, \tilde{\Pi}^{i} \in \mathcal{A}^{a}_{s}$, which are not modifications of each other. For a given scalar $\xi \in (0,1)$, note that the stochastic process $\Pi^{i,\,\xi}$ defined as $\Pi^{i,\,\xi} = \xi\Pi^{i} + (1-\xi)\tilde{\Pi}^{i}$ is an admissible auxiliary control, by virtue of convexity of the state space $\Delta$. Recalling the dynamics of the auxiliary state process (\ref{astate}), it follows that given $z^{i}\in \mathds{R}^{\lvert Z^{i}\rvert}$, and $X^{-i} \in \mathcal{A}_{s}$ we have
\begin{align*}
u^{i}\!\left(\mathrm{w}^{i,\,\Pi^{i,\,\xi}}_{T}\right)\,&=\, -\frac{1}{\delta^{i}}\exp\!\left(-\delta^{i}\mathrm{w}^{i}_{0}+\delta^{i}\theta^{-i}\!\!\int^{T}_{s}\!\!\!\pi^{i,\,\xi}_{t}x^{-i}_{t}dt-\delta^{i}\sigma\!\!\int^{T}_{s}\!\!\!\pi^{i,\,\xi}_{t}dB_{t}\right)\\
&=\,-\frac{1}{\delta^{i}}\exp\!\left(\xi\left(-\delta^{i}\mathrm{w}^{i,\,\Pi^{i}}_{T}\right)+(1-\xi)\left(-\delta^{i}\mathrm{w}^{i,\,\tilde{\Pi}^{i}}_{T}\right)\right)
\end{align*}

\par Note that the second equality above follows as a consequence of the definition of the stochastic process $\Pi^{i,\,\xi}$. Given that the control processes $\Pi^{i}, \tilde{\Pi}^{i} \in \mathcal{A}^{a}_{s}$ are not modifications of each other by assumption, and in view of the fact that the class of admissible auxiliary controls $\mathcal{A}^{a}_{s}$ is non-redundant as shown earlier, it follows that we have $\mathrm{w}^{i,\,\Pi^{i}}_{T} \neq \mathrm{w}^{i,\,\tilde{\Pi}^{i}}_{T}$, $\mathds{P}$-almost surely. This fact in conjunction with Jensen's inequality then implies that we have
\begin{align*}
u^{i}\!\left(\mathrm{w}^{i,\,\Pi^{i,\,\xi}}_{T}\right)\,&>\,\xi\left(-\frac{1}{\delta^{i}}\exp\!\left(-\delta^{i}\mathrm{w}^{i,\,\Pi^{i}}_{T}\right)\right)+(1-\xi)\left(-\frac{1}{\delta^{i}}\exp\!\left(-\delta^{i}\mathrm{w}^{i,\,\tilde{\Pi}^{i}}_{T}\right)\right)\\
&=\, \xi\, u^{i}\!\left(\mathrm{w}^{i,\,\Pi^{i}}_{T}\right)+(1-\xi)\, u^{i}\!\left(\mathrm{w}^{i,\,\tilde{\Pi}^{i}}_{T}\right)
\end{align*}

\par The strict concavity of the function $H^{i}$ follows from the above upon taking expectations, whereby the uniqueness of the optimal auxiliary control is immediate in view of (\ref{auxvf}).
\end{proof}

\par With the above results at our disposal, we are in a position to solve explicitly for the (unique) Markov\textendash Nash equilibrium in closed-form for the constant volatility case that involves solving a coupled system of ordinary differential equations associated with individual best-responses which can be achieved through routine computation as shown in the proposition below. 

\par Before we state and prove the result, we draw attention to the fact that the statement of the proposition below relies upon three prerequisites, which we discuss in detail for better understanding. To this end, we let $\chi$ and $\varphi$ denote positive scalars defined as
\[\chi = \sqrt{\frac{\delta^{1}\,\delta^{2}\,}{\theta^{1}\,\theta^{2}\,}}\left(\sigma\right)^{2}\  \text{and}\ \ \varphi = \sqrt{\frac{\delta^{1}\,\theta^{1}\,}{\delta^{2}\,\theta^{2}\,}}\]

\par The first condition in the statement of \textcolor{azul-pesc}{\Cref{mnec}} below requires that the two investors not be \textit{too similar} in terms of their price impact and risk-aversion. It is important to note that we impose this condition as it allows us to focus attention on the general case of asymmetric equilibrium. We can still derive analytical expressions for Markov\textendash Nash equilibrium strategies of the two large investors when this condition does not hold. However, these offer little additional qualitative insight and the computation involved is tedious hence we omit this case.

\par The second condition requires that the the initial holding of at least one of the two large investors be non-zero, or equivalently, since $\pi^{i}_{s^{-}} = 0$ for $i \in \mathcal{I}$ by assumption, we require that the initial trade of at least one large investor be non-zero. This condition has a direct bearing on the uniqueness of Markov\textendash Nash equilibrium, as upon relaxing this condition we obtain an additional, though still deterministic, Markov\textendash Nash equilibrium wherein each investor starts with zero holdings in the risky asset and does not trade subsequently. 

\par However, it is straightforward to see that this \textit{zero holding} equilibrium is \textit{unstable} in the sense of \cite{kohlberg1986strategic}, as even a slight unilateral deviation by either large investor suffices to induce off equilibrium play subsequently.

\par  Finally, the third condition requires that the two large investors are always within their \textit{control} region. This requirement implies that the portfolio constraint is never binding for either of the two large investors. In other words, the third condition allows us to focus on the case when liquidity constraints arise solely on account of the price impact of the two large investors and the resulting effect on the trading behaviour of the two large investors.

\par If this condition is relaxed, there is a possibility that the strategies of the two players eventually enter the \textit{continuation} region whereby the market becomes illiquid and further trading becomes infeasible. Thus, subject to the three conditions above, the proposition below derives the unique Markov\textendash Nash equilibrium.

\begin{proposition}\label{mnec}
\par Consider a fixed initial time $s \in [0,T)$, and suppose that the local volatility function satisfies \textcolor{azul-pesc}{\Cref{sigcon}}. Given fixed initial states $y^{1}, y^{2} \in \mathds{R}^{\left\lvert Y^{i}\right\rvert}$, suppose the following conditions hold
\begin{enumerate}[(i)]
\item $ \varphi \neq 1$
\item $ \left\lvert\, \pi^{1}_{s}\,\right\rvert \bigvee \left\lvert\,\pi^{2}_{s}\,\right\rvert \neq 0$
\item $\exp\!{\left(\chi T\right)}\left(\left(1+\varphi\right)\left\lvert\, \pi^{1}_{s}\,\right\rvert +\left(1+\frac{1}{\varphi}\right) \left\lvert\,\pi^{2}_{s}\,\right\rvert\,\right)\leq \left\vert\,\overline{\pi}\,\right\rvert\, \bigwedge\, \left\vert\,\underline{\pi}\,\right\rvert $
\end{enumerate}

\par Then, the unique Markovian\textendash Nash equilibrium portfolio processes $\left(\Pi^{1,*},\Pi^{2,*}\right)$ of the two investors are deterministic and are given by
\begin{align*}
&\begin{cases}\pi^{1,*}_{t}\,=\, \pi^{1}_{s}\cosh\!{\left(\chi t\right)}\,+\, \left(\pi^{2}_{s}/\varphi\right)\sinh\!{\left(\chi t\right)},\ t\in [s,T]\\
\pi^{1,*}_{s}-\,\pi^{1,*}_{s^{-}}\, =\, \pi^{1}_{s}\\
\pi^{1,*}_{s^{-}}\, =\, 0\end{cases}\\
&\begin{cases}\pi^{2,*}_{t}\,=\, \varphi\pi^{1}_{s}\sinh\!{\left(\chi t\right)}\,+\, \pi^{2}_{s}\cosh\!{\left(\chi t\right)},\ t\in [s,T]\\
\pi^{2,*}_{s}-\,\pi^{2,*}_{s^{-}}\, =\, \pi^{2}_{s}\\
\pi^{2,*}_{s^{-}}\, =\, 0\end{cases}
\end{align*}
\end{proposition}

\begin{proof}
\par Suppose that conditions $\mathrm{(i)}$ and $\mathrm{(ii)}$ in the statement of the proposition hold, then it can be checked from routine computation that $\Pi^{1,*}$ and $\Pi^{2,*}$, as defined in the statement of the proposition, constitute a solution to the following system of ordinary differential equations, for $t\in [s,T]$ given deterministic initial values $\pi^{1}_{s}$ and $\pi^{2}_{s}$
\begin{equation*}
\setlength{\jot}{5pt}
\begin{aligned}
&\frac{d\pi^{1}_{t}}{dt}\, =\,-\frac{\delta^{2}\left(\sigma\right)^{2}}{\theta^{1}}\,\pi^{2}_{t}\\
&\frac{d\pi^{2}_{t}}{dt}\, =\,-\frac{\delta^{1}\left(\sigma\right)^{2}}{\theta^{2}}\,\pi^{1}_{t} \\
&\ \pi^{i}_{s} - \pi^{i}_{s^{-}}\,=\,\pi^{i}_{s},\ 
 \pi^{i}_{s^{-}}\,=\, 0,\ \forall\, {i \in \mathcal{I}}
\end{aligned}
\end{equation*}

\par Further, in view of condition $\mathrm{(iii)}$, the stochastic processes $X^{1,*}$ and $X^{2,*}$ defined as $x^{i,\,*}_{t}(\omega) = -\partial_{+} \pi^{i,\,*}_{t}(\omega)/\partial t$, for $s\leq t \leq T$, $i \in \mathcal{I}$ are admissible strategies in the sense that $X^{1,*}, X^{2,*} \in \mathcal{A}_{s}$. Hence, we can invoke \textcolor{azul-pesc}{\Cref{opaux}} in conjunction with condition $\mathrm{(iii)}$ in the statement of the proposition to ascertain that the stochastic process $\Pi^{1,*}$ is an optimal auxiliary control for investor $1$ given $X^{2,*}$, while the stochastic process $\Pi^{2,*}$ constitutes an optimal auxiliary control for investor $2$ given where $X^{1,*}$ when the local volatility function satisfies \textcolor{azul-pesc}{\Cref{sigcon}}.

\par Thus, as a result of the above and by way of \textcolor{azul-pesc}{\Cref{trajeq1}}, it follows that the stochastic processes $X^{1,*}$ and $X^{2,*}$ are mutual best-responses and hence constitute a Markovian\textendash Nash equilibrium, as defined in \textcolor{azul-pesc}{\Cref{mneq}}, of the Merton\textendash Cournot stochastic differential game. The uniqueness result follows as an immediate consequence of \textcolor{azul-pesc}{\Cref{trajeq1}}, \textcolor{azul-pesc}{\Cref{trajeq2}} and \textcolor{azul-pesc}{\Cref{uniq}}.
\end{proof}

\par Next, we discuss the qualitative implications of the Markov\textendash Nash Equilibrium trading strategies derived above. From \textcolor{azul-pesc}{\Cref{mnec}}, it follows that strategic competition for liquidity generates excessive trading compared to the canonical Merton model. To see this, note that if $\theta^{1}$ and $\theta^{2}$ were identically zero, the price process of the risky asset becomes a martingale which implies that the returns associated with the risky asset are zero and hence the investors have no incentive to trade in the risky asset. 

\par This result appears strikingly paradoxical at first since price impact can be interpreted as a form of transaction cost, and hence it should serve to discourage investors from trading relative to a frictionless market. However, unlike a monopolistic framework, a strategic framework leads to the imposition of an externality on the investment opportunity set of a given large investor on account of the price impact of the other large investor(s), which in turn compels an investor to optimally adjust  risky asset holdings in response.

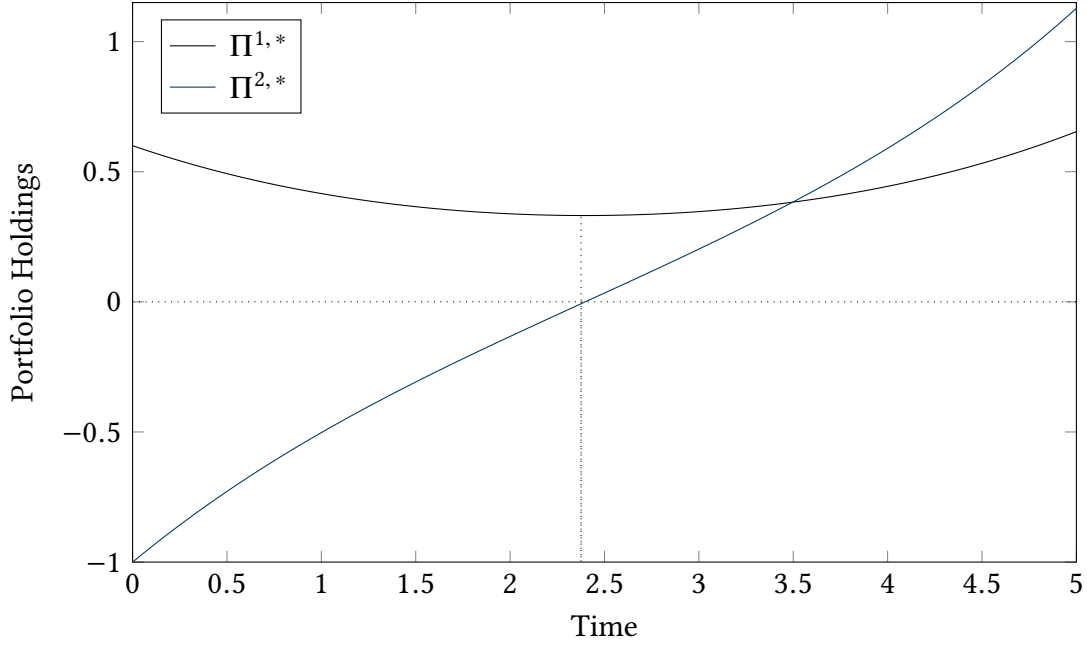
\begin{figure}[h]
\centering
\begin{tikzpicture}
\begin{axis}[
    x tick label style={/pgf/number format/1000 sep=, legend pos=north west },
    ylabel={Portfolio Holdings},
    xlabel={Time},
    xmin=0, xmax=5,
    ymin=-1, ymax=1.15,
    xtick={0,0.5,1,1.5,2,2.5,3,3.5,4,4.5,5},
    ymajorgrids=false,
    grid style=none,
]
\addplot[
       color=black,
       solid,
       domain=0:5,
       samples=1500,
       mark=none,
       smooth
       ] 
       {(0.6*cosh(0.5*x))-(0.5*(sinh(0.5*x)))};
     \addlegendentry{$\ \Pi^{1,\,*}$}
\addplot[
       color=azul-pesc,
       solid,
       domain=0:5,
       samples=1500,
       mark=none,
       smooth
       ]      
       {(1.2*(sinh(0.5*x)))-cosh(0.5*x)};
    \addlegendentry{$\ \Pi^{2,\,*}$}
\addplot[
    color=black,
    dotted,
    mark=,
    ]
    coordinates {
    (0,0)(0.5,0)(1,0)(1.5,0)(2,0)(2.5,0)(3,0)(3.5,0)(4,0)(4.5,0)(5,0)
    };
\addplot[
    color=black,
    dotted,
    mark=,
    ]
    coordinates {
    (2.376,0)(2.376,-0.5)(2.376,-1)(2.376,0.331)
    };
\end{axis}
\end{tikzpicture}
\floatfoot{Note: Illustration of Markov\textendash Nash equilibrium portfolio holdings of the two large investors corresponding to parameter values $\pi^{1}_{0} = 0.6$, $\pi^{2}_{0} = -1$, $\chi = 0.5$, $\varphi = 2$. For these specific parameter values, the optimal strategy for the second investor is to always buy the asset, while the optimal strategy for the first investor is to initially sell the asset until such time as the holding of the second investor is zero, indicated by the dashed lines, and to buy the asset subsequently.}
\caption{Markov\textendash Nash Equilibrium Portfolio Holdings}
\label{fig:2}
\end{figure}

\par Given the form of the equilibrium trading strategies derived in \textcolor{azul-pesc}{\Cref{mnec}}, the intuition underlying excessive trading on part of the investors is immediate in the instance when the initial holding (trade) of the two investors have the same sign. \textcolor{azul-pesc}{\Cref{fig:2}} provides an illustration of excessive trading in equilibrium by the two investors when their respective initial trades differ in their signs. For the specified parameter values in \textcolor{azul-pesc}{\Cref{fig:2}}, investor $1$ initially sells the asset in equilibrium since the initial holding of the second investor is negative, which generates a negative externality on the investment opportunity set of the first investor.

\par On the other hand, for the same specified parameter values investor $2$ always buys the asset in equilibrium. This is because the holding of the first investor always remains positive thereby generating a positive externality on the investment opportunity set of the second investor. Eventually the holding of the second investor becomes zero at time $t^{c} = (1/\chi)\,\mathrm{arctanh}(\pi^{2}_{0}/\varphi\pi^{1}_{0})$ at which point investor $1$ switches from selling to buying the asset given that subsequent to time $t^{c}$ the externality imposed on the investment opportunity set of investor $1$ by investor $2$ becomes positive. 

\par Interestingly while the nature of the externality imposed on the investment opportunity set of a rival and hence the direction of the rival's trading in the risky asset is determined solely by the sign of an investor's holding in the risky asset, the rate at which the rival trades depends on whether the investor is buying or selling the risky asset, that is, whether the investor is increasing or reducing the risky asset holding. For example, in the setup of \textcolor{azul-pesc}{\Cref{fig:2}} investor $2$ buys the asset at a slower rate when investor $1$ is selling the asset compared to when investor $1$ switches to buying the asset.

\section{Concluding Remarks}\label{conc}

\par In the present work, we analyzed an augmented version of Merton's canonical portfolio choice problem with two large investors whose trading influences the price of the underlying traded financial asset by modelling the resulting strategic interaction among investors as a non-zero sum singular stochastic differential game in the spirit of dynamic Cournot duopoly. We established an equivalence result for an investor's best-response problem and an auxiliary classical optimal control problem by extending the diffeomorphic flow method introduced in \cite{lasry2000classe}. We further showed that the optimal trajectories of the two control problems are indistinguishable under certain regularity conditions, and under the assumption of constant asset price volatility, we showed that the unique Markov\textendash Nash equilibrium is deterministic, which allowed us to reflect on the importance of imperfect competition in rationalizing the excessive trading puzzle.  

\par Our work also presents interesting avenues for future research which extend and complement our analysis. For instance, given the parsimonious additive permanent price impact framework we consider, it is straightforward to generalize most results presented here to the case of a multiple (but finite) large investors. A more involved but interesting generalization would analyze strategic competition for liquidity in a mean-field framework with major and minor players. It should also be noted that in the present work the assumption of exponential utility plays a central role and as a result it is natural to ask whether our approach extends to other commonly used utility functions such as the Constant Absolute Risk-Aversion (CARA) utility function, for example. Economically, a strategic portfolio choice framework which allows investors to trade using market as well as limit orders provides a greater degree of verisimilitude and remains an interesting challenge for future research.

\newpage
\raggedright
\bibliographystyle{rss}
\pagestyle{plain}
\bibliography{Reference}

\newpage

\justifying

\begin{appendices}

\titleformat{\section}[block]{\vspace{1em}\large\centering\scshape\bfseries}{Appendix \thesection.}{0.5em}{} 

\section{Data Appendix}

\par The estimates reported in \textcolor{azul-pesc}{\Cref{fig:1}} are sourced from the December 9, 2022 release of the Flow of Funds Accounts for the United States released by the Federal Reserve Board. Specifically, we refer Table L.223 of the Flow of Funds Accounts which documents the market value levels of corporate equity holdings for different investor categories. The estimates for the share of corporate equity held by Mutual Funds follow from Table L.223, by clubbing together the values reported for long-term mutual funds with close-ended mutual funds as well as exchange-traded funds. All ratios are computed with respected to the market value level of corporate equity held by all sectors, which is again directly reported in Table L.223. As detailed in the financial accounts guide, the value of holdings for the household sectors is computed as a residual by the Federal Reserve Board and thus includes the share of corporate equity owned by nonprofit organizations in addition to preferred stock and closely held corporate equity.

\par It is straightforward to factor out the value of closely held corporate equity as these estimates are reported by the Federal Reserve Board in Table L.223 itself. While, although the Federal Reserve Board publishes a series on financial assets held by private nonprofit organizations and charitable organizations respectively, these club the corporate equity holdings with mutual fund shares and hence serve as noisy estimates at best. Moreover, in order to eliminate the value of preferred stock, \cite{french2008presidential} relies on proprietary data obtained though personal correspondence from Standard and Poor, which unfortunately is neither publicly available nor deducible from the estimates presented in \cite{french2008presidential}. Thus, in view of these limitations, our estimates of the share of corporate equity held directly by households represents an upper bound, which nevertheless are fairly consistent with the values reported in \cite{french2008presidential} and \cite{stambaugh2014presidential}.

\section{Technical Appendix}

\subsection{Proof of \textcolor{azul-pesc}{\Cref{linpi}}}

\begin{proof}
\par In the proof, we limit attention to trading strategies defined over $\left[0,T\right]$, without loss of generality. We first show that a linear price impact function implies no dynamic arbitrage. To this end, we suppose $\kappa^{i}(x^{i}_{t}) = -\theta^{i}x^{i}_{t}$, where $\theta^{i} > 0$, and we consider a round-trip trade $\left\{x^{i}_{t}\right\}_{t\, \in\, [0,T]}$. It then follows from (\ref{portval}) that we have
\begin{equation*}
    \mathds{E}\left[W^{i}_{T}\right]\,=\, W^{i}_{0}\,+\,\mathds{E}\left[\int^{T}_{0}\!\! \pi^{i}_{t}\,\kappa^{-i}\left(x^{-i}_{t}\right)dt\right]\,+\,\mathds{E}\left[\int^{T}_{0}\!\!\pi^{i}_{t}\,\kappa^{i}\left(x^{i}_{t}\right)dt\right]
\end{equation*}

\par By virtue of the fact that $\kappa^{i}$ is linear, we can rewrite the above equation as
\begin{equation*}
    \mathds{E}\left[W^{i}_{T}\right]\,=\, W^{i}_{0}\,+\,\mathds{E}\left[\int^{T}_{0}\!\!\pi^{i}_{t}\, \kappa^{-i}\left(x^{-i}_{t}\right)dt\right]\,+\,\mathds{E}\left[\int^{T}_{0}\!\!-\pi^{i}_{t}\,\theta^{i}x^{i}_{t}\,dt\right]
\end{equation*}

\par Recalling the fact that we have $d\pi^{i}_{t}\, =\, -x^{i}_{t}\,dt$ by definition, we can rewrite the third term on the right hand side of the above equation as follows
\begin{equation*}
\mathds{E}\left[\int^{T}_{0}\!\!-\pi^{i}_{t}\,\theta^{i}x^{i}_{t}\,dt\right]\,=\,\mathds{E}\left[\int^{T}_{0}\!\!\pi^{i}_{t}\,\theta^{i}\,d\pi^{i}_{t}\right]\,=\,\frac{\theta^{i}}{2}\,\mathds{E}\left[\left(\pi^{i}_{T}\right)^{2} - \left(\pi^{i}_{0}\right)^{2}\right]
\end{equation*}

\par The claim then follows since $\pi^{i}_{T} = \pi^{i}_{0}$ which in turn follows from the fact that $\left\{x^{i}_{t}\right\}_{t\, \in\, [0,T]}$ is a round-trip trade.

\par We now prove the converse; the absence of dynamic arbitrage opportunities implies that the price impact function must be linear. We first establish an ancillary result which states that if there are no dynamic arbitrage opportunities then the function $\kappa^{i}$ must be an odd function. To this end, we fix $x \in \mathds{R}\backslash\{0\}$, and consider the following deterministic trading strategy for investor $i$
\begin{equation*}
    x^{i}_{t}\ =\ \begin{cases}\ x,& t\, \in\, \left[0,\,\frac{T}{2}\right]\\ 
    -x,& t\, \in\, \left(\frac{T}{2},\,T\right]\end{cases}
\end{equation*}

\par It is immediate that the trading strategy defined above is a round-trip trade. Further, recalling the fact that $d\pi^{i}_{t}\, =\, -x^{i}_{t}\,dt$, it follows from the definition above that we have
\begin{equation*}
    \int^{T}_{0}\!\!\pi^{i}_{t}\kappa^{i}\left(x^{i}_{t}\right)dt= -x\kappa^{i}(x)\!\mathlarger{\int}^{\frac{\,T}{2}}_{0}\!\!\left(\int^{t}_{0}\!\!du\right)dt - x\kappa^{i}(-x)\!\mathlarger{\int}^{T}_{\frac{\,T}{2}}\!\!\left(\int^{\frac{\,T}{2}}_{0}\!\!du\right)dt+ x\kappa^{i}(-x)\!\mathlarger{\int}^{T}_{\frac{\,T}{2}}\!\!\left(\int^{t}_{\frac{\,T}{2}}du\right)dt 
\end{equation*}

\par In view of the above, the absence of dynamic arbitrage opportunities then implies that the following must hold
\begin{equation}\label{da1}
    \mathds{E}\left[\int^{T}_{0}\!\!\pi^{i}_{t}\,\kappa^{i}\left(x^{i}_{t}\right)dt\right]\, =\,-\frac{x\,T^{2}}{8}\left[\kappa^{i}(x)\,+\,\kappa^{i}(-x)\right]\,\leq\, 0
\end{equation}

\par The claim follows on replacing $x$ by $-x$ in (\ref{da1}). Next, we fix $\alpha$, $\beta > 0$, $\tau = \beta T/(\alpha+\beta)$ and consider the following deterministic trading strategy for investor $i$
\begin{equation*}
    x^{i}_{t}\ =\ \begin{cases}\ \alpha, & t\, \in\, \left[0,\,\tau\right]\\ 
    -\beta, & t\, \in\, \left(\tau,\,T\right]\end{cases}
\end{equation*}

\par It is straightforward to check that the trading strategy defined above is a round-trip trade. Also, it follows from the definition above that we have
\begin{equation*}
    \int^{T}_{0}\!\!\pi^{i}_{t}\,\kappa^{i}\left(x^{i}_{t}\right)dt\,=\, -\alpha\kappa^{i}(\alpha)\!\!\mathlarger{\int}^{\tau}_{0}\!\!\left(\int^{t}_{0}\!\!du\right)dt\, -\, \alpha\kappa^{i}\!\left(\beta\right)\!\!\mathlarger{\int}^{T}_{\tau}\!\!\left(\int^{\tau}_{0}\!\!du\right)dt\,+\,\beta\kappa^{i}\!\left(\beta\right)\!\!\mathlarger{\int}^{T}_{\tau}\!\!\left(\int^{t}_{\tau}\!\!du\right)dt 
\end{equation*}

\par The absence of dynamic arbitrage opportunities together with the equation above then implies that
\begin{equation}\label{lin1}
    \mathds{E}\left[\int^{T}_{0}\!\!\pi^{i}_{t}\,\kappa^{i}\left(x^{i}_{t}\right)dt\right]\, =\,-\frac{\alpha\beta\,T^{2}}{2(\alpha+\beta)^{2}}\left[\alpha\kappa^{i}(-\beta)\,+\,\beta\kappa^{i}(\alpha)\right]\,\leq\, 0
\end{equation}

\par Next, given $\alpha$ and $\beta$ as above, we let $\hat{\tau} = \alpha T / (\alpha + \beta)$ and consider the following deterministic strategy for investor $i$
\begin{equation*}
    x^{i}_{t}\ =\ \begin{cases} -\beta, & t\, \in\, \left[0,\,\hat{\tau}\right]\\ 
    \ \alpha, & t\, \in\, \left(\hat{\tau},\,T\right]\end{cases}
\end{equation*}

\par The trading strategy clearly defines a round-trip trade. Moreover, it follows from the definition above that we have
\begin{equation*}
    \int^{T}_{0}\!\!\pi^{i}_{t}\,\kappa^{i}\left(x^{i}_{t}\right)\,dt\,=\, \beta\kappa^{i}\!\left(-\beta\right)\!\!\mathlarger{\int}^{\hat{\tau}}_{0}\!\!\left(\int^{t}_{0}\!\!du\right)dt\, +\, \beta\kappa^{i}(\alpha)\!\!\mathlarger{\int}^{T}_{\hat{\tau}}\!\!\left(\int^{\hat{\tau}}_{0}\!\!du\right)dt\,-\, \alpha\kappa^{i}(\alpha)\!\!\mathlarger{\int}^{T}_{\hat{\tau}}\!\!\left(\int^{t}_{\hat{\tau}}\!\!du\right)dt 
\end{equation*}

\par The equation above in conjunction with the absence of dynamic arbitrage opportunities then implies that
\begin{equation}\label{lin2}
    \mathds{E}\left[\int^{T}_{0}\!\!\pi^{i}_{t}\,\kappa^{i}\left(x^{i}_{t}\right)\,dt\right]\, =\,\frac{\alpha\beta\,T^{2}}{2(\alpha+\beta)^{2}}\left[\alpha\kappa^{i}(-\beta)\,+\,\beta\kappa^{i}(\alpha)\right]\,\leq\, 0
\end{equation}

\par From (\ref{lin1}) and (\ref{lin2}), it then follows that for $\alpha$, $\beta > 0$, we have $\alpha\kappa^{i}(-\beta)+\beta\kappa^{i}(\alpha) = 0$, and hence that
\begin{equation*}
\alpha\kappa^{i}(-\beta)\,=\,-\beta\kappa^{i}(\alpha)
\end{equation*} 

\par Since the above must hold for all $\alpha$, $\beta > 0$, we fix $\alpha > 0$ and let $\beta = 1$. Recalling the fact that $\kappa^{i}$ must be an odd function then implies that we must have  
\begin{equation*}
\kappa^{i}(\alpha)\,=\,-\alpha\kappa^{i}(-1)\, =\, \alpha\kappa^{i}(1) 
\end{equation*}

\par It only remains to show that $\kappa^{i}(0) = 0$. To this end, we consider the following deterministic strategy for $i$th investor 
\begin{equation*}
    x^{i}_{t}\ =\ \begin{cases}-\kappa^{i}(0), & t \,\in\, \left[0,\, \frac{T}{3}\right]\\ 
    \ \ 0, & t \,\in\, \left(\frac{T}{3},\, \frac{2T}{3}\right] \\  \ \kappa^{i}(0), & t \,\in\, \left[\frac{2T}{3},\, T\right] \end{cases}
\end{equation*}

\par Since the trading strategy defined above is a round-trip trade, it follows that
\begin{align}\label{kap}
\begin{split}
    \int^{T}_{0}\!\!\pi^{i}_{t}\,\kappa^{i}\left(x^{i}_{t}\right)\,dt\,=\,&\ \kappa^{i}(0)\,\kappa^{i}\!\left(-\kappa^{i}(0)\right)\!\!\mathlarger{\int}^{\frac{\,T}{3}}_{0}\!\left(\int^{t}_{0}du\right)dt\, +\, \left(\kappa^{i}(0)\right)^{2}\!\!\mathlarger{\int}^{\frac{\,2T}{3}}_{\frac{\,T}{3}}\!\left(\int^{\frac{\,T}{3}}_{0}\!du\right)dt\\
    & +\,\kappa^{i}(0)\,\kappa^{i}\!\left(\kappa^{i}(0)\right)\!\!\mathlarger{\int}^{T}_{\frac{\,2T}{3}}\!\left(\int^{\frac{\,T}{3}}_{0}\!du\right)dt\, -\,\kappa^{i}(0)\,\kappa^{i}\!\left(\kappa^{i}(0)\right)\!\!\mathlarger{\int}^{T}_{\frac{\,2T}{3}}\!\left(\int^{t}_{\frac{\,2T}{3}}\!du\right)dt
\end{split}
\end{align}

\par Recall that absence of dynamic arbitrage implies that $\kappa^{i}$ must be an odd function, which in conjunction with (\ref{kap}) implies the following
\begin{equation*}
    \mathds{E}\left[\int^{T}_{0}\!\!\pi^{i}_{t}\,\kappa^{i}\left(x^{i}_{t}\right)\,dt\right]\, =\,\frac{\,T^{2}}{9}\left(\kappa^{i}(0)\right)^{2}\,\leq\, 0
\end{equation*}

\par It follows that $\kappa^{i}(0)$ must be zero, and hence the price impact function must be linear.
\end{proof}

\subsection{Proof of \textcolor{azul-pesc}{\Cref{exist1}}}

\begin{proof}
\par \textsc{Existence} - We prove existence by way of Euler approximation approach as in \cite{applebaum2009levy}. To this end, we construct a sequence of stochastic differential equations through discretization of (\ref{stated}) on successively finer grids. Subsequently, we show that the constructed sequence has at least one uniformly convergent subsequence. Finally, we establish that the limit process of this uniformly convergent subsequence is a solution to (\ref{stated}). In what follows, we simplify notation by assuming that $s=0$, without any loss of generality. 

\par We define a sequence of stochastic processes $\{Y^{i,\,n},\,n\in \mathds{N}\cup 0\}$, by setting $Y^{i,\,0}_{t} = Y^{i}_{0}$ for $t \in [0,T]$. Further, for $n \geq 1$ we let $Y^{i,\,n}_{0} = Y^{i}_{0}$, and introduce $k(n,t) = \lfloor t n\rfloor/n$, $k^{+}(n,t) = \left(\lfloor t n\rfloor/n + 1\right)\wedge T$. Then, given $X^{i}, X^{-i} \in \mathcal{A}_{0}$, we define $Y^{i,\,n}_{t}$ for $t \in \left(k(n,t),\, k^{+}(n,t)\right]$, recursively as follows
\[Y^{i,\,n}_{t}\, =\, Y^{i,\,n}_{k(n,\,t)}\, +\,\int\limits^{t}_{k(n,\,t)}\!\!\!\!a^{i}\left(Y^{i,\,n}_{k(n,\,u)}\right)x^{-i}_{u}\, du\,+\, \int\limits^{t}_{k(n,\,t)}\!\!\!\!b^{i}\left(Y^{i,\,n}_{k(n,\,u)}\right)x^{i}_{u}\,du \,+\, \int\limits^{t}_{k(n,\,t)}\!\!\!\!\mathrm{v}^{i}\left(Y^{i,\,n}_{k(n,\,u)}\right)dB_{u}\]

\par Note that the recursive form of the definition above enables us to rewrite the expression defining $Y^{i,\,n}_{t}$ equivalently as follows
\begin{equation*}
Y^{i,\,n}_{t}\, =\, Y^{i,\,n}_{0}\, +\, \int\limits^{t}_{0}\!\!a^{i}\left(Y^{i,\,n}_{k(n,\,u)}\right)x^{-i}_{u}\,du\, +\,\int\limits^{t}_{0}\!\!b^{i}\left(Y^{i,\,n}_{k(n,\,u)}\right)x^{i}_{u}\,du\,+\, \int\limits^{t}_{0}\!\!\mathrm{v}^{i}\left(Y^{i,\,n}_{k(n,\,u)}\right)dB_{u}
\end{equation*}

\par Next, we prove existence and uniqueness of a strong solution to the system of stochastic differential equations defined by (\ref{stated}) using a localization argument. To this end, we define $\mathds{F}$-stopping times $\hat{\tau}_{k,\,x^{i}}$, $\hat{\tau}_{k,\,x^{-i}}$, and $\hat{\tau}_{k}$ as follows
\begin{equation*}
\hat{\tau}_{k,\,x^{i}}\, =\, \inf\left\{t > s:\, \left\lvert x^{i}_{t}\right\rvert\, >\, k\right\},\ \ \hat{\tau}_{k,\,x^{-i}}\, =\, \inf\left\{t > s:\, \left\lvert x^{-i}_{t}\right\rvert\, >\, k\right\},\ \ \hat{\tau}_{k}\, =\, \hat{\tau}_{k,\,x^{i}}\,\wedge\,\hat{\tau}_{k,\,x^{-i}}
\end{equation*}

\par Further, given $X^{i}, X^{-i} \in \mathcal{A}_{0}$ and $Y^{i}_{0} \in \mathds{R}^{\left\lvert\, Y^{i}\,\right\rvert}$, it follows by definition of $Y^{i,\,n}$ that for $t \in [0,T]$ we have
\begin{align*}
    \mathds{1}_{\left\{t\,\leq\,\hat{\tau}_{k}\right\}}\left\lvert\,Y^{i,\,1}_{t} - Y^{i,\,0}_{t}\,\right\rvert\ \leq\,& \int\limits^{t}_{0}\!\! \mathds{1}_{\left\{u\,\leq\,\hat{\tau}_{k}\right\}}\left\lvert\, a^{i}\!\left(Y^{i,\,1}_{k(n,\,u)}\right)\,\right\rvert\left\lvert\, x^{-i}_{u}\,\right\rvert du\, +\, \int\limits^{t}_{0}\!\! \mathds{1}_{\left\{u\,\leq\,\hat{\tau}_{k}\right\}}\left\lvert\, b^{i}\!\left(Y^{i,\,1}_{k(n,\,u)}\right)\,\right\rvert\left\lvert\, x^{i}_{u}\,\right\rvert du \\ & +\, \left\lvert\,\int\limits^{t}_{0}\!\!\mathds{1}_{\left\{u\,\leq\,\hat{\tau}_{k}\right\}}\mathrm{v}^{i}\!\left(Y^{i,\,1}_{k(n,\,u)}\right)dB_{u}\,\right\rvert
\end{align*}

\par Given the above, it then follows in view of Cauchy\textendash Schwarz inequality and the definition of the $\mathds{F}$-stopping time $\hat{\tau}_{k}$ that we can find a positive constant $C$ (dependent on $k$) such that
\begin{align*}
    \mathds{E}\left[\,\sup\limits_{t\,\in\, [0,\,T]}\mathds{1}_{\left\{t\,\leq\,\hat{\tau}_{k}\right\}}\left\lvert\,Y^{i,\,1}_{t} - Y^{i,\,0}_{t}\,\right\rvert^{2}\right]\, \leq\ &\,  C\mathds{E}\left[\,\int\limits^{T}_{0}\!\! \mathds{1}_{\left\{u\,\leq\,\hat{\tau}_{k}\right\}}\left\lvert\, a^{i}\!\left(Y^{i,\,1}_{k(n,\,u)}\right)\,\right\rvert^{2}\!\!du\,\right]\\ 
    & +\,C\mathds{E}\left[\,\int\limits^{T}_{0}\!\! \mathds{1}_{\left\{u\,\leq\,\hat{\tau}_{k}\right\}}\left\lvert\, b^{i}\!\left(Y^{i,\,1}_{k(n,\,u)}\right)\,\right\rvert^{2}\!\!du\,\right]\\ 
    & +\, C\mathds{E}\left[\,\sup\limits_{t\,\in\,[0,\,T]}\,\left\lvert\,\int\limits^{t}_{0}\!\!\mathds{1}_{\left\{u\,\leq\,\hat{\tau}_{k}\right\}}\mathrm{v}^{i}\!\left(Y^{i,\,1}_{k(n,\,u)}\right)dB_{u}\right\rvert^{2}\right]
\end{align*}

\par We invoke Burkholder\textendash Davis\textendash Gundy inequality in view of the third term on the right-hand side of the equation above  which implies the existence of a positive constant $C$ (dependent on $k$) such that the equation above leads us to
\begin{align*}
    \mathds{E}\left[\,\sup\limits_{t\,\in\, [0,\,T]}\mathds{1}_{\left\{t\,\leq\,\hat{\tau}_{k}\right\}}\left\lvert\,Y^{i,\,1}_{t} - Y^{i,\,0}_{t}\,\right\rvert^{2}\right]\, \leq\ &\,  C\mathds{E}\left[\,\int\limits^{T}_{0}\!\! \mathds{1}_{\left\{u\,\leq\,\hat{\tau}_{k}\right\}}\left\lvert\, a^{i}\!\left(Y^{i,\,1}_{k(n,\,u)}\right)\,\right\rvert^{2}\!du\,\right]\\ 
    & +\,C\mathds{E}\left[\,\int\limits^{T}_{0}\!\! \mathds{1}_{\left\{u\,\leq\,\hat{\tau}_{k}\right\}}\left\lvert\, b^{i}\!\left(Y^{i,\,1}_{k(n,\,u)}\right)\,\right\rvert^{2}\!du\,\right]\\ 
    & +\, C\mathds{E}\left[\,\int\limits^{T}_{0}\!\!\mathds{1}_{\left\{u\,\leq\,\hat{\tau}_{k}\right\}}\left\lvert\,\mathrm{v}^{i}\!\left(Y^{i,\,1}_{k(n,\,u)}\right)\,\right\rvert^{2}du\right]
\end{align*}

\par Next, by definition of the $\mathds{F}$-stopping time $\hat{\tau}_{k}$ it follows that $\mathds{1}_{\left\{t\,\leq\,\hat{\tau}_{k}\right\}}\lvert\,\pi^{i,\,1}_{k(n,\,t)}\rvert$, $\mathds{1}_{\left\{t\,\leq\,\hat{\tau}_{k}\right\}}\lvert\,\pi^{-i,\,1}_{k(n,\,t)}\rvert$ are bounded above by $kT$, for $t \in [0,T]$. Thus, it follows in view of the definition of the vector-valued functions $a^{i}$, $b^{i}$, and $\mathrm{v}^{i}$, in conjunction with the fact that the local volatility function $\sigma$ satisfies \textcolor{azul-pesc}{\Cref{siglip}}, that we can find a positive constant $C$ (dependent on $k$) such that the equation above leads us to
\begin{equation}\label{ind1}
    \mathds{E}\left[\,\sup\limits_{t\,\in\, [0,\,T]}\mathds{1}_{\left\{t\,\leq\,\hat{\tau}_{k}\right\}}\left\lvert\,Y^{i,\,1}_{t} - Y^{i,\,0}_{t}\,\right\rvert^{2}\right]\, \leq \,  CT\left(1 + k^{2}T^{2}\right)
\end{equation}

\par Similarly, for $n \in \mathds{N}$, with $n > 1$, we again invoke the definition of the $\mathds{F}$-stopping time $\hat{\tau}_{k}$ along with Cauchy\textendash Schwarz inequality to assert that we can find a positive constant $C$ (dependent on $k$) such that
\begin{align*}
    \mathds{E}\left[\sup\limits_{t\,\in\, [0,\,T]}\mathds{1}_{\left\{t\,\leq\,\hat{\tau}_{k}\right\}}\left\lvert\,Y^{i,\,n}_{t} - Y^{i,\,n-1}_{t}\,\right\rvert^{2}\right] \leq &\, C\mathds{E}\left[\,\int\limits^{T}_{0}\!\! \mathds{1}_{\left\{u\,\leq\,\hat{\tau}_{k}\right\}}\left\lvert\, a^{i}\!\left(Y^{i,\,n}_{k(n,\,u)}\right) - a^{i}\!\left(Y^{i,\,n-1}_{k(n,\,u)}\right)\,\right\rvert^{2}\!du\,\right]\\ 
    & + C\mathds{E}\left[\,\int\limits^{T}_{0}\!\! \mathds{1}_{\left\{u\,\leq\,\hat{\tau}_{k}\right\}}\left\lvert\, b^{i}\!\left(Y^{i,\,n}_{k(n,\,u)}\right)-b^{i}\!\left(Y^{i,\,n-1}_{k(n,\,u)}\right)\,\right\rvert^{2}\!du\,\right]\\ 
    & + C\mathds{E}\left[\sup_{t\,\in\,[0,\,T]}\,\left\lvert\,\int\limits^{t}_{0}\!\!\mathds{1}_{\left\{u\,\leq\,\hat{\tau}_{k}\right\}}\left(\mathrm{v}^{i}\!\left(Y^{i,\,n}_{k(n,\,u)}\right)-\mathrm{v}^{i}\!\left(Y^{i,\,n-1}_{k(n,\,u)}\right)\right)dB_{u}\right\rvert^{2}\right]
\end{align*}

\par Next, we recall that by definition the vector-valued functions $a^{i}$ and $b^{i}$ are Lipschitz continuous, which in conjunction with Burkholder\textendash Davis\textendash Gundy inequality implies that the equation above leads us to
\begin{align*}
    \mathds{E}\left[\,\sup\limits_{t\,\in\, [0,\,T]}\mathds{1}_{\left\{t\,\leq\,\hat{\tau}_{k}\right\}}\left\lvert\,Y^{i,\,n}_{t} - Y^{i,\,n-1}_{t}\,\right\rvert^{2}\right]\, \leq\ &\, C\mathds{E}\left[\,\int\limits^{T}_{0}\!\! \mathds{1}_{\left\{u\,\leq\,\hat{\tau}_{k}\right\}}\left\lvert\,Y^{i,\,n}_{k(n,\,u)} - Y^{i,\,n-1}_{k(n,\,u)}\,\right\rvert^{2}\!du\,\right]\\ 
    & +\, C\mathds{E}\left[\,\int\limits^{T}_{0}\!\!\mathds{1}_{\left\{u\,\leq\,\hat{\tau}_{k}\right\}}\,\left\lvert\mathrm{v}^{i}\!\left(Y^{i,\,n}_{k(n,\,u)}\right)-\mathrm{v}^{i}\!\left(Y^{i,\,n-1}_{k(n,\,u)}\right)\right\rvert^{2}dB_{u}\right]
\end{align*}

\par Further, since the local volatility function $\sigma$ is Lipschitz continuous in view of {\textcolor{azul-pesc}{\Cref{siglip}}}, it follows from the definition of $\mathrm{v}^{i}$ and the $\mathds{F}$-stopping time $\hat{\tau}_{k}$ that the function $\mathrm{v}^{i}$ satisfies Lipschitz continuity on $[0,\hat{\tau}_{k}]$. This fact in conjunction with Tonelli's theorem implies that we can then find a positive constant $C$ (dependent on $k$) such that the equation above leads us to
\begin{equation}\label{ind2}
    \mathds{E}\left[\,\sup\limits_{t\,\in\, [0,\,T]}\mathds{1}_{\left\{t\,\leq\,\hat{\tau}_{k}\right\}}\left\lvert\,Y^{i,\,n}_{t} - Y^{i,\,n-1}_{t}\,\right\rvert^{2}\right]\, \leq \ C\!\!\int\limits^{T}_{0}\!\!\mathds{E}\left[\,\sup_{\hat{u}\,\in\,[0,\,u]} \mathds{1}_{\left\{\hat{u}\ \leq\ \hat{\tau}_{k}\right\}}\left\lvert\,Y^{i,\,n}_{\hat{u}} - Y^{i,\,n-1}_{\hat{u}}\,\right\rvert^{2}\right]du
\end{equation}

\par In view of (\ref{ind1}) and (\ref{ind2}), we then proceed by induction to deduce the following bound
\begin{equation}\label{ind3}
    \mathds{E}\left[\,\sup\limits_{t\,\in\, [0,\,T]}\mathds{1}_{\left\{t\,\leq\,\hat{\tau}_{k}\right\}}\left\lvert\,Y^{i,\,n}_{t} - Y^{i,\,n-1}_{t}\,\right\rvert^{2}\right]\, \leq \ \frac{C\left(1+k^{2}T^{2}\right)T^{n}}{n!}
\end{equation}

\par Note that the positive constant $C$ in the equation above is dependent on $k$. Having derived the bound above, we next establish that $\left\{Y^{i,\,n}_{t}\right\}_{n\,\in\,\mathds{N}}$ is convergent in $\mathds{L}^{2}$ for $t \in [0, \hat{\tau}_{k}]$. To this end, consider $m,n \in \mathds{N}$ and note that given the equation above it follows that we have
\begin{equation}\label{ind4}
    \mathds{1}_{\left\{t\,\leq\,\hat{\tau}_{k}\right\}}\left\lvert\,Y^{i,\,n}_{t} - Y^{i,\,m}_{t}\,\right\rvert_{\,\mathds{L}^{2}}\ \leq\ \mathlarger{\sum}\limits^{n}_{j\,=\ m\,+\,1}\mathds{1}_{\left\{t\,\leq\,\hat{\tau}_{k}\right\}}\left\lvert\,Y^{i,\,j}_{t} - Y^{i,\ j-1}_{t}\,\right\rvert_{\,\mathds{L}^{2}}\ \leq\ \mathlarger{\sum}\limits^{n}_{j\, =\ m\,+\,1}\frac{\left[\,C\left(1+k^{2}T^{2}\right)T^{j}\right]^{1/2}}{\left(j!\right)^{1/2}}
\end{equation}

\par It is immediate that the sum on the right-hand side of the equation above converges, which in turn implies that the sequence $\left\{Y^{i,\,n}_{t}\right\}_{n\,\in\,\mathds{N}} \in \mathds{L}^{2}$ is Cauchy for $t \in [0, \hat{\tau}_{k}]$. Further, since $\mathds{L}^{2}$ is a Banach space, it is complete and hence the sequence $\left\{Y^{i,\,n}_{t}\right\}_{n\,\in\,\mathds{N}}$ converges to some $Y^{i}_{t} \in \mathds{L}^{2}$ for $t \in [0, \hat{\tau}_{k}]$. To establish $\mathds{P}$-almost sure uniform convergence of the sequence to the limit process $Y^{i}$ obtained above, we note that given (\ref{ind3}), it follows in view of Chebyshev's inequality that we have
\begin{equation*}
    \mathds{P}\left(\,\sup\limits_{t\,\in\, [0,\,T]}\mathds{1}_{\left\{t\,\leq\,\hat{\tau}_{k}\right\}}\left\lvert\,Y^{i,\,n}_{t} - Y^{i,\,n-1}_{t}\,\right\rvert^{2}\,\geq\,\frac{1}{2^{n}}\,\right)\ \leq \ \frac{C\left(1+k^{2}T^{2}\right)(4T)^{n}}{n!}
\end{equation*}

\par Since the term on the right-hand side of the equation above converges to zero as $n \rightarrow \infty$, we invoke Borel\textendash Cantelli lemma \cite[Theorem 8.3.4]{dudley2018real} which then implies that the equation above leads us to 
\begin{equation*}
    \mathds{P}\left(\,\limsup_{n\, \rightarrow\, \infty}\,\sup\limits_{t\,\in\, [0,\,T]}\mathds{1}_{\left\{t\,\leq\,\hat{\tau}_{k}\right\}}\left\lvert\,Y^{i,\,n}_{t} - Y^{i,\,n-1}_{t}\,\right\rvert^{2}\,\geq\,\frac{1}{2^{n}}\,\right)\ = \ 0
\end{equation*}

\par Note that we can re-write the above equivalently as
\begin{equation*}
    \mathds{P}\left(\,\liminf_{n\, \rightarrow\, \infty}\,\sup\limits_{t\,\in\, [0,\,T]}\mathds{1}_{\left\{t\,\leq\,\hat{\tau}_{k}\right\}}\left\lvert\,Y^{i,\,n}_{t} - Y^{i,\,n-1}_{t}\,\right\rvert^{2}\,<\,\frac{1}{2^{n}}\,\right)\ = \ 1
\end{equation*}

\par It is then immediate from the above that the sequence $\{Y^{i,\,n}\}_{n\,\in\,\mathds{N}}$ is $\mathds{P}$-almost surely uniformly convergent on the interval $[0,\hat{\tau}_{k}]$ to the limit process $Y^{i} \in \mathds{L}^{2}$ obtained above.

\par Next, we show that the stochastic process $\{Y^{i}\}$ that is obtained as a uniform limit of the Euler approximation is indeed a solution to the stochastic differential equation (\ref{stated}) on the interval $[0,\hat{\tau}_{k}]$. To this end, we define a stochastic process $\hat{Y}^{i}$ for $t \in [0,\hat{\tau}_{k}]$ as follows
\begin{equation*}
\hat{Y}^{i}_{t}\, =\, Y^{i}_{0}\, +\, \int\limits^{t}_{0}\!\!a^{i}\left(Y^{i}_{u}\right)x^{-i}_{u}\,du\, +\,\int\limits^{t}_{0}\!\!b^{i}\left(Y^{i}_{u}\right)x^{i}_{u}\,du\,+\, \int\limits^{t}_{0}\!\!\mathrm{v}^{i}\left(Y^{i}_{u}\right)dB_{u}
\end{equation*}

\par Given the definition above, we appeal to an argument identical to the one used in deriving (\ref{ind2}) above, in order to assert the existence the of a positive constant $C$ (dependent on $k$) such that we have
\begin{equation*}
    \mathds{E}\left[\mathds{1}_{\left\{t\,\leq\,\hat{\tau}_{k}\right\}}\left\lvert\,\hat{Y}^{i}_{t} - Y^{i,\,n}_{t}\,\right\rvert^{2}\right] \leq C\!\!\int\limits^{T}_{0}\!\!\mathds{E}\left[ \mathds{1}_{\left\{u\ \leq\ \hat{\tau}_{k}\right\}}\left\lvert\,Y^{i}_{u} - Y^{i,\,n}_{u}\,\right\rvert^{2}\right]du \leq CT\sup\limits_{t\, \in\, [0,\,T]}\,\mathds{E}\left[\,\mathds{1}_{\left\{t\,\leq\,\hat{\tau}_{k}\right\}}\left\lvert\,Y^{i}_{t} - Y^{i,\,n}_{t}\,\right\rvert^{2}\right]
\end{equation*}

\par Further, it follows in view of the bound derived in (\ref{ind4}) that we have
\begin{equation*}
    \sup\limits_{t\, \in\, [0,\,T]}\,\mathds{E}\left[\,\mathds{1}_{\left\{t\,\leq\,\hat{\tau}_{k}\right\}}\left\lvert\,Y^{i}_{t} - Y^{i,\,n}_{t}\,\right\rvert^{2}\right]\ \ \leq\ \ \left(\,\mathlarger{\sum}\limits^{\infty}_{j\, =\ n\,+\,1}\frac{\left[\ C\left(1+k^{2}T^{2}\right)T^{j}\right]^{1/2}}{\left(j!\right)^{1/2}}\right)^{2}
\end{equation*}

\par Thus, we combine the inequalities above and recall that the right-hand side of the equation above converges to zero as $n \rightarrow \infty$ to establish that 
\begin{equation*}
    \limsup\limits_{n\, \rightarrow\, \infty}\,\sup\limits_{t\,\in\,[0,T]}\,\mathds{1}_{\left\{t\,\leq\,\hat{\tau}_{k}\right\}}\left\lvert\,\hat{Y}^{i}_{t} - Y^{i,\,n}_{t}\,\right\rvert_{\,\mathds{L}^2}\ = \ 0
\end{equation*}

\par It is immediate from the above and the uniqueness of limits that we have \[\mathds{P}\Big(\hat{Y}^{i}_{t}\, =\, Y^{i}_{t},\ t \in \left[0, \hat{\tau}_{k}\right]\Big)\, =\, 1\]

\par \textsc{Uniqueness} - Let $Y^{i}$ and $\tilde{Y}^{i}$ be two distinct stochastic processes which are solutions to the system of stochastic differential equations defined by (\ref{stated}). We again employ an argument identical to the one used in deriving (\ref{ind2}) above, in order to assert the existence the of a positive constant $C$ (dependent on $k$) such that we have
\begin{equation*}
    \mathds{E}\left[\,\sup\limits_{t\,\in\, [0,\,T]}\mathds{1}_{\left\{t\,\leq\,\hat{\tau}_{k}\right\}}\left\lvert\,Y^{i}_{t} - \tilde{Y}^{i}_{t}\,\right\rvert^{2}\right]\, \leq \ C\!\!\int\limits^{T}_{0}\!\!\mathds{E}\left[\,\sup_{\hat{u}\,\in\,[0,\,u]} \mathds{1}_{\left\{\hat{u}\ \leq\ \hat{\tau}_{k}\right\}}\left\lvert\,Y^{i}_{\hat{u}} - \tilde{Y}^{i}_{\hat{u}}\,\right\rvert^{2}\right]du
\end{equation*}

\par Given the equation above, it is then immediate from Gronwall's lemma that we have
\begin{equation*}
    \mathds{E}\left[\,\sup\limits_{t\,\in\, [0,\,T]}\mathds{1}_{\left\{t\,\leq\,\hat{\tau}_{k}\right\}}\left\lvert\,Y^{i}_{t} - \tilde{Y}^{i}_{t}\,\right\rvert^{2}\right]\, = \ 0
\end{equation*}

\par In view of the above, it is then immediate from Chebyshev's inequality that we have \[\mathds{P}\Big(Y^{i}_{t}\, =\, \tilde{Y}^{i}_{t},\ t \in \left[0, \hat{\tau}_{k}\right]\Big)\, =\, 1\]

\par \textsc{Non-Explosion} - Finally, we establish that the solution is non-explosive. To this end, we let $\zeta$ denote the lifetime of $\{Y^{i}\}$ defined as $\zeta = \liminf_{n\,\rightarrow \,\infty}\{\,t > 0: \lvert Y^{i}_{t}\rvert \geq  n \}$, where $n \in \mathds{N}$. It is immediate in view of the above that
\begin{equation*}
\zeta\, =\, \liminf_{k\, \rightarrow\, \infty}\,\hat{\tau}_{k}    
\end{equation*}

\par It then remains to show that $\zeta > T$, $\mathds{P}$-almost surely. To this end, we follow the approach employed in \cite{lan2014new}, whereby we consider the following test function defined on the positive real line
\begin{equation*}
    g(x)\, =\, \int\limits^{x}_{0}\!\!\!\frac{d\epsilon}{\epsilon+1}
\end{equation*}

\par The function $g(x)$ defined above is strictly increasing, and strictly concave on the positive real line. Moreover, since $\lim_{\,x \rightarrow 0}\, D_{1}g(x) = 1$, we can find a concave extension $\tilde{g}$ of $g$ on the entire real line, such that $\tilde{g}(x) = g(x)$, for all $x \geq 0$. Also, the second-order derivative $D^{2}_{1}\,\tilde{g}(x)$ of the extended function $\tilde{g}$, is a non-positive Radon measure in the generalized sense \cite[Appendix III]{revuz2013continuous}. Given the above, we let $\vartheta_{t} = \left\lvert Y^{i}_{t}\right\rvert^{2}$, and apply the It\^{o}\textendash Meyer\textendash Tanaka formula \cite[Theorem VI.1.5]{revuz2013continuous} to the function $\tilde{g}\!\left(\vartheta_{t\,\wedge\,\zeta}\right)$ so as to obtain the following
\begin{align*}
\tilde{g}\!\left(\vartheta_{t\,\wedge\,\zeta}\right)\, =\, &\ \tilde{g}(\vartheta_{0})\,+\, 2\!\!\int\limits_{0}^{t\,\wedge\,\zeta}\!\!\!D_{1}\tilde{g}(\vartheta_{u})\left\langle Y^{i}_{u}\,,\,a^{i}\!\left(Y^{i}_{u}\right)x^{-i}_{u}\right\rangle du\,
+\,2\!\!\int\limits_{0}^{t\,\wedge\,\zeta}\!\!\!D_{1}\tilde{g}(\vartheta_{u})\left\langle Y^{i}_{u}\,,\,b^{i}\!\left(Y^{i}_{u}\right)x^{i}_{u}\right\rangle du\\
& +\, \int\limits_{0}^{t\,\wedge\,\zeta}\!\!\!D_{1}\tilde{g}(\vartheta_{u})\left\lvert\mathrm{v}^{i}\left(Y^{i}_{u}\right)\right\rvert^{2}\!du\,
+\, 2\!\!\int\limits_{0}^{t\,\wedge\,\zeta}\!\!\!D_{1}\tilde{g}(\vartheta_{u})\left\langle Y^{i}_{u}\,,\mathrm{v}^{i}\left(Y^{i}_{u}\right)\right\rangle dB_{u}\,
+\,\frac{1}{2}\int_{\mathds{R}}\mathrm{L}^{q}_{t\,\wedge\,\zeta}(\vartheta)\,D^{2}_{1}\,\tilde{g}\left(dq\right)
\end{align*}

\par Note that the stochastic process $\mathrm{L}^{q}$ above denotes the local time process of $\vartheta$. Since, the local time is always non-negative by definition, and the Radon measure $D^{2}_{1}\,\tilde{g}$ is non-positive given that $\tilde{g}$ is concave, it follows that upon taking expectations we have
\begin{multline*}
\mathds{E}\left[g\!\left(\vartheta_{t\,\wedge\,\zeta}\right)\right]\leq g(\vartheta_{0})+ 2\mathds{E}\left[\int\limits_{0}^{t\,\wedge\,\zeta}\!\!\!D_{1}g(\vartheta_{u})\left\langle Y^{i}_{u}\,,\,a^{i}\!\left(Y^{i}_{u}\right)x^{-i}_{u}\right\rangle du\right]
 + 2\mathds{E}\left[\int\limits_{0}^{t\,\wedge\,\zeta}\!\!\!D_{1}g(\vartheta_{u})\left\langle Y^{i}_{u}\,,\,b^{i}\!\left(Y^{i}_{u}\right)x^{i}_{u}\right\rangle du\right]\\
+\mathds{E}\left[\mathlarger{\int}\limits^{t\,\wedge\,\zeta}_{0}\!D_{1}g(\vartheta_{u})\left\lvert\mathrm{v}^{i}\!\left(Y^{i}_{u}\right)\right\rvert^{2}\!du\right]
 +2\mathds{E}\left[\,\sup\limits_{0\,\leq\, \hat{t}\, \leq\, t}\ \left\lvert\, \int\limits_{0}^{\hat{t}}\!\mathds{1}_{\left\{\,u\,\leq\,\zeta\,\right\}}D_{1}{g}(\vartheta_{u})\left\langle Y^{i}_{u}\,,\mathrm{v}^{i}\left(Y^{i}_{u}\right)\right\rangle dB_{u}\,\right\rvert\ \right]
\end{multline*}

\par Given the equation above, we invoke Burkholder\textendash Davis\textendash Gundy inequality, by way of which we can find a positive constant $C$ (not dependent on $t$) such that the equation above leads us to
\begin{align*}
\mathds{E}\left[g\!\left(\vartheta_{t\,\wedge\,\zeta}\right)\right]\, \leq\,&\ C\,+\,  C\mathds{E}\left[\int\limits_{0}^{t\,\wedge\,\zeta}\!\!\!D_{1}\tilde{g}(\vartheta_{u})\left\langle Y^{i}_{u}\,,\,a^{i}\!\left(Y^{i}_{u}\right)x^{-i}_{u}\right\rangle du\right]\,
+\, C\mathds{E}\left[\int\limits_{0}^{t\,\wedge\,\zeta}\!\!\!D_{1}\tilde{g}(\vartheta_{u})\left\langle Y^{i}_{u}\,,\,b^{i}\!\left(Y^{i}_{u}\right)x^{i}_{u}\right\rangle du\right]\\
& +\, C\mathds{E}\left[\int\limits^{t\,\wedge\,\zeta}_{0}\!\!\!D_{1}g(\vartheta_{u})\left\lvert\mathrm{v}^{i}\!\left(Y^{i}_{u}\right)\right\rvert^{2}\!\! du\right]\,
+ \,C\mathds{E}\left[\,\int\limits_{0}^{t\,\wedge\,\zeta}\!\!\big\lvert D_{1}{g}(\vartheta_{u})\big\rvert^{2}\left\lvert Y^{i}_{u}\right\rvert^{2}\,\left\lvert\mathrm{v}^{i}\!\left(Y^{i}_{u}\right)\right\rvert^{2}\!\! du\,\right]^{1/2}
\end{align*}

\par Further, recalling the definition of the test function $g$ and in view of Cauchy\textendash Schwarz inequality, we can find a positive constant $C$ (not dependent on $t$) such that the equation above leads us to
\begin{equation*}
\mathds{E}\left[g\!\left(\vartheta_{t\,\wedge\,\zeta}\right)\right]\! \leq  C\,+\, C\mathds{E}\left[\int\limits_{0}^{t\,\wedge\,\zeta}\!\!\!\left\lvert\, a^{i}\!\left(Y^{i}_{u}\right)\right\rvert \left\lvert x^{-i}_{u}\right\rvert du\right]\,
+\, C\mathds{E}\left[\int\limits_{0}^{t\,\wedge\,\zeta}\!\!\!\left\lvert\, b^{i}\!\left(Y^{i}_{u}\right)\right\rvert \left\lvert x^{i}_{u}\right\rvert du\right]\,
+\, C\mathds{E}\left[\int\limits^{t\,\wedge\,\zeta}_{0}\!\!\!\left\lvert\,\mathrm{v}^{i}\!\left(Y^{i}_{u}\right)\right\rvert^{2}\!\! du\right]
\end{equation*}

\par Next, we recall the definition of the vector-valued functions $a^{i}$, $b^{i}$ and $\mathrm{v}^{i}$ in conjunction with the fact that the local volatility function satisfies {\color{azul-pesc}\Cref{siglip}} to ascertain that there exists a positive constant $C$ (not dependent on $t$) such that for $0 \leq t \leq T$ we obtain the following upper bounds
\begin{align*}
\left\lvert\, a^{i}\!\left(Y^{i}_{t}\right)\,\right\rvert \left\lvert\, x^{-i}_{t}\,\right\rvert\, \leq\,&\ C\Big(\left\lvert\, x^{-i}_{t}\, \right\rvert +\left\lvert\, x^{-i}_{t}\, \right\rvert\left\lvert\,\pi^{i}_{t}\,\right\rvert + \left\lvert\, x^{-i}_{t}\, \right\rvert\left\lvert\,\pi^{-i}_{t}\,\right\rvert\Big)\\
\left\lvert\, b^{i}\!\left(Y^{i}_{t}\right)\,\right\rvert \left\lvert\, x^{i}_{t}\,\right\rvert\, \leq\,&\ C\Big(\left\lvert\, x^{i}_{t}\, \right\rvert + \left\lvert\, x^{i}_{t}\, \right\rvert\left\lvert\,\pi^{i}_{t}\,\right\rvert + \left\lvert\, x^{i}_{t}\, \right\rvert\left\lvert\,\pi^{-i}_{t}\,\right\rvert\Big)\\
\left\lvert\,\mathrm{v}^{i}\!\left(Y^{i}_{t}\right)\,\right\rvert^{2}\, \leq\,&\ C\!\left(1 + \left\lvert\,\pi^{i}_{t}\,\right\rvert^{2} + \left\lvert\,\pi^{-i}_{t}\,\right\rvert^{2}\right)
\end{align*}

\par In view of the above, we invoke H\"{o}lder's inequality which in conjunction with the fact that $X^{i},\, X^{-i} \in \mathcal{A}_{0}$, for $i, -i \in \mathcal{I}$ then implies that we have
\begin{equation}
\mathds{E}\left[\int\limits^{t\,\wedge\,\zeta}_{0}\!\!\left\lvert x^{-i}_{u}\right\rvert\,\left\lvert \pi^{i}_{u}\right\rvert du\right]\ \leq\ \mathds{E}\left[\,\left(\,\int\limits^{T}_{0}\!\!\left\lvert x^{-i}_{u}\right\rvert du\,\right)\,\left(\,\int\limits^{T}_{0}\!\!\left\lvert x^{i}_{l}\right\rvert dl\,\right)\,\right]\ <\ \infty
\end{equation}

\par Similarly, given the definition of $\pi^{i}_{t}$, we can appeal to H\"{o}lder's inequality again to assert that for $i, -i \in \mathcal{I}$, such that $X^{i},\, X^{-i} \in \mathcal{A}_{0}$ we have
\begin{equation}
\mathds{E}\left[\int\limits^{t\,\wedge\,\zeta}_{0}\!\!\left\lvert\, \pi^{-i}_{u}\,\right\rvert^{2} du\right]\, \leq\ \mathds{E}\left[\int\limits^{T}_{0}\!\!\left(\int^{u}_{0}\!\left\lvert x^{-i}_{l}\right\rvert dl\,\right)^{2}\!\!\!du\right]\, \leq\ \mathds{E}\left[T\, \left(\,\int\limits^{T}_{0}\!\left\lvert x^{-i}_{u}\right\rvert du\right)^{2}\right]\ <\ \infty
\end{equation}

\par Therefore, given that the processes $X^{i},\, X^{-i} \in \mathcal{A}_{0}$, and in view of the bounds derived above, it then follows that for $T \in \mathds{R}_{+}$, we must have
\begin{equation*}
\mathds{E}\left[g\!\left(\vartheta_{t\,\wedge\,\zeta}\right)\right]\ =\ \mathds{E}\left[\int\limits^{\ \vartheta_{t\,\wedge\, \zeta}}_{0}\!\!\!\!\frac{d\epsilon}{\epsilon + 1}\,\right]\ \leq\ C_{T}\, < \infty
\end{equation*}

\par Notice that here the bound $C_{T}$ is a function of $T$ alone. Now, suppose $\mathds{P}\left(\zeta<T\right) > 0$. Then, there must exist $\hat{t} \leq T \in \mathds{R}_{+}$ such that $\mathds{P}\left(\zeta<\hat{t}\,\right) > 0$. From the above, it then follows that with non-zero probability, we must have
\begin{equation*}
\mathds{1}_{\left\{\zeta\,\leq\,\hat{t}\right\}}\ \mathds{E}\left[\,g\!\left(\vartheta_{\,\hat{t}\,\wedge\,\zeta}\right)\,\right]\ =\ \mathds{E}\left[\,\int\limits^{\vartheta_{\zeta}}_{0}\!\!\!\!\frac{d\epsilon}{\epsilon + 1}\,\right]\ \leq\ C_{T}\ <\ \infty
\end{equation*}

\par However, we have $g(\vartheta_{\zeta}) = \infty$ by definition, and hence we arrive at a contradiction. Thus, it follows that we have $\zeta = \infty$, $\mathds{P}$-almost surely and hence the solution is non-explosive. In particular, this implies that the system of stochastic differential equation has a well-defined strong solution for $0 < T < \infty$.
\end{proof}

\subsection{Proof of \textcolor{azul-pesc}{\Cref{reg}}}

\begin{proof}
\par(i) Recall that $b^{i}: \left(S,\pi^{i},\pi^{j},W^{i},W^{j}\right)^{\mathrm{T}} \rightarrow \left(-\theta^{i},-1,0,-\theta^{i}\pi^{i},-\theta^{i}\pi^{-i}\right)^{\mathrm{T}}$, which implies that the Jacobian matrix $J_{b^{i}}$ associated with the vector-valued function $b^{i}$ is well-defined and bounded with respect to the Hilbert\textendash Schmidt norm. This in turn implies that $b^{i}$ is Lipschitz continuous on its domain. The existence and uniqueness of $\phi^{i}$ then follows as an immediate consequence of the Picard\textendash Lindel\"of Theorem \cite[Theorem 6.1.3]{applebaum2009levy}.

\par (ii) We shall first prove that $\varepsilon^{i}_{q}$ is continuously differentiable. To this end, we recall from (i) above that $b^{i}$ is Lipschitz continuous on its domain. Let $K$ denote a Lipschitz constant for $b^{i}$ and consider vectors $y,\hat{y}$ in the domain of $b^{i}$. Lipschitz continuity of $b^{i}$ together with (\ref{flow}) then implies the following
\begin{equation*}
\left\lvert\varepsilon^{i}_{q}\!\left(y\right)-\varepsilon^{i}_{q}\!\left(\hat{y}\right)\right\rvert\leq\left\lvert y-\hat{y}\right\rvert+\!\mathlarger{\int}_{0}^{q}\!\!\left\lvert b^{i}\!\left(\varepsilon^{i}_{s}\!\left(y\right)\right)-b^{i}\!\left(\varepsilon^{i}_{s}\!\left(\hat{y}\right)\right)\right\rvert ds\leq \left\lvert y-\hat{y}\right\rvert + K\!\mathlarger{\int}_{0}^{q}\!\!\left\lvert\varepsilon^{i}_{s}\!\left(y\right)-\varepsilon^{i}_{s}\!\left(\hat{y}\right)\right\rvert ds
\end{equation*}

\par From the above, we obtain the following as a corollary of Gronwall's lemma \cite[Proposition 6.1.4]{applebaum2009levy}  
\begin{equation}\label{lipphi}
\left\lvert\,\varepsilon^{i}_{q}\left(y\right)-\varepsilon^{i}_{q}\left(\hat{y}\right)\,\right\rvert \leq \left\lvert y-\hat{y}\right\rvert e^{K\,\lvert q\rvert}
\end{equation}

\par Further, note from (i) above that the Jacobian $J_{b^{i}}$ is a matrix of scalar functions which are constant and hence independent of the state vector. It is then immediate that for these functions, partial derivatives of all orders exist and are bounded. Consequently, it follows from the Picard\textendash Lindel\"of theorem that for each vector $y$ in the domain of $b^{i}$, there exists a unique solution $e(q,y)$, to the following matrix-valued differential equation
\begin{equation}\label{dphi}
\frac{\partial}{\partial q}\,e(q,y)\,=\,J_{b^{i}}\!\left(e(q,y)\right),\quad e(0,y)\, =\, \mathds{I}_{5}
\end{equation}

\par The fact that $\varepsilon^{i}_{q}$ is continuously differentiable at a given point $y$ in its domain and that $\partial\varepsilon^{i}_{q}(y)/\partial y$ equals $e(q,y)$ follows from \cite[Theorem 6.1.7]{applebaum2009levy}. Likewise, the fact that $\varepsilon^{i}_{q}$ is twice continuously differentiable follows from applying the preceding arguments to the first derivative.
\end{proof}

\subsection{Proof of \textcolor{azul-pesc}{\Cref{ext}}}

\begin{proof}
\par \textsc{Existence} - In order to prove existence, we extend the Euler approximation method employed in \cite{lan2014new}. First, we construct a sequence of stochastic differential equations through discretization of (\ref{astate}) on successively finer grids. Subsequently, we show that the constructed sequence has at least one uniformly convergent subsequence. Finally, we establish that the limit process of this uniformly convergent subsequence is a solution to (\ref{astate}). In what follows, we simplify notation by assuming that $s=0$, without loss of generality. 

\par We define a sequence of stochastic processes $\{Z^{i,\,n},\,n\in \mathds{N}\cup 0\}$, by setting $Z^{i,\,0}_{t} = Z^{i}_{0}$ for $t \in [0,T]$. Further, for $n \geq 1$ we let $Z^{i,\,n}_{0} = Z^{i}_{0}$, and introduce $k(n,t) = \lfloor t n\rfloor/n$, $k^{+}(n,t) = \left(\lfloor t n\rfloor/n + 1\right)\wedge T$. Then, given $X^{-i} \in \mathcal{A}_{0}$, $\pi^{i} \in \mathcal{A}^{a}_{0}$ we define $Z^{i,\,n}_{t}$ for $t \in \left(k(n,t),\, k^{+}(n,t)\right]$, recursively as follows
\[Z^{i,\,n}_{t}\, =\ Z^{i,\,n}_{k(n,\,t)}\ +\int\limits^{t}_{k(n,\,t)}\!\!\!\!\beta^{i}\!\left(\pi^{i}_{u},Z^{i,\,n}_{k(n,\,u)}\right)x^{-i}_{u}du\ + \int\limits^{t}_{k(n,\,t)}\!\!\!\!\mathrm{\nu}^{i}\!\left(\pi^{i}_{u},Y^{i,\,n}_{k(n,\,u)}\right)dB_{u}\]

\par Note that the recursive form of the definition above enables us to rewrite the expression defining $Z^{i,\,n}_{t}$ equivalently as follows
\begin{equation*}
Z^{i,\,n}_{t}\, =\ Z^{i,\,n}_{0}\ +\int\limits^{t}_{0}\!\!\beta^{i}\!\left(\pi^{i}_{u},Z^{i,\,n}_{k(n,\,u)}\right)x^{-i}_{u}du\ + \int\limits^{t}_{0}\!\!\mathrm{\nu}^{i}\!\left(\pi^{i}_{u},Y^{i,\,n}_{k(n,\,u)}\right)dB_{u}
\end{equation*}

\par Next, we prove existence and uniqueness of a strong solution to the system of stochastic differential equations defined by (\ref{astate}) using a localization argument. To this end, we define $\mathds{F}$-stopping time $\hat{\tau}_{k}$ as follows
\begin{equation*}
\hat{\tau}_{k}\, =\, \inf\left\{t > s:\, \left\lvert x^{-i}_{t}\right\rvert\, >\, k\right\}
\end{equation*}

\par Further, given $X^{-i} \in \mathcal{A}_{0}$, $\pi^{i} \in \mathcal{A}^{a}_{0}$ and $Z^{i}_{0} \in \mathds{R}^{\left\lvert\, Z^{i}\,\right\rvert}$, it follows by definition of $Z^{i,\,n}$ that for $t \in [0,T]$ we have
\begin{equation*}
    \mathds{1}_{\left\{t\,\leq\,\hat{\tau}_{k}\right\}}\left\lvert\,Z^{i,\,1}_{t} - Z^{i,\,0}_{t}\,\right\rvert\ \leq\,\int\limits^{t}_{0}\!\! \mathds{1}_{\left\{u\,\leq\,\hat{\tau}_{k}\right\}}\left\lvert\, \beta^{i}\!\left(\pi^{i}_{u},Z^{i,\,1}_{k(n,\,u)}\right)\,\right\rvert\left\lvert\,x^{-i}_{u}\,\right\rvert du\, +\, \left\lvert\,\int\limits^{t}_{0}\!\!\mathds{1}_{\left\{u\,\leq\,\hat{\tau}_{k}\right\}}\,\mathrm{\nu}^{i}\!\left(\pi^{i}_{u},Z^{i,\,1}_{k(n,\,u)}\right)dB_{u}\,\right\rvert
\end{equation*}

\par Given the above, it then follows in view of Cauchy\textendash Schwarz inequality and the definition of the $\mathds{F}$-stopping time $\hat{\tau}_{k}$ that we can find a positive constant $C$ (dependent on $k$) such that
\begin{align*}
    \mathds{E}\left[\,\sup\limits_{t\,\in\, [0,\,T]}\mathds{1}_{\left\{t\,\leq\,\hat{\tau}_{k}\right\}}\left\lvert\,Z^{i,\,1}_{t} - Z^{i,\,0}_{t}\,\right\rvert^{2}\right]\, \leq\ &\,  C\mathds{E}\left[\,\int\limits^{T}_{0}\!\! \mathds{1}_{\left\{u\,\leq\,\hat{\tau}_{k}\right\}}\left\lvert\, \beta^{i}\!\left(\pi^{i}_{u},Z^{i,\,1}_{k(n,\,u)}\right)\,\right\rvert^{2}\!\!du\,\right]\\ 
    & +\, C\mathds{E}\left[\,\sup\limits_{t\,\in\,[0,\,T]}\,\left\lvert\,\int\limits^{t}_{0}\!\!\mathds{1}_{\left\{u\,\leq\,\hat{\tau}_{k}\right\}}\,\mathrm{\nu}^{i}\!\left(\pi^{i}_{u},Z^{i,\,1}_{k(n,\,u)}\right)dB_{u}\right\rvert^{2}\right]
\end{align*}

\par Further, we invoke Burkholder\textendash Davis\textendash Gundy inequality in view of the third term on the right-hand side of the equation above which implies the existence of a positive constant $C$ (dependent on $k$) such that the equation above leads us to
\begin{align*}
    \mathds{E}\left[\,\sup\limits_{t\,\in\, [0,\,T]}\mathds{1}_{\left\{t\,\leq\,\hat{\tau}_{k}\right\}}\left\lvert\,Z^{i,\,1}_{t} - Z^{i,\,0}_{t}\,\right\rvert^{2}\right]\, \leq\ &\,  C\mathds{E}\left[\,\int\limits^{T}_{0}\!\! \mathds{1}_{\left\{u\,\leq\,\hat{\tau}_{k}\right\}}\left\lvert\, \beta^{i}\!\left(\pi^{i}_{u},Z^{i,\,1}_{k(n,\,u)}\right)\,\right\rvert^{2}\!\!du\,\right]\\ 
    & +\, C\mathds{E}\left[\,\int\limits^{T}_{0}\!\! \mathds{1}_{\left\{u\,\leq\,\hat{\tau}_{k}\right\}}\left\lvert\, \nu^{i}\!\left(\pi^{i}_{u},Z^{i,\,1}_{k(n,\,u)}\right)\,\right\rvert^{2}\!\!du\,\right]
\end{align*}

\par Next, by definition of the $\mathds{F}$-stopping time $\hat{\tau}_{k}$ it follows that $\mathds{1}_{\left\{t\,\leq\,\hat{\tau}_{k}\right\}}\lvert\,\pi^{-i,\,1}_{k(n,\,t)}\,\rvert$ is bounded above by $kT$, for $t \in [0,T]$. Thus, it follows in view of the definition of the vector-valued functions $\beta^{i}$, $\nu^{i}$ in conjunction with the fact that the local volatility function $\sigma$ satisfies \textcolor{azul-pesc}{\Cref{siglip}} and $\pi^{i} \in \mathcal{A}^{a}_{0}$, that we can find a positive constant $C$ (dependent on $k$) such that the equation above leads us to
\begin{equation}\label{ind5}
    \mathds{E}\left[\,\sup\limits_{t\,\in\, [0,\,T]}\mathds{1}_{\left\{t\,\leq\,\hat{\tau}_{k}\right\}}\left\lvert\,Z^{i,\,1}_{t} - Z^{i,\,0}_{t}\,\right\rvert^{2}\right]\, \leq \,  CT\left(1 + k^{2}T^{2}\right)
\end{equation}

\par Similarly, for $n \in \mathds{N}$, with $n > 1$, we again invoke the definition of the $\mathds{F}$-stopping time $\hat{\tau}_{k}$ along with Cauchy\textendash Schwarz inequality to assert that we can find a positive constant $C$ (dependent on $k$) such that
\begin{multline*}
    \mathds{E}\left[\,\sup\limits_{t\,\in\, [0,\,T]}\mathds{1}_{\left\{t\,\leq\,\hat{\tau}_{k}\right\}}\left\lvert\,Z^{i,\,n}_{t} - Z^{i,\,n-1}_{t}\,\right\rvert^{2}\right] \leq  C\mathds{E}\left[\,\int\limits^{T}_{0}\!\! \mathds{1}_{\left\{u\,\leq\,\hat{\tau}_{k}\right\}}\left\lvert\, \beta^{i}\!\left(\pi^{i}_{u},Z^{i,\,n}_{k(n,\,u)}\right) - \beta^{i}\!\left(\pi^{i}_{u},Z^{i,\,n-1}_{k(n,\,u)}\right)\,\right\rvert^{2}\!du\,\right]\\ 
    + C\mathds{E}\left[\,\sup_{t\,\in\,[0,\,T]}\,\left\lvert\,\int\limits^{t}_{0}\!\!\mathds{1}_{\left\{u\,\leq\,\hat{\tau}_{k}\right\}}\,\left(\mathrm{\nu}^{i}\!\left(\pi^{i}_{u},Z^{i,\,n}_{k(n,\,u)}\right)-\mathrm{\nu}^{i}\!\left(\pi^{i}_{u},Z^{i,\,n-1}_{k(n,\,u)}\right)\right)dB_{u}\right\rvert^{2}\right]
\end{multline*}

\par Next, we recall Burkholder\textendash Davis\textendash Gundy inequality in conjunction with the definition of the vector-valued functions $\beta^{i}$, which implies that the equation above leads us to
\begin{align*}
    \mathds{E}\left[\,\sup\limits_{t\,\in\, [0,\,T]}\mathds{1}_{\left\{t\,\leq\,\hat{\tau}_{k}\right\}}\left\lvert\,Z^{i,\,n}_{t} - Z^{i,\,n-1}_{t}\,\right\rvert^{2}\right] \leq\ &\, C\mathds{E}\left[\,\int\limits^{T}_{0}\!\! \mathds{1}_{\left\{u\,\leq\,\hat{\tau}_{k}\right\}}\left\lvert\,Z^{i,\,n}_{k(n,\,u)} - Z^{i,\,n-1}_{k(n,\,u)}\,\right\rvert^{2}\!du\,\right]\\ 
    & + C\mathds{E}\left[\,\int\limits^{T}_{0}\!\!\mathds{1}_{\left\{u\,\leq\,\hat{\tau}_{k}\right\}}\,\left\lvert\mathrm{\nu}^{i}\!\left(\pi^{i}_{u},Z^{i,\,n}_{k(n,\,u)}\right)-\mathrm{\nu}^{i}\!\left(\pi^{i}_{u},Z^{i,\,n-1}_{k(n,\,u)}\right)\right\rvert^{2}\!dB_{u}\,\right]
\end{align*}

\par Further, given that the local volatility function $\sigma$ satisfies {\textcolor{azul-pesc}{\Cref{siglip}}} and $\pi^{i} \in \mathcal{A}^{a}_{0}$ it then follows from the definition of $\mathrm{\nu}^{i}$ and the $\mathds{F}$-stopping time $\hat{\tau}_{k}$ that the function $\mathrm{\nu}^{i}$ satisfies Lipschitz continuity on $[0,\hat{\tau}_{k}]$. This fact in conjunction with Tonelli's theorem implies that we can then find a positive constant $C$ (dependent on $k$) such that the equation above leads us to
\begin{equation}\label{ind6}
    \mathds{E}\left[\,\sup\limits_{t\,\in\, [0,\,T]}\mathds{1}_{\left\{t\,\leq\,\hat{\tau}_{k}\right\}}\left\lvert\,Z^{i,\,n}_{t} - Z^{i,\,n-1}_{t}\,\right\rvert^{2}\right]\, \leq \ C\!\!\int\limits^{T}_{0}\!\!\mathds{E}\left[\,\sup_{\hat{u}\,\in\,[0,\,u]} \mathds{1}_{\left\{\hat{u}\ \leq\ \hat{\tau}_{k}\right\}}\left\lvert\,Z^{i,\,n}_{\hat{u}} - Z^{i,\,n-1}_{\hat{u}}\,\right\rvert^{2}\right]du
\end{equation}

\par In view of (\ref{ind5}) and (\ref{ind6}), we then proceed by induction to deduce the following bound
\begin{equation}\label{ind7}
    \mathds{E}\left[\,\sup\limits_{t\,\in\, [0,\,T]}\mathds{1}_{\left\{t\,\leq\,\hat{\tau}_{k}\right\}}\left\lvert\,Z^{i,\,n}_{t} - Z^{i,\,n-1}_{t}\,\right\rvert^{2}\right]\, \leq \ \frac{C\left(1+k^{2}T^{2}\right)T^{n}}{n!}
\end{equation}

\par Note that the positive constant $C$ in the equation above is dependent on $k$. Having derived the bound above, we next establish that $\left\{Z^{i,\,n}_{t}\right\}_{n\,\in\,\mathds{N}}$ is convergent in $\mathds{L}^{2}$ for $t \in [0, \hat{\tau}_{k}]$. To this end, consider $m,n \in \mathds{N}$ and note that given the equation above it follows that we have
\begin{equation}\label{ind8}
    \mathds{1}_{\left\{t\,\leq\,\hat{\tau}_{k}\right\}}\left\lvert\,Z^{i,\,n}_{t} - Z^{i,\,m}_{t}\,\right\rvert_{\,\mathds{L}^{2}}\ \leq\ \mathlarger{\sum}\limits^{n}_{j\,=\ m\,+\,1}\mathds{1}_{\left\{t\,\leq\,\hat{\tau}_{k}\right\}}\left\lvert\,Z^{i,\,j}_{t} - Z^{i,\ j-1}_{t}\,\right\rvert_{\,\mathds{L}^{2}}\ \leq\ \mathlarger{\sum}\limits^{n}_{j\, =\ m\,+\,1}\frac{\left[\,C\left(1+k^{2}T^{2}\right)T^{j}\right]^{1/2}}{\left(j!\right)^{1/2}}
\end{equation}

\par It is immediate that the sum on the right-hand side of the equation above converges, which in turn implies that the sequence $\left\{Z^{i,\,n}_{t}\right\}_{n\,\in\,\mathds{N}} \in \mathds{L}^{2}$ is Cauchy for $t \in [0, \hat{\tau}_{k}]$. Further, since $\mathds{L}^{2}$ is a Banach space, it is complete and hence the sequence $\left\{Z^{i,\,n}_{t}\right\}_{n\,\in\,\mathds{N}}$ converges to some $Z^{i}_{t} \in \mathds{L}^{2}$ for $t \in [0, \hat{\tau}_{k}]$. To establish $\mathds{P}$-almost sure uniform convergence of the sequence to the limit process $Z^{i}$ obtained above, we note that given (\ref{ind7}), it follows in view of Chebyshev's inequality that we have
\begin{equation*}
    \mathds{P}\left(\,\sup\limits_{t\,\in\, [0,\,T]}\mathds{1}_{\left\{t\,\leq\,\hat{\tau}_{k}\right\}}\left\lvert\,Z^{i,\,n}_{t} - Z^{i,\,n-1}_{t}\,\right\rvert^{2}\,\geq\,\frac{1}{2^{n}}\,\right)\ \leq \ \frac{C\left(1+k^{2}T^{2}\right)(4T)^{n}}{n!}
\end{equation*}

\par Since the term on the right-hand side of the equation above converges to zero as $n \rightarrow \infty$, we invoke Borel\textendash Cantelli lemma \cite[Theorem 8.3.4]{dudley2018real} which then implies that the equation above leads us to 
\begin{equation*}
    \mathds{P}\left(\,\limsup_{n\, \rightarrow\, \infty}\,\sup\limits_{t\,\in\, [0,\,T]}\mathds{1}_{\left\{t\,\leq\,\hat{\tau}_{k}\right\}}\left\lvert\,Z^{i,\,n}_{t} - Z^{i,\,n-1}_{t}\,\right\rvert^{2}\,\geq\,\frac{1}{2^{n}}\,\right)\ = \ 0
\end{equation*}

\par Note that we can re-write the above equivalently as
\begin{equation*}
    \mathds{P}\left(\,\liminf_{n\, \rightarrow\, \infty}\,\sup\limits_{t\,\in\, [0,\,T]}\mathds{1}_{\left\{t\,\leq\,\hat{\tau}_{k}\right\}}\left\lvert\,Z^{i,\,n}_{t} - Z^{i,\,n-1}_{t}\,\right\rvert^{2}\,<\,\frac{1}{2^{n}}\,\right)\ = \ 1
\end{equation*}

\par It is then immediate from the above that the sequence $\{Z^{i,\,n}\}_{n\,\in\,\mathds{N}}$ is $\mathds{P}$-almost surely uniformly convergent on the interval $[0,\hat{\tau}_{k}]$ to the limit process $Z^{i} \in \mathds{L}^{2}$ obtained above.

\par Next, we show that the stochastic process $\{Z^{i}\}$ that is obtained as a uniform limit of the Euler approximation is indeed a solution to the stochastic differential equation (\ref{astate}) on the interval $[0,\hat{\tau}_{k}]$. To this end, we define a stochastic process $\hat{Z}^{i}$ for $t \in [0,\hat{\tau}_{k}]$ as follows
\begin{equation*}
\hat{Z}^{i}_{t}\, =\, Z^{i}_{0}\, +\, \int\limits^{t}_{0}\!\!\beta^{i}\!\left(\pi^{i}_{u},Z^{i}_{u}\right)x^{-i}_{u}du\,+\, \int\limits^{t}_{0}\!\!\mathrm{\nu}^{i}\!\left(\pi^{i}_{u},Z^{i}_{u}\right)dB_{u}
\end{equation*}

\par Given the definition above, we appeal to an argument identical to the one used in deriving (\ref{ind6}) above, in order to assert the existence the of a positive constant $C$ (dependent on $k$) such that we have
\begin{equation*}
    \mathds{E}\left[\mathds{1}_{\left\{t\,\leq\,\hat{\tau}_{k}\right\}}\left\lvert\,\hat{Z}^{i}_{t} - Z^{i,\,n}_{t}\,\right\rvert^{2}\right] \leq C\!\!\int\limits^{T}_{0}\!\!\mathds{E}\left[ \mathds{1}_{\left\{u\ \leq\ \hat{\tau}_{k}\right\}}\left\lvert\,Z^{i}_{u} - Z^{i,\,n}_{u}\,\right\rvert^{2}\right]du \leq CT \!\! \sup\limits_{t\, \in\, [0,\,T]}\mathds{E}\left[\mathds{1}_{\left\{t\,\leq\,\hat{\tau}_{k}\right\}}\left\lvert\,Z^{i}_{t} - Z^{i,\,n}_{t}\,\right\rvert^{2}\right]
\end{equation*}

\par Further, it follows in view of the bound derived in (\ref{ind8}) that we have
\begin{equation*}
    \sup\limits_{t\, \in\, [0,\,T]}\,\mathds{E}\left[\,\mathds{1}_{\left\{t\,\leq\,\hat{\tau}_{k}\right\}}\left\lvert\,Z^{i}_{t} - Z^{i,\,n}_{t}\,\right\rvert^{2}\right]\ \ \leq\ \ \left(\,\mathlarger{\sum}\limits^{\infty}_{j\, =\ n\,+\,1}\frac{\left[\ C\left(1+k^{2}T^{2}\right)T^{j}\right]^{1/2}}{\left(j!\right)^{1/2}}\right)^{2}
\end{equation*}

\par Thus, we combine the inequalities above and recall that the right-hand side of the equation above converges to zero as $n \rightarrow \infty$ to establish that 
\begin{equation*}
    \limsup\limits_{n\, \rightarrow\, \infty}\,\sup\limits_{t\,\in\,[0,T]}\,\mathds{1}_{\left\{t\,\leq\,\hat{\tau}_{k}\right\}}\left\lvert\,\hat{Z}^{i}_{t} - Z^{i,\,n}_{t}\,\right\rvert_{\,\mathds{L}^2}\ = \ 0
\end{equation*}

\par It is immediate from the above and the uniqueness of limits that we have \[\mathds{P}\Big(\hat{Z}^{i}_{t}\, =\, Z^{i}_{t},\ t \in \left[0, \hat{\tau}_{k}\right]\Big)\, =\, 1\]

\par \textsc{Uniqueness} - Let $Z^{i}$ and $\tilde{Z}^{i}$ be two distinct stochastic processes which are solutions to the system of stochastic differential equations defined by (\ref{astate}). We again employ an argument identical to the one used in deriving (\ref{ind6}) above, in order to assert the existence the of a positive constant $C$ (dependent on $k$) such that we have
\begin{equation*}
    \mathds{E}\left[\,\sup\limits_{t\,\in\, [0,\,T]}\mathds{1}_{\left\{t\,\leq\,\hat{\tau}_{k}\right\}}\left\lvert\,Z^{i}_{t} - \tilde{Z}^{i}_{t}\,\right\rvert^{2}\right]\, \leq \ C\!\!\int\limits^{T}_{0}\!\!\mathds{E}\left[\,\sup_{\hat{u}\,\in\,[0,\,u]} \mathds{1}_{\left\{\hat{u}\ \leq\ \hat{\tau}_{k}\right\}}\left\lvert\,Z^{i}_{\hat{u}} - \tilde{Z}^{i}_{\hat{u}}\,\right\rvert^{2}\right]du
\end{equation*}

\par Given the equation above, it is then immediate from Gronwall's lemma that we have
\begin{equation*}
    \mathds{E}\left[\,\sup\limits_{t\,\in\, [0,\,T]}\mathds{1}_{\left\{t\,\leq\,\hat{\tau}_{k}\right\}}\left\lvert\,Z^{i}_{t} - \tilde{Z}^{i}_{t}\,\right\rvert^{2}\right]\, = \ 0
\end{equation*}

\par In view of the above, it is then immediate from Chebyshev's inequality that we have \[\mathds{P}\Big(Z^{i}_{t}\, =\, \tilde{Z}^{i}_{t},\ t \in \left[0, \hat{\tau}_{k}\right]\Big)\, =\, 1\]

\par \textsc{Non-Explosion} - Finally, we establish that the solution is non-explosive. To this end, we let $\zeta$ denote the lifetime of $\{Z^{i}\}$ defined as $\zeta = \liminf_{n\,\rightarrow \,\infty}\{\,t > 0: \lvert Z^{i}_{t}\rvert \geq  n \}$, where $n \in \mathds{N}$. It is immediate in view of the above that
\begin{equation*}
\zeta\, =\, \liminf_{k\, \rightarrow\, \infty}\,\hat{\tau}_{k}    
\end{equation*}

\par It then remains to show that $\zeta > T$, $\mathds{P}$-almost surely. To this end, we follow the approach employed in \cite{lan2014new}, whereby we consider the following test function defined on the positive real line
\begin{equation*}
    g(x)\, =\, \int\limits^{x}_{0}\!\!\!\frac{d\epsilon}{\epsilon+1}
\end{equation*}

\par The function $g(x)$ defined above is strictly increasing, and strictly concave on the positive real line. Moreover, since $\lim_{\,x \rightarrow 0}\, D_{1}g(x) = 1$, we can find a concave extension $\tilde{g}$ of $g$ on the entire real line, such that $\tilde{g}(x) = g(x)$, for all $x \geq 0$. Also, the second-order derivative $D^{2}_{1}\,\tilde{g}(x)$ of the extended function $\tilde{g}$, is a non-positive Radon measure in the generalized sense \cite[Appendix III]{revuz2013continuous}. Given the above, we let $\vartheta_{t} = \left\lvert Z^{i}_{t}\right\rvert^{2}$, and apply the It\^{o}\textendash Meyer\textendash Tanaka formula \cite[Theorem VI.1.5]{revuz2013continuous} to the function $\tilde{g}\!\left(\vartheta_{t\,\wedge\,\zeta}\right)$ so as to obtain the following
\begin{align*}
\tilde{g}\!\left(\vartheta_{t\,\wedge\,\zeta}\right)\, =\ &\ \tilde{g}(\vartheta_{0})\,+\, 2\!\!\int\limits_{0}^{t\,\wedge\,\zeta}\!\!\!D_{1}\tilde{g}(\vartheta_{u})\left\langle Z^{i}_{u}\,,\,\beta^{i}\!\left(\pi^{i}_{u},Z^{i}_{u}\right)x^{-i}_{u}\right\rangle du\,
+\, \int\limits_{0}^{t\,\wedge\,\zeta}\!\!\!D_{1}\tilde{g}(\vartheta_{u})\left\lvert\mathrm{\nu}^{i}\!\left(\pi^{i}_{u},Z^{i}_{u}\right)\right\rvert^{2}\!du\\
& +\, 2\!\!\int\limits_{0}^{t\,\wedge\,\zeta}\!\!\!D_{1}\tilde{g}(\vartheta_{u})\left\langle Z^{i}_{u}\,,\mathrm{\nu}^{i}\!\left(\pi^{i}_{u},Z^{i}_{u}\right)\right\rangle dB_{u}\,
+\,\frac{1}{2}\int_{\mathds{R}}\mathrm{L}^{q}_{t\,\wedge\,\zeta}(\vartheta)\,D^{2}_{1}\,\tilde{g}\left(dq\right)
\end{align*}

\par Note that the stochastic process $\mathrm{L}^{q}$ above denotes the local time process of $\vartheta$. Since, the local time is always non-negative by definition, and the Radon measure $D^{2}_{1}\,\tilde{g}$ is non-positive given that $\tilde{g}$ is concave, it follows that upon taking expectations we have
\begin{align*}
\mathds{E}\left[g\!\left(\vartheta_{t\,\wedge\,\zeta}\right)\right]\leq\, &\, g(\vartheta_{0})+ 2\mathds{E}\left[\int\limits_{0}^{t\,\wedge\,\zeta}\!\!\!D_{1}g(\vartheta_{u})\left\langle Z^{i}_{u}\,,\,\beta^{i}\!\left(\pi^{i}_{u},Z^{i}_{u}\right)x^{-i}_{u}\right\rangle du\right]
+\mathds{E}\left[\mathlarger{\int}\limits^{t\,\wedge\,\zeta}_{0}\!D_{1}g(\vartheta_{u})\left\lvert\nu^{i}\!\left(\pi^{i}_{u},Z^{i}_{u}\right)\right\rvert^{2}\!du\right]\\
 &+2\mathds{E}\left[\,\sup\limits_{0\,\leq\, \hat{t}\, \leq\, t}\ \left\lvert\, \int\limits_{0}^{\hat{t}}\!\mathds{1}_{\left\{\,u\,\leq\,\zeta\,\right\}}D_{1}{g}(\vartheta_{u})\left\langle Z^{i}_{u}\,,\mathrm{\nu}^{i}\!\left(\pi^{i}_{u},Z^{i}_{u}\right)\right\rangle dB_{u}\,\right\rvert\ \right]
\end{align*}

\par Given the equation above, we invoke Burkholder\textendash Davis\textendash Gundy inequality, by way of which we can find a positive constant $C$ (not dependent on $t$) such that the equation above leads us to
\begin{align*}
\mathds{E}\left[g\!\left(\vartheta_{t\,\wedge\,\zeta}\right)\right]\, \leq\,&\, C + C\mathds{E}\left[\int\limits_{0}^{t\,\wedge\,\zeta}\!\!\!D_{1}g(\vartheta_{u})\left\langle Z^{i}_{u}\,,\,\beta^{i}\!\left(\pi^{i}_{u},Z^{i}_{u}\right)x^{-i}_{u}\right\rangle du\right]
+C\mathds{E}\left[\mathlarger{\int}\limits^{t\,\wedge\,\zeta}_{0}\!D_{1}g(\vartheta_{u})\left\lvert\nu^{i}\!\left(\pi^{i}_{u},Z^{i}_{u}\right)\right\rvert^{2}\!du\right]\\
& + \, C\mathds{E}\left[\,\int\limits_{0}^{t\,\wedge\,\zeta}\!\!\big\lvert D_{1}{g}(\vartheta_{u})\big\rvert^{2}\left\lvert Z^{i}_{u}\right\rvert^{2}\,\left\lvert\mathrm{\nu}^{i}\!\left(\pi^{i}_{u},Z^{i}_{u}\right)\right\rvert^{2}\!\! du\,\right]^{1/2}
\end{align*}

\par Further, recalling the definition of the test function $g$ and in view of Cauchy\textendash Schwarz inequality, we can find a positive constant $C$ (not dependent on $t$) such that the equation above leads us to
\begin{equation*}
\mathds{E}\left[g\!\left(\vartheta_{t\,\wedge\,\zeta}\right)\right]\, \leq\,  C\,+\, C\mathds{E}\left[\int\limits_{0}^{t\,\wedge\,\zeta}\!\!\left\lvert\, \beta^{i}\!\left(\pi^{i}_{u},Z^{i}_{u}\right)\right\rvert \left\lvert x^{-i}_{u}\right\rvert du\right]\,
+\, C\mathds{E}\left[\int\limits^{t\,\wedge\,\zeta}_{0}\!\!\left\lvert\,\mathrm{\nu}^{i}\!\left(\pi^{i}_{u},Z^{i}_{u}\right)\right\rvert^{2}\!du\right]
\end{equation*}

\par Next, we recall the definition of the vector-valued functions $\beta^{i}$, $\nu^{i}$ in conjunction with the fact that the local volatility function satisfies {\color{azul-pesc}\Cref{siglip}} and $\pi^{i} \in \mathcal{A}^{a}_{0}$ to ascertain that there exists a positive constant $C$ (not dependent on $t$) such that for $0 \leq t \leq T$ we obtain the following upper bounds
\begin{align*}
\left\lvert\, \beta^{i}\!\left(\pi^{i}_{t},Z^{i}_{t}\right)\,\right\rvert\, \left\lvert\, x^{-i}_{t}\,\right\rvert\, \leq\,&\, C\,\Big(\left\lvert\, x^{-i}_{t}\, \right\rvert + \left\lvert\, x^{-i}_{t}\, \right\rvert\,\left\lvert\,\pi^{-i}_{t}\,\right\rvert\Big)\\
\left\lvert\,\mathrm{\nu}^{i}\!\left(\pi^{i}_{t},Z^{i}_{t}\right)\,\right\rvert^{2}\, \leq\,&\, C\left(1 + \left\lvert\,\pi^{-i}_{t}\,\right\rvert^{2}\right)
\end{align*}

\par In view of the above, we invoke H\"{o}lder's inequality which in conjunction with the fact that $X^{-i} \in \mathcal{A}_{0}$, then implies that we have
\begin{equation}
\mathds{E}\left[\int\limits^{t\,\wedge\,\zeta}_{0}\!\!\left\lvert x^{-i}_{u}\right\rvert\,\left\lvert \pi^{-i}_{u}\right\rvert du\right]\, \leq\, \mathds{E}\left[\left(\,\int\limits^{T}_{0}\!\!\left\lvert x^{-i}_{u}\right\rvert du\,\right)\,\left(\,\int\limits^{T}_{0}\!\!\left\lvert x^{-i}_{l}\right\rvert dl\,\right)\right]\, <\, \infty
\end{equation}

\par Similarly, given the definition of $\pi^{-i}_{t}$, we can appeal to H\"{o}lder's inequality again to assert that given $X^{-i} \in \mathcal{A}_{0}$ we have
\begin{equation}
\mathds{E}\left[\int\limits^{t\,\wedge\,\zeta}_{0}\!\!\left\lvert\, \pi^{-i}_{u}\,\right\rvert^{2}\! du\right]\, \leq\, \mathds{E}\left[\int\limits^{T}_{0}\!\!\left(\int^{u}_{0}\!\left\lvert x^{-i}_{l}\right\rvert dl\,\right)^{2}\!\!\!du\right]\, \leq\, \mathds{E}\left[T \left(\,\int\limits^{T}_{0}\!\left\lvert x^{-i}_{u}\right\rvert du\right)^{2}\right]\, <\, \infty
\end{equation}

\par Therefore, given that the process $X^{-i} \in \mathcal{A}_{0}$, and in view of the bounds derived above, it then follows that for $T \in \mathds{R}_{+}$, we must have
\begin{equation*}
\mathds{E}\left[g\!\left(\vartheta_{t\,\wedge\,\zeta}\right)\right]\, =\, \mathds{E}\left[\!\!\int\limits^{\ \vartheta_{t\,\wedge\, \zeta}}_{0}\!\!\!\!\frac{d\epsilon}{\epsilon + 1}\right]\, \leq\, C_{T}\, < \infty
\end{equation*}

\par Notice that here the bound $C_{T}$ is a function of $T$ alone. Now, suppose $\mathds{P}\left(\zeta<T\right) > 0$. Then, there must exist $\hat{t} \leq T \in \mathds{R}_{+}$ such that $\mathds{P}\left(\zeta<\hat{t}\,\right) > 0$. From the above, it then follows that with non-zero probability, we must have
\begin{equation*}
\mathds{1}_{\left\{\zeta\,\leq\,\hat{t}\right\}}\ \mathds{E}\left[g\!\left(\vartheta_{\,\hat{t}\,\wedge\,\zeta}\right)\right]\, =\, \mathds{E}\left[\int\limits^{\vartheta_{\zeta}}_{0}\!\!\!\!\frac{d\epsilon}{\epsilon + 1}\right]\, \leq\, C_{T}\, <\, \infty
\end{equation*}

\par However, we have $g(\vartheta_{\zeta}) = \infty$ by definition, and hence we arrive at a contradiction. Thus, it follows that we have $\zeta = \infty$, $\mathds{P}$-almost surely and hence the solution is non-explosive. In particular, this implies that the system of stochastic differential equation has a well-defined strong solution for $0 < T < \infty$.
\end{proof}

\subsection{Regularity of State Process $Y^{i}$}

\par We state and prove the following ancillary lemma which collects certain regularity properties of the original state process $Y^{i}$, which are frequently called upon in the proof of \textcolor{azul-pesc}{\Cref{jconv}} below.

\begin{lemma}\label{momin} 
\par Consider an initial time $s \in [0,T]$, and assume that $k \in \mathds{N}$, $X^{i},\, X^{-i}  \in \mathcal{A}_{s}$ are given. Suppose the local volatility function $\sigma$ satisfies {\textcolor{azul-pesc}{\Cref{siglip}}}, and let $\tau_{k,\,x^{i}}$, $\tau_{k,\,x^{-i}}$, $\tau_{k}$ denote $\mathds{F}$-stopping times defined as follows
\begin{equation*}
\tau_{k,\,x^{i}}\, =\, \inf\left\{t > s:\, \int\limits_{s}^{t}\!\!\left\lvert x^{i}_{u}\right\rvert du\, >\, k\right\},\ \ \tau_{k,\,x^{-i}}\, =\, \inf\left\{t > s:\, \int\limits_{s}^{t}\!\!\left\lvert x^{-i}_{u}\right\rvert du\, >\, k\right\},\ \ \tau_{k}\, =\, \tau_{k,\,x^{i}}\,\wedge\,\tau_{k,\,x^{-i}}
\end{equation*}

\begin{enumerate}[(i)]
\item Suppose that the dynamics of $Y^{i}$ satisfy (\ref{stated}), with $Y^{i}_{s} = y^{i}$ given, then there exists a constant $C > 0$ (dependent on $k$ and $T$) such that the following holds
\begin{equation*}
\mathds{E}\left[\sup\limits_{t\,\in\,[s,T]}\mathds{1}_{\left\{t\,\leq\,\tau_{k}\right\}}\left\lvert\, Y^{i}_{t}\,\right\rvert^{4}\right]\,\leq\,C \left(1 + \left\lvert y^{i}\right\rvert^{4}\right)   
\end{equation*}

\item Suppose that the dynamics of $Y^{i}$, $\widehat{Y}^{i}$ satisfy (\ref{stated}), with $Y^{i}_{s} = y^{i}$, and $\widehat{Y}^{i}_{s} = \hat{y}^{i}$ respectively, then there exists a constant $C > 0$ (dependent on $k$ and $T$) such that the following holds
\begin{equation*}
\mathds{E}\left[\sup\limits_{t\,\in\,[s,T]}\left\lvert\, Y^{i}_{t}-\widehat{Y}^{i}_{t}\,\right\rvert^{2}\right]\,\leq\,C \left\lvert\, y^{i} - \hat{y}^{i}\,\right\rvert^{2}   
\end{equation*}

\item Suppose the dynamics of $Y^{i,\,s}$ satisfy (\ref{stated}), for an initial time $s \in [0,T]$, with $Y^{i}_{s} = y^{i}$, and the dynamics of $Y^{i,\,\hat{s}}$ satisfy (\ref{stated}), for an initial time $\hat{s} \in [0,s)$, with $Y^{i,\,\hat{s}}_{\hat{s}} = y^{i}$, we then have
\begin{equation*}
\lim\limits_{\hat{s}\,\uparrow\,s}\,\mathds{E}\left[\;\left\lvert\, Y^{i,\,s}_{T}-Y^{i,\,\hat{s}}_{T}\,\right\rvert^{2}\right]\,=\,0   
\end{equation*}
\end{enumerate}
\end{lemma}

\begin{proof}
\par $\mathbf{(i)}$ Recall that by definition we have 
\begin{align*}
    & a^{i}\!\left(Y^{i}_{t}\right)\ =\ \left(-\theta^{-i},\,0,\,-1,\,-\theta^{-i}\pi^{i}_{t},\,-\theta^{-i}\pi^{-i}_{t}\right)^{\mathrm{T}}\\
    & b^{i}\!\left(Y^{i}_{t}\right)\ =\ \left(-\theta^{i},\,-1,\,0,\,-\theta^{i}\pi^{i}_{t},\,-\theta^{i}\pi^{-i}_{t}\right)^{\mathrm{T}}
\end{align*} 

\par  It is then immediate that the functions $a^{i}$, $b^{i}$ are Lipschitz continuous, and hence satisfy a linear growth condition. This fact along with (\ref{stated}) then implies that there exists $C > 0$, such that for $t \in [s,T]$ we have
\begin{align}\label{momlo}
\begin{split}
\mathds{1}_{\left\{t\,\leq\,\tau_{k}\right\}}\left\lvert\, Y^{i}_{t}\,\right\rvert\,\leq\,&\,\left\lvert\, y^{i}\,\right\rvert\,+\,C\!\int\limits^{t}_{s}\!\!\mathds{1}_{\left\{u\leq \tau_{k}\right\}}\Big(\left\lvert x^{i}_{u}\right\rvert+\left\lvert x^{-i}_{u}\right\rvert\Big)du\,+\,C\!\int\limits^{t}_{s}\!\!\mathds{1}_{\left\{u\leq \tau_{k}\right\}}\Big(\left\lvert x^{i}_{u}\right\rvert+\left\lvert x^{-i}_{u}\right\rvert\Big)\left\lvert\,Y^{i}_{u}\,\right\rvert du\\
&+\,\left\lvert\,\int\limits^{t}_{s}\!\!\mathds{1}_{\left\{u\leq \tau_{k}\right\}}\mathrm{v}^{i}\!\left(Y^{i}_{u} \right)dB_{u}\,\right\rvert
\end{split}
\end{align}

\par For notational convenience, we let $f^{i}_{t}$ denote the random variable defined as follows
\begin{equation*}
f^{i}_{t}\ =\ \left\lvert\, y^{i}\,\right\rvert\,+\,C\!\int\limits^{t}_{s}\!\!\mathds{1}_{\left\{u\leq \tau_{k}\right\}}\Big(\left\lvert x^{i}_{u}\right\rvert+\left\lvert x^{-i}_{u}\right\rvert\Big)du\,+\,\left\lvert\,\int\limits^{t}_{s}\!\!\mathds{1}_{\left\{u\leq \tau_{k}\right\}}\mathrm{v}^{i}\!\left(Y^{i}_{u} \right)dB_{u}\,\right\rvert    
\end{equation*}

\par With the help of the notation introduced above, we can rewrite (\ref{momlo}) succinctly as follows
\begin{equation*}
\mathds{1}_{\left\{t\,\leq\,\tau_{k}\right\}}\left\lvert\, Y^{i}_{t}\,\right\rvert\,\leq\,f^{i}_{t}\,+\,C\!\int\limits^{t}_{s}\!\!\mathds{1}_{\left\{u\leq \tau_{k}\right\}}\Big(\left\lvert x^{i}_{u}\right\rvert+\left\lvert x^{-i}_{u}\right\rvert\Big)\left\lvert\,Y^{i}_{u}\,\right\rvert du
\end{equation*}

\par In view of the equation above we invoke Gronwall's lemma in conjunction with the definition of $\tau_{k}$, in order to obtain a positive constant $C$ (dependent on $k$ and $T$), such that the following holds
\begin{equation*}
\mathds{1}_{\left\{t\,\leq\,\tau_{k}\right\}}\left\lvert\, Y^{i}_{t}\,\right\rvert\,\leq\,f^{i}_{t}\,+\,C\!\int\limits^{t}_{s}\!\!\mathds{1}_{\left\{u\leq \tau_{k}\right\}}f^{i}_{u}\,\Big(\left\lvert x^{i}_{u}\right\rvert+\left\lvert x^{-i}_{u}\right\rvert\Big)\,du\,\leq\sup\limits_{u\,\in\,[s,\,t]}f^{i}_{u}
\end{equation*}

\par From the equation above it then follows that we have
\begin{equation*}
\sup\limits_{t\,\in\,[s,\,T]}\mathds{1}_{\left\{t\,\leq\,\tau_{k}\right\}}\left\lvert\, Y^{i}_{t}\,\right\rvert^{4}\,\leq\sup\limits_{t\,\in\,[s,\,T]}\left(\,\sup\limits_{u\,\in\,[s,\,t]}f^{i}_{u}\,\right)^{\!4}\!\leq\,\left(\,\sup\limits_{t\,\in\,[s,\,T]}f^{i}_{t}\,\right)^{\!4}
\end{equation*}

\par Also, it is immediate from the definition of $f^{i}_{t}$ above that there exists $C > 0$ such that we have
\begin{equation*}
\left(\,\sup\limits_{t\,\in\,[s,\,T]}f^{i}_{t}\,\right)^{\!4}\! \leq\ \left\lvert\, y^{i}\,\right\rvert^{4}+\,C\,\left(\int\limits^{T}_{s}\!\!\mathds{1}_{\left\{u\leq \tau_{k}\right\}}\Big(\left\lvert x^{i}_{u}\right\rvert+\left\lvert x^{-i}_{u}\right\rvert\Big)\,du\right)^{\!4}\!+\,\left(\,\sup\limits_{t\,\in\,[s,\,T]}\left\lvert\ \int\limits^{t}_{s}\!\!\mathds{1}_{\left\{u\leq \tau_{k}\right\}}\,\mathrm{v}^{i}\!\left(Y^{i}_{u} \right)dB_{u}\,\right\rvert\,\right)^{\!4}    
\end{equation*}

\par Next, in view of the above, we invoke the Burkholder\textendash Davis\textendash Gundy inequality \cite[Theorem IV.4.1]{revuz2013continuous} along with the definition of $\tau_{k}$ in order to arrive at the following
\begin{equation*}
\mathds{E}\left[\,\sup\limits_{t\,\in\,[s,\,T]}\mathds{1}_{\left\{t\,\leq\,\tau_{k}\right\}}\left\lvert\, Y^{i}_{t}\,\right\rvert^{4}\,\right]\,\leq\,C\left(1\,+\left\lvert\, y^{i}\,\right\rvert^{4}\,\right)\,+\,\mathds{E}\left[\ \left(\,\int\limits^{T}_{s}\!\!\mathds{1}_{\left\{u\leq \tau_{k}\right\}}\left\lvert\,\mathrm{v}^{i}\!\left(Y^{i}_{u} \right)\,\right\rvert^{2}\!\!du\,\right)^{\!2}\,\right]
\end{equation*}

\par Further, recall that by definition we have
\begin{equation*}
   \mathrm{v}^{i}\!\left(Y^{i}_{t}\right)\ =\ \sigma\left(S_{t}\right)\left(-\theta^{-i}\!,0,\,0,\,\pi^{i}_{t},\,\pi^{-i}_{t}\right)^{\mathrm{T}}\\
\end{equation*}

\par Given that the local volatility function $\sigma$ satisfies {\textcolor{azul-pesc}{\Cref{siglip}}}, we know that it is bounded above. Further, given the definition of $\tau_{k}$ it follows that $\mathds{1}_{\left\{t\leq\tau_{k}\right\}}\left\lvert\pi^{i}_{t}\right\rvert$ and $\mathds{1}_{\left\{t\leq\tau_{k}\right\}}\left\lvert\pi^{-i}_{t}\right\rvert$ are bounded above by $k$. The claim is then immediate in view of the above.

\par \vspace{0.5em}

\par \textbf{(ii)} Given that the functions $a^{i}$, $b^{i}$ are Lipschitz continuous, it follows in view of (\ref{stated}) that there exists a positive constant $C$ such that for $t \in [s,T]$ we have
\begin{multline}\label{momini}
      \mathds{1}_{\left\{ t\,\leq\,\tau_{k}\right\}}\left\lvert\,Y^{i}_{t} - \widehat{Y}^{i}_{t}\,\right\rvert\ \leq  \left\lvert\,y^{i} - \hat{y}^{i}\,\right\rvert\ +\ C\int\limits_{s}^{t}\!\!\mathds{1}_{\left\{ u\,\leq\,\tau_{k}\right\}} \left\lvert x^{-i}_{u}\right\rvert\,\left\lvert\,Y^{i}_{u} - \widehat{Y}^{i}_{u}\,\right\rvert du
      + C\int\limits_{s}^{t}\!\!\mathds{1}_{\left\{ u\,\leq\,\tau_{k}\right\}} \left\lvert x^{i}_{u}\right\rvert\,\left\lvert\,Y^{i}_{u} - \widehat{Y}^{i}_{u}\,\right\rvert du\\
      +\ \left\lvert\,\int\limits_{s}^{t}\!\!\mathds{1}_{\left\{u\,\leq\,\tau_{k}\right\}}\left(\mathrm{v}^{i}\!\left(Y^{i}_{u}\right) - \mathrm{v}^{i}\!\left(\widehat{Y}^{i}_{u}\right)\right)dB_{u}\,\right\rvert
\end{multline}

\par For the sake of notational convenience we let $f_{t}$ denote the random variable which is defined as
\begin{equation*}
     f_{t}\ =\ \left\lvert\,\int\limits_{s}^{t}\!\!\mathds{1}_{\left\{u\,\leq\,\tau_{k}\right\}}\left(\mathrm{v}^{i}\!\left(Y^{i}_{u}\right) - \mathrm{v}^{i}\!\left(\widehat{Y}^{i}_{u}\right)\right)dB_{u}\,\right\rvert
\end{equation*}

\par Further, we let $\bigtriangledown y = \left\lvert\, y^{i}-\hat{y}^{i}\,\right\rvert$. With the help of the notation introduced above, we can thus rewrite (\ref{momini}) succinctly as
\begin{equation*}
    \mathds{1}_{\left\{ t\,\leq\,\tau_{k}\right\}}\left\lvert\,Y^{i}_{t} - \widehat{Y}^{i}_{t}\,\right\rvert\ \leq\, \left(\bigtriangledown y+f_{t}\right)\,
     + \int\limits_{s}^{t}\!\!\mathds{1}_{\left\{ u\,\leq\,\tau_{k}\right\}}\,\Big(\left\lvert x^{i}_{u}\right\rvert + \left\lvert x^{-i}_{u}\right\rvert\Big)\left\lvert\,Y^{i}_{u} - \widehat{Y}^{i}_{u}\,\right\rvert du
\end{equation*}

\par Next, we invoke Gronwall's lemma in conjunction with the definition of $\tau_{k}$ which implies that there exists a positive constant $C$ (dependent on $k$ and $T$) such that we have
\begin{equation*}
    \mathds{1}_{\left\{ t\,\leq\,\tau_{k}\right\}}\left\lvert\,Y^{i}_{t} - \widehat{Y}^{i}_{t}\,\right\rvert\,\leq\, \left(\bigtriangledown y+f_{t}\right)\,
     +\,C\!\mathlarger{\int}\limits_{s}^{t}\!\!\mathds{1}_{\left\{ u\,\leq\,\tau_{k}\right\}}\Big(\left\lvert x^{i}_{u}\right\rvert + \left\lvert x^{-i}_{u}\right\rvert\Big)\left(\bigtriangledown y+f_{u}\right)du\,\leq\,C\left(\bigtriangledown y+\sup\limits_{u\,\in\,[s,\,t]}\!f_{u}\right)
\end{equation*}

\par From the equation above it then follows that there exists a positive constant $C$ such that
\begin{equation*}
    \sup\limits_{t\,\in\,[s,\,T]}\mathds{1}_{\left\{ t\,\leq\,\tau_{k}\right\}}\left\lvert\,Y^{i}_{t} - \widehat{Y}^{i}_{t}\,\right\rvert\ \leq\, C\left(\bigtriangledown y+\sup\limits_{t\,\in\,[s,\,T]}\!f_{t}\right)
\end{equation*}

\par Further, in view of the equation above, and given the definitions of the random variables $\bigtriangledown y$, $f_{t}$, and the $\mathds{F}$-stopping time $\tau_{k}$, we can find a positive constant $C$ such that we have
\begin{equation*}
    \mathds{E}\left[\,\sup\limits_{t\,\in\,[s,\,T]}\mathds{1}_{\left\{ t\,\leq\,\tau_{k}\right\}}\left\lvert\,Y^{i}_{t} - \widehat{Y}^{i}_{t}\,\right\rvert^{2}\right]\,\leq \ C\left\lvert\, y^{i} - \hat{y}^{i}\,\right\rvert^{2}\, 
    +\,C\mathds{E}\left[\,\sup\limits_{t\,\in\,[s,\,T]} \left\lvert\,\int\limits_{s}^{t}\!\!\mathds{1}_{\left\{u\,\leq\,\tau_{k}\right\}}\left(\mathrm{v}^{i}\!\left(Y^{i}_{u}\right) - \mathrm{v}^{i}\!\left(\widehat{Y}^{i}_{u}\right)\right)dB_{u}\,\right\rvert\,\right]^{2}
\end{equation*}

\par Moreover, since the local volatility function $\sigma$ is Lipschitz continuous in view of {\textcolor{azul-pesc}{\Cref{siglip}}}, it follows in view of the definition of $\mathrm{v}^{i}$ and the $\mathds{F}$-stopping time $\tau_{k}$ that the function $\mathrm{v}^{i}$ satisfies Lipschitz continuity on $[s,\tau_{k}]$. By way of Burkholder\textendash Davis\textendash Gundy inequality we can then find a positive constant $C$ such that
\begin{equation*}
    \mathds{E}\left[\,\sup\limits_{t\,\in\,[s,\,T]}\mathds{1}_{\left\{ t\,\leq\,\tau_{k}\right\}}\left\lvert\,Y^{i}_{t} - \widehat{Y}^{i}_{t}\,\right\rvert^{2}\right]\,\leq \ C\left\lvert\, y^{i} - \hat{y}^{i}\,\right\rvert^{2}\, 
    +\,C\mathds{E}\left[\,\int\limits_{s}^{T}\!\!\mathds{1}_{\left\{ u\,\leq\,\tau_{k}\right\}}\left\lvert\,Y^{i}_{u} - \widehat{Y}^{i}_{u}\,\right\rvert^{2}\!\!du\,\right]
\end{equation*}

\par In view of Tonelli's theorem it is then immediate from the equation above that we have
\begin{equation*}
    \mathds{E}\left[\,\sup\limits_{t\,\in\,[s,\,T]}\mathds{1}_{\left\{ t\,\leq\,\tau_{k}\right\}}\left\lvert\,Y^{i}_{t} - \widehat{Y}^{i}_{t}\,\right\rvert^{2}\right]\,\leq\ C\left\lvert\, y^{i} - \hat{y}^{i}\,\right\rvert^{2}\, 
    +\, C\!\mathlarger{\int}\limits_{s}^{T}\!\mathds{E}\left[\,\sup\limits_{u\,\in\,[s,\,T]}\mathds{1}_{\left\{ u\,\leq\,\tau_{k}\right\}}\left\lvert\,Y^{i}_{u} - \widehat{Y}^{i}_{u}\,\right\rvert^{2}\right]\!du
\end{equation*}

\par Given the equation above, we invoke Gronwall's lemma again to find a positive constant $C>0$, such that we have
\begin{equation*}
    \mathds{E}\left[\sup\limits_{t\,\in\,[s,\,T]}\mathds{1}_{\left\{ t\,\leq\,\tau_{k}\right\}}\left\lvert\,Y^{i}_{t} - \widehat{Y}^{i}_{t}\,\right\rvert^{2}\right]\,\leq\ C\left\lvert\, y^{i} - \hat{y}^{i}\,\right\rvert^{2}\, 
\end{equation*}

\par Note that the choice of $k$ above was arbitrary which in conjunction with the fact that $X^{-i}\in\mathcal{A}_{s}$, $X^{i}\in\mathcal{A}_{s}$ implies that we can find $k\in\mathds{N}$ such that $T\leq\tau_{k}$, whereupon the claim follows immediately.

\par \vspace{0.5em}

\par \textbf{(iii)} Since $0\leq \hat{s}< s \leq T$, it follows in view of (\ref{stated}) that we have
\begin{align*}
    \mathds{1}_{\left\{ t\,\leq\,\tau_{k}\right\}}\left\lvert\;Y^{i,\,s}_{T} - Y^{i,\,\hat{s}}_{T}\;\right\rvert \leq &\int\limits_{\hat{s}\,+\,(T-s)}^{T}\!\!\!\!\mathds{1}_{\left\{ u\,\leq\,\tau_{k}\right\}}\left\lvert\,a^{i}\!\left(Y^{i,\,s}_{u}\right)\,\right\rvert\, \left\lvert x^{-i}_{u}\right\rvert du\,
     + \int\limits_{\hat{s}\,+\,(T-s)}^{T}\!\!\!\!\mathds{1}_{\left\{ u\,\leq\,\tau_{k}\right\}} \left\lvert\,b^{i}\!\left(Y^{i,\,s}_{u}\right)\,\right\rvert\,\left\lvert x^{i}_{u}\right\rvert du\\
    & + \left\lvert\ \ \int\limits_{\,\hat{s}\,+\,(T-s)}^{T}\!\!\!\!\mathds{1}_{\left\{u\,\leq\,\tau_{k}\right\}}\mathrm{v}^{i}\!\left(Y^{i,\,s}_{u}\right)dB_{u}\ \right\rvert
\end{align*}

\par Given the equation above it follows that upon rearrangement and taking expectations we have
\begin{align*}
  \begin{split}
    \mathds{E}\left[\mathds{1}_{\left\{ t\,\leq\,\tau_{k}\right\}}\left\lvert\;Y^{i,\,s}_{T} - Y^{i,\,\hat{s}}_{T}\;\right\rvert^{2}\right]\!\leq &\, \mathds{E}\left[\int\limits_{0}^{s-\hat{s}}\!\!\mathds{1}_{\left\{ u\,\leq\,\tau_{k}\right\}}\left\lvert\,a^{i}\!\left(Y^{i,\,s}_{u}\right)\,\right\rvert\, \left\lvert x^{-i}_{u}\right\rvert du\right]^{2}\!\!
     +\mathds{E}\left[\int\limits_{0}^{s-\hat{s}}\!\!\mathds{1}_{\left\{ u\,\leq\,\tau_{k}\right\}} \left\lvert\,b^{i}\!\left(Y^{i,\,s}_{u}\right)\,\right\rvert\,\left\lvert x^{i}_{u}\right\rvert du\right]^{2}\\
    & + \mathds{E}\left[\sup\limits_{\,t\, \in \,[\,0,\,s-\hat{s}\,]\,}\left\lvert\ \, \int\limits_{0}^{t}\!\!\mathds{1}_{\left\{u\,\leq\,\tau_{k}\right\}}\mathrm{v}^{i}\!\left(Y^{i,\,s}_{u}\right)dB_{u}\ \right\rvert\ \right]^{2}
  \end{split}
\end{align*}

\par Next, we recall the definition of the vector-valued functions $a^{i}$, $b^{i}$ and the $\mathds{F}$-stopping time $\tau_{k}$, whereby it follows that $\mathds{1}_{\left\{ u\,\leq\,\tau_{k}\right\}} \left\lvert\,a^{i}(Y^{i,\,s}_{u})\,\right\rvert$ and $\mathds{1}_{\left\{ u\,\leq\,\tau_{k}\right\}} \left\lvert\,b^{i}(Y^{i,\,s}_{u})\,\right\rvert$ are bounded above. This fact in conjunction with Burkholder\textendash Davis\textendash Gundy inequality then implies that there exists a positive constant $C$ (dependent on $k$ and $T$) such that the equation above can be re-written as 
\begin{align*}
  \begin{split}
    \mathds{E}\left[\mathds{1}_{\left\{ t\,\leq\,\tau_{k}\right\}}\left\lvert\;Y^{i,\,s}_{T} - Y^{i,\,\hat{s}}_{T}\;\right\rvert^{2}\right]\, \leq\, & \, C\mathds{E}\left[\,\int\limits_{0}^{s-\hat{s}}\!\!\mathds{1}_{\left\{ u\,\leq\,\tau_{k}\right\}}\left\lvert x^{-i}_{u}\right\rvert du\,\right]^{2}\!+\,C\mathds{E}\left[\,\int\limits_{0}^{s-\hat{s}}\!\!\mathds{1}_{\left\{ u\,\leq\,\tau_{k}\right\}} \left\lvert x^{i}_{u}\right\rvert du\,\right]^{2}\\
    & +\, C\mathds{E}\left[\,\int\limits_{0}^{s-\hat{s}}\!\!\mathds{1}_{\left\{u\,\leq\,\tau_{k}\right\}}\left\lvert\,\mathrm{v}^{i}\!\left(Y^{i,\,s}_{u}\right)\,\right\rvert^{2}\!du\,\right]
  \end{split}
\end{align*}

\par Given the equation above, we consider the limit as $\hat{s}\uparrow s$ which then leads us to
\begin{align}\label{timlim}
\begin{split}
    \lim\limits_{\hat{s}\,\uparrow\,s}\,\mathds{E}\left[\mathds{1}_{\left\{ t\,\leq\,\tau_{k}\right\}}\left\lvert\;Y^{i,\,s}_{T} - Y^{i,\,\hat{s}}_{T}\;\right\rvert^{2}\right]\, \leq\,& \,  \lim\limits_{\hat{s}\,\uparrow\,s}\,C\mathds{E}\left[\,\int\limits_{0}^{s-\hat{s}}\!\!\mathds{1}_{\left\{ u\,\leq\,\tau_{k}\right\}}\left\lvert x^{-i}_{u}\right\rvert du\,\right]^{2}\\
    & +\ \lim\limits_{\hat{s}\,\uparrow\,s}\, C\mathds{E}\left[\,\int\limits_{0}^{s-\hat{s}}\!\!\mathds{1}_{\left\{ u\,\leq\,\tau_{k}\right\}} \left\lvert x^{i}_{u}\right\rvert du\,\right]^{2}\\
     & +\, \lim\limits_{\hat{s}\,\uparrow\,s}\,C\mathds{E}\left[\,\int\limits_{0}^{s-\hat{s}}\!\!\mathds{1}_{\left\{u\,\leq\,\tau_{k}\right\}}\left\lvert\,\mathrm{v}^{i}\!\left(Y^{i,\,s}_{u}\right)\,\right\rvert^{2}\!du\,\right]
\end{split}
\end{align}

\par Further, in view of the definition of the vector-valued function $\mathrm{v}^{i}$, and the $\mathds{F}$-stopping time $\tau_{k}$ it follows that $\mathds{1}_{\left\{ u\,\leq\,\tau_{k}\right\}} \left\lvert\,\mathrm{v}^{i}(Y^{i,\,s}_{u})\,\right\rvert$ is bounded above, which in conjunction with the fact that $X^{i}, X^{-i} \in \mathcal{A}_{s}$ then implies that there exists a positive constant $C$ (dependent on $k$ and $T$) such that we have
\begin{equation*}
   \left[\,\int\limits_{0}^{s-\hat{s}}\!\!\mathds{1}_{\left\{ u\,\leq\,\tau_{k}\right\}}\left\lvert x^{-i}_{u}\right\rvert du\,\right]^{2} \text{and}\ \left[\,\int\limits_{0}^{s-\hat{s}}\!\!\mathds{1}_{\left\{ u\,\leq\,\tau_{k}\right\}} \left\lvert x^{i}_{u}\right\rvert du\,\right]^{2} \leq\ k^{2},\ \left[\,\int\limits_{0}^{s-\hat{s}}\!\!\mathds{1}_{\left\{u\,\leq\,\tau_{k}\right\}}\left\lvert\,\mathrm{v}^{i}\!\left(Y^{i,\,s}_{u}\right)\,\right\rvert^{2} du\,\right]\
     \leq\ C
\end{equation*}

\par Thus, we then invoke Dominated Convergence Theorem whereupon (\ref{timlim}) implies that we have
\begin{align*}
    \lim\limits_{\hat{s}\,\uparrow\,s}\mathds{E}\left[\mathds{1}_{\left\{ t\,\leq\,\tau_{k}\right\}}\left\lvert\;Y^{i,\,s}_{T} - Y^{i,\,\hat{s}}_{T}\;\right\rvert^{2}\right]\! \leq &\, C\mathds{E}\left[\,\lim\limits_{\hat{s}\,\uparrow\,s}\left(\,\int\limits_{0}^{s-\hat{s}}\!\!\mathds{1}_{\left\{ u\,\leq\,\tau_{k}\right\}}\left\lvert x^{-i}_{u}\right\rvert du\,\right)^{2}\right]+ C\mathds{E}\left[\lim\limits_{\hat{s}\,\uparrow\,s}\left(\,\int\limits_{0}^{s-\hat{s}}\!\!\mathds{1}_{\left\{ u\,\leq\,\tau_{k}\right\}} \left\lvert x^{i}_{u}\right\rvert du\,\right)^{2}\right]\\
    & + C\mathds{E}\left[\lim\limits_{\hat{s}\,\uparrow\,s}\int\limits_{0}^{s-\hat{s}}\!\!\mathds{1}_{\left\{u\,\leq\,\tau_{k}\right\}}\left\lvert\,\mathrm{v}^{i}\!\left(Y^{i,\,s}_{u}\right)\,\right\rvert^{2}\!du\,\right]
\end{align*}

\par Since $X^{i}, X^{-i} \in \mathcal{A}_{s}$, it follows immediately from the continuity of the Lebesgue integral that we have
\begin{equation*}
   \lim\limits_{\hat{s}\,\uparrow\,s}\left[\,\int\limits_{0}^{s-\hat{s}}\!\!\mathds{1}_{\left\{ u\,\leq\,\tau_{k}\right\}}\left\lvert x^{-i}_{u}\right\rvert du\,\right]\ =\ 0\quad \text{and}\quad 
   \lim\limits_{\hat{s}\,\uparrow\,s}\left[\,\int\limits_{0}^{s-\hat{s}}\!\!\mathds{1}_{\left\{ u\,\leq\,\tau_{k}\right\}} \left\lvert x^{i}_{u}\right\rvert du\,\right]\ =\ 0
\end{equation*}

\par Also, as established earlier, we know that $\mathds{1}_{\left\{ u\,\leq\,\tau_{k}\right\}} \left\lvert\,\mathrm{v}^{i}(Y^{i,\,s}_{u})\,\right\rvert$ is bounded above which implies the existence of a positive constant $C$ (dependent on $k$ and $T$) such that we have
\begin{equation*}
   \lim\limits_{\hat{s}\,\uparrow\,s}\,\int\limits_{0}^{s-\hat{s}}\!\!\mathds{1}_{\left\{u\,\leq\,\tau_{k}\right\}}\left\lvert\,\mathrm{v}^{i}\!\left(Y^{i,\,s}_{u}\right)\,\right\rvert^{2}\!du\ \leq\ \lim\limits_{\hat{s}\,\uparrow\,s}\,C\left(s-\hat{s}\right)\  =\ 0
\end{equation*}

\par The claim follows immediately in view of the above.
\end{proof}

\subsection{Proof of \textcolor{azul-pesc}{\Cref{jconv}}}

\begin{proof}
\par $\mathbf{(i)}$ Given the assumption on the functional form of the utility function, it follows that we have $-\infty\leq J^{i}(\cdot)\leq 0$. Thus, in order to establish non-degeneracy of the best-response value function of the $i$th investor, given $X^{-i}\in \mathcal{A}_{s}$, it suffices to show that $J^{i}(\cdot) > -\infty$. To this end, given a fixed initial time $s\in[0,T]$, we define the deterministic strategy $\widehat{X}^{i} = \left\{\hat{x}^{i}_{t}\right\}_{t\,\in\,[s,T]}$ for the $i$th investor, where $\hat{x}^{i}_{t}\, =\, 0,$ for $t\,\in\,[s,T]$.

\par It readily follows from {\textcolor{azul-pesc}{\Cref{admcon}}} that $\widehat{X}^{i}\in\mathcal{A}_{s}$. Additionally, let $\widehat{W}^{i}_{T}$ denote the terminal portfolio value of the $i$th investor corresponding to the strategy $\widehat{X}^{i}.$ Moreover, recall that $\pi^{i}_{s}\in\mathds{R}$ represents the initial stock holdings of the $i$th investor, which is assumed to be given. It then follows in view of the definition of $\widehat{W}^{i}_{T}$ that we have
\begin{equation*}
    \mathds{E}\left[u^{i}\!\left(\widehat{W}^{i}_{T}\right)\right] =\, \mathds{E}\left[-\frac{1}{\delta^{i}}\exp{\left(-\delta^{i}W^{i}_{s} \ -\delta^{i}\pi^{i}_{s}\!\!\int_{s}^{T}\!\!\!\!dS_{t}\right)}\right]
\end{equation*}

\par Further, recalling (\ref{price}) in conjunction with the fact that $\hat{x}^{i}_{t} = 0$ for $t\in[s,T]$, we can re-write the above as follows, where $C$ denotes a non-positive constant
\begin{equation*}
    \mathds{E}\left[u^{i}\!\left(\widehat{W}^{i}_{T}\right)\right] =\, C\mathds{E}\left[\exp{\left(\delta^{i}\pi^{i}_{s}\theta^{-i}\!\!\int_{s}^{T}\!\!\!\!x^{-i}_{t}dt -\delta^{i}\pi^{i}_{s}\!\!\int_{s}^{T}\!\!\!\!\sigma\!\left(S_{t}\right)dB_{t}\right)}\right]
\end{equation*}

\par From the above, it is immediate that in order to establish that the expected utility of the $i$th investor corresponding to the deterministic strategy $\hat{X}^{i}$ is finite, it suffices to show that
\begin{equation*}
 \mathds{E}\left[\exp{\left(\delta^{i}\pi^{i}_{s}\theta^{-i}\!\!\int_{s}^{T}\!\!\!\!x^{-i}_{t}dt -\delta^{i}\pi^{i}_{s}\!\!\int_{s}^{T}\!\!\!\!\sigma\!\left(S_{t}\right)dB_{t}\right)}\right]< \infty  
\end{equation*}

\par To this end, first note that by way of Cauchy\textendash Schwarz inequality we have
\begin{align*}
\begin{split}
     \mathds{E}\left[\exp{\left(\delta^{i}\pi^{i}_{s}\theta^{-i}\!\!\int_{s}^{T}\!\!\!\!x^{-i}_{t}dt -\delta^{i}\pi^{i}_{s}\!\!\int_{s}^{T}\!\!\!\!\sigma\!\left(S_{t}\right)dB_{t}\right)}\right] \leq &\, \mathds{E}\left[\exp{\left(2\delta^{i}\pi^{i}_{s}\theta^{-i}\!\!\int_{s}^{T}\!\!\!\!x^{-i}_{t}dt\right)}\right]\\
     &\times\mathds{E}\left[\exp{\left(-2\delta^{i}\pi^{i}_{s}\!\!\int_{s}^{T}\!\!\!\!\sigma\!\left(S_{t}\right)dB_{t}\right)}\right] 
\end{split}
\end{align*}

\par Moreover, in view of the fact that $X^{-i} \in \mathcal{A}_{s}$ it follows immediately that we have
\begin{equation*}
    \mathds{E}\left[\exp{\left(2\delta^{i}\pi^{i}_{s}\theta^{-i}\!\!\int_{s}^{T}\!\!\!\!x^{-i}_{t}dt \right)}\right]<\infty
\end{equation*}

\par Further, we invoke \cite[Theorem 43, Page 140]{protter2005stochastic} in conjunction with the fact that the local volatility function $\sigma$ satisfies \textcolor{azul-pesc}{\Cref{siglip}} so as to obtain
\begin{equation*}    \mathds{E}\left[\exp{\left(-2\delta^{i}\pi^{i}_{s}\!\!\int_{s}^{T}\!\!\!\!\sigma\!\left(S_{t}\right)dB_{t}\right)}\right]\,\leq\, \mathds{E}\left[\exp{\left(\frac{1}{2}\left(4\delta^{i}\pi^{i}_{s}\right)^{2}\!\!\int_{s}^{T}\!\!\!\!\sigma^{2}\!\left(S_{t}\right)dt\right)}\right] < \infty
\end{equation*}

\par From the above it is then immediate that we have $J^{i}(\cdot)\geq\mathds{E}\left[u^{i}\!\left(\widehat{W}^{i}_{T}\right)\right]>-\infty$.

\par \vspace{0.5em}

\par $\mathbf{(ii)}$ We note that the sequence of functions $\left\{{J}^{i,n}\!\left(s,y^{i};X^{-i}\right)\right\}_{n\,\in\,\mathds{N}}$ is non-decreasing, and is bounded above by ${J}^{i}\!\left(s,y^{i};X^{-i}\right)$, which then implies that we have
\begin{equation}\label{jsup}
\limsup\limits_{n\,\rightarrow\,\infty}\,{J}^{i,n}\!\left(s,y^{i};X^{-i}\right)\,\leq\, J^{i}\!\left(s,y^{i};X^{-i}\right)
\end{equation}

\par Next, for a given $\rho > 0$, we consider a $\rho$-optimal control for the original control problem. That is, we consider $X^{i,\,\rho}\, \in\, \mathcal{A}_{s}$ such that $X^{i,\,\rho}$ satisfies the following inequality 
\begin{equation}\label{jrho}
J^{i}\!\left(s,y^{i};X^{-i}\right)\,-\,\rho\, \leq\, U^{i}\!\left(X^{i,\,\rho},\,X^{-i}\right)
\end{equation}

\par The existence of a non-trivial $X^{i,\,\rho}$ follows from the definition of $J^{i}$, and the fact that $J^{i}(\cdot)\, >\, -\infty$. Then, given $X^{i,\,\rho}$ we define a sequence of controls $\{X^{i,\,\rho_n}=\{x^{i,\,\rho_n}_{t}\}_{\,t\, \in\, [s,T]}\}_{n\,\in\,\mathds{N}}$ as follows
\begin{equation*}
    x^{i,\,\rho_n}_{t}\,=\,\begin{cases} \,x^{i,\,\rho}_{t},\, & \text{if}\,\left\lvert x^{i,\,\rho}_{t}\right\rvert\,\leq\, n\\ \,n,\,& \text{if}\,\left\lvert x^{i,\,\rho}_{t}\right\rvert \,>\, n
    \end{cases}
\end{equation*}

\par It is immediate from the definition above that $X^{i,\,\rho_n} \in \mathcal{A}^{n}_{s}$. Also, it follows that the sequence $\lvert x^{i,\,\rho_n}_{t}\rvert$ is dominated by $\lvert x^{i,\,\rho}_{t}\rvert$. Moreover, we have that $X^{i,\,\rho_n}\uparrow X^{i,\,\rho}$, as $n\rightarrow\infty$, which in conjunction with the fact that $X^{i,\,\rho}\, \in\, \mathcal{A}_{s}$, and Dominated Convergence Theorem then implies that we have $\mathds{P}$-almost surely
\begin{equation}\label{xconv}
\lim\limits_{n\,\rightarrow\,\infty}\,\mathds{E}\left[\int\limits^{T}_{s}\!\!\big\lvert\, x^{i,\,\rho_n}_{t}\,-\, x^{i,\,\rho}_{t}\,\big\rvert\, dt\right]\,=\,0 
\end{equation}

\par Next, for a given $X^{-i} \in \mathcal{A}_{s}$ we consider the sequence of stochastic processes $\{Y^{i,\,\rho_n}\}_{n\,\in\,\mathds{N}}$, where $Y^{i,\,\rho_n}$ corresponds to the strong solution of the system of stochastic differential equations (\ref{stated}), with $x^{i}_{t}\, =\, x^{i,\,\rho_n}_{t}$ for all $t\, \in\, [s,T]$. Further, given $k\,\in\,\mathds{N}$, we define the $\mathds{F}$-stopping time $\tau_{k,\,x^{i}}$, $\tau_{k,\,x^{-i}}$ and $\tau_{k}$ as follows
\begin{equation}\label{stop}
\tau_{k,\,x^{i}} = \inf\left\{t > s:\, \int\limits_{s}^{t}\!\!\big\lvert\, x^{i,\,\rho}_{u}\,\big\rvert\, du\, >\, k\right\}, \tau_{k,\,x^{-i}} = \inf\left\{t > s:\, \int\limits_{s}^{t}\!\!\left\lvert\, x^{-i}_{u}\,\right\rvert du\, >\, k\right\}, \tau_{k} = \tau_{k,\,x^{i}}\,\wedge\,\tau_{k,\,x^{-i}}
\end{equation}

\par Note that since $\lvert x^{i,\,\rho_n}_{t}\rvert$ is dominated by $\lvert x^{i,\,\rho}_{t}\rvert$, it then follows in view of the definition above that for $t<\tau_{k}$ we have the following  
\begin{equation*}
\int\limits^{t}_{s}\!\!\big\lvert\,x^{i,\,\rho_n}_{u}\,\big\rvert\, du\, \leq \int\limits^{t}_{s}\!\!\big\lvert\,x^{i,\,\rho}_{u}\,\big\rvert\, du\, \leq  k
\end{equation*}

\par We next prove that the sequence $\{Y^{i,\,\rho_n}\}_{n\,\in\,\mathds{N}}$ converges uniformly in expectation to $Y^{i,\,\rho}$ where the latter corresponds to the strong solution of the system of stochastic differential equations (\ref{stated}), with $x^{i}_{t}\,=\,x^{i,\,\rho}_{t}$, for all $t\, \in\, [s,T]$. To this end, note that in view of the definition of $Y^{i,\,\rho_n}$ and $Y^{i,\,\rho}$, and given the fact that the functions $a^{i}(Y^{i})$, $b^{i}(Y^{i})$ are Lipschitz continuous, we obtain the following
\begin{align}\label{momlon}
  \begin{split}
    \mathds{1}_{\left\{ t\,\leq\,\tau_{k}\right\}}\left\lvert\,Y^{i,\,\rho}_{t} - Y^{i,\,\rho_n}_{t}\,\right\rvert\ \leq\,&\, \int\limits_{s}^{t}\!\!\mathds{1}_{\left\{ u\,\leq\,\tau_{k}\right\}} \left\lvert x^{-i}_{u}\right\rvert\,\left\lvert\,Y^{i,\,\rho}_{u} - Y^{i,\,\rho_n}_{u}\right\rvert du\\
    & + \left\lvert\,\int\limits_{s}^{t}\!\!\mathds{1}_{\left\{u\,\leq\,\tau_{k}\right\}}\left(\mathrm{v}^{i}\!\left(Y^{i,\,\rho}_{u}\right) - \mathrm{v}^{i}\!\left(Y^{i,\,\rho_n}_{u}\right)\right)dB_{u}\,\right\rvert \\
    & + C\!\int\limits_{s}^{t}\!\!\mathds{1}_{\left\{ u\,\leq\,\tau_{k}\right\}}\left(1 + \left\lvert Y^{i,\,\rho}_{u}\right\rvert\,\right)\,\left\lvert\, x^{i,\,\rho}_{u} - x^{i,\,\rho_n}_{u}\right\rvert du\\
    & + C\!\int\limits_{s}^{t}\!\!\mathds{1}_{\left\{ u\,\leq\,\tau_{k}\right\}}\left\lvert x^{i,\,\rho_n}_{u}\right\rvert\,\left\lvert Y^{i,\,\rho}_{u} - Y^{i,\,\rho_n}_{u}\right\rvert du
  \end{split}
\end{align}

\par For the sake of notational convenience, we let $f^{n}_{t}$ and $l^{n}_{t}$ denote random variables which are defined as follows
\begin{align*}
     & f^{n}_{t}\ =\ \left\lvert\,\int\limits_{s}^{t}\!\!\mathds{1}_{\left\{u\,\leq\,\tau_{k}\right\}}\left(\mathrm{v}^{i}\!\left(Y^{i,\,\rho}_{u}\right) - \mathrm{v}^{i}\!\left(Y^{i,\,\rho_n}_{u}\right)\right)dB_{u}\,\right\rvert \\
     & l^{n}_{t}\ =\ C\!\int\limits_{s}^{t}\!\!\mathds{1}_{\left\{ u\,\leq\,\tau_{k}\right\}}\left(1 + \left\lvert Y^{i,\,\rho}_{u}\right\rvert\,\right)\,\left\lvert\, x^{i,\,\rho}_{u} - x^{i,\,\rho_n}_{u}\right\rvert du
\end{align*}

\par With the help of the notation introduced above, we can rewrite (\ref{momlon}) succinctly as follows
\begin{equation*}
    \mathds{1}_{\left\{ t\,\leq\,\tau_{k}\right\}}\left\lvert\,Y^{i,\,\rho}_{t} - Y^{i,\,\rho_n}_{t}\,\right\rvert\ \leq\, \left(f^{n}_{t}\,+\,l^{n}_{t}\right)\,
     +\,C\!\int\limits_{s}^{t}\!\!\mathds{1}_{\left\{ u\,\leq\,\tau_{k}\right\}}\Big(\left\lvert x^{i,\,\rho_n}_{u}\right\rvert + \left\lvert x^{-i}_{u}\right\rvert\Big)\left\lvert\, Y^{i,\,\rho}_{u} - Y^{i,\,\rho_n}_{u}\right\rvert du
\end{equation*}

\par Next, we invoke Gronwall's lemma in conjunction with the definition of $\tau_{k}$ in order to obtain a positive constant $C$ (dependent on $k$ and $T$) such that we have
\begin{equation*}
    \mathds{1}_{\left\{ t\,\leq\,\tau_{k}\right\}}\left\lvert\,Y^{i,\,\rho}_{t} - Y^{i,\,\rho_n}_{t}\,\right\rvert\leq \left(f^{n}_{t}\,+\,l^{n}_{t}\right)
     +\,C\!\mathlarger{\int}\limits_{s}^{t}\!\!\mathds{1}_{\left\{ u\,\leq\,\tau_{k}\right\}}\Big(\left\lvert x^{i,\,\rho_n}_{u}\right\rvert + \left\lvert x^{-i}_{u}\right\rvert\Big)\left(f^{n}_{u}\,+\,l^{n}_{u}\right)du \leq C\!\!\sup\limits_{u\,\in\,[s,\,t]}\!\!\left(f^{n}_{u}\,+\,l^{n}_{u}\right)
\end{equation*}

\par From the equation above, it then follows that there exists a positive constant $C$ such that we have
\begin{equation*}
    \sup\limits_{t\,\in\,[s,\,T]}\mathds{1}_{\left\{ t\,\leq\,\tau_{k}\right\}}\left\lvert\,Y^{i,\,\rho}_{t} - Y^{i,\,\rho_n}_{t}\,\right\rvert\ \leq\, C\!\!\sup\limits_{t\,\in\,[s,\,T]}\!\!\left(f^{n}_{t}\,+\,l^{n}_{t}\right)\\
\end{equation*}

\par In view of the equation above, and given the definitions of the random variables $f^{n}_{t}$, $l^{n}_{t}$, and the $\mathds{F}$-stopping time $\tau_{k}$, we can invoke Cauchy\textendash Schwarz inequality to assert the existence of a positive constant $C$ such that the following holds
\begin{align*}
\begin{split}
    \mathds{E}\left[\sup\limits_{t\,\in\,[s,\,T]}\mathds{1}_{\left\{ t\,\leq\,\tau_{k}\right\}}\left\lvert\,Y^{i,\,\rho}_{t} - Y^{i,\,\rho_n}_{t}\,\right\rvert^{2}\right]\,\leq\, &\, C\mathds{E}\left[\int\limits_{s}^{T}\!\!\mathds{1}_{\left\{ u\,\leq\,\tau_{k}\right\}}\left(1 + \left\lvert Y^{i,\,\rho}_{u}\right\rvert\,\right)\left\lvert\, x^{i,\,\rho}_{u} - x^{i,\,\rho_n}_{u}\right\rvert du\right]^{2}\\ 
    & +\, C\mathds{E}\left[\,\sup\limits_{t\,\in\,[s,\,T]} \left\lvert\,\int\limits_{s}^{t}\!\!\mathds{1}_{\left\{u\,\leq\,\tau_{k}\right\}}\left(\mathrm{v}^{i}\!\left(Y^{i,\,\rho}_{u}\right) - \mathrm{v}^{i}\!\left(Y^{i,\,\rho_n}_{u}\right)\right)dB_{u}\,\right\rvert\,\right]^{2}
\end{split}
\end{align*}

\par Further, since the local volatility function $\sigma$ satisfies {\textcolor{azul-pesc}{\Cref{siglip}}}, it is Lipschitz continuous. It then follows in view of the definition of $\mathrm{v}^{i}$ and the $\mathds{F}$-stopping time $\tau_{k}$ that the function $\mathrm{v}^{i}$ satisfies Lipschitz continuity on $[s,\tau_{k}]$. Thus, by way of H\"{o}lder's inequality and Burkholder\textendash Davis\textendash Gundy inequality we can find a positive constant $C$ such that the following holds
\begin{align*}
    \mathds{E}\left[\sup\limits_{t\,\in\,[s,\,T]}\mathds{1}_{\left\{ t\,\leq\,\tau_{k}\right\}}\left\lvert\,Y^{i,\,\rho}_{t} - Y^{i,\,\rho_n}_{t}\,\right\rvert^{2}\right]\, \leq\,&\,  C\mathds{E}\left[\int\limits_{s}^{T}\!\!\mathds{1}_{\left\{ u\,\leq\,\tau_{k}\right\}}\left\lvert\,Y^{i,\,\rho}_{u} - Y^{i,\,\rho_n}_{u}\right\rvert^{2}\!\!du\right]\\
    +\,& \, C\,\sqrt{\mathds{E}\left[1 + \sup\limits_{t\,\in\,[s,\,T]}\mathds{1}_{\left\{ t\,\leq\,\tau_{k}\right\}}\left\lvert Y^{i,\,\rho}_{t}\right\rvert^{4}\,\right]}\\
    & \times\,\sqrt{\mathds{E}\left[\int\limits_{s}^{T}\!\!\mathds{1}_{\left\{ u\,\leq\,\tau_{k}\right\}}\left\lvert\, x^{i,\,\rho}_{u} - x^{i,\,\rho_n}_{u}\right\rvert du\right]}
\end{align*}

\par In view of Tonelli's theorem and {\textcolor{azul-pesc}{\Cref{momin}}}, it is then immediate from the above that we have
\begin{align*}
    \mathds{E}\left[\sup\limits_{t\,\in\,[s,\,T]}\mathds{1}_{\left\{ t\,\leq\,\tau_{k}\right\}}\left\lvert\,Y^{i,\,\rho}_{t} - Y^{i,\,\rho_n}_{t}\,\right\rvert^{2}\right]\,\leq\ &\ C\!\!\mathlarger{\int}\limits_{s}^{T}\!\!\mathds{E}\left[\sup\limits_{u\,\in\,[s,\,T]}\mathds{1}_{\left\{ u\,\leq\,\tau_{k}\right\}}\left\lvert\,Y^{i,\,\rho}_{u} - Y^{i,\,\rho_n}_{u}\right\rvert^{2}\right]du\\ 
    & +\,C\,\sqrt{1 + \left\lvert Y^{i,\,\rho}_{s}\right\rvert^{4}}\,\sqrt{\mathds{E}\left[\int\limits_{s}^{T}\!\!\mathds{1}_{\left\{ u\,\leq\,\tau_{k}\right\}}\left\lvert\, x^{i,\,\rho}_{u} - x^{i,\,\rho_n}_{u}\right\rvert du\right]}
\end{align*}

\par Given the equation above, we invoke Gronwall's lemma and consider the limit as $n \rightarrow \infty$, where the desired convergence of the sequence $\{Y^{i,\,n}\}_{n\,\in\,\mathds{N}} $ to $Y^{i,\,\rho}$ for $t \in [s,\tau_{k}]$ follows immediately in view of (\ref{xconv}). Moreover, since the choice of $k$ above was arbitrary, and given that $X^{-i}\in\mathcal{A}_{s}$, $X^{i,\,\rho_n}\in\mathcal{A}^{n}_{s}\subseteq\mathcal{A}_{s}$, we can find $k\in\mathds{N}$ such that $T\leq\tau_{k}$. Hence, it follows that the convergence result above holds $\mathds{P}$-almost surely for $t \in [s,T]$.

\par With the ancillary results above at our disposal, we establish the convergence result in the claim. To this end, note that since $-\infty<\mathds{E}\,[u^{i}(W^{i,\,\rho}_{T})]\leq0$ as shown earlier, we exploit the fact that the Bernoulli utility function $u^{i}$ is strictly concave in conjunction with Cauchy\textendash Schwarz inequality so as to obtain
\begin{align*}
0\leq &\,\limsup\limits_{n\,\rightarrow\,\infty}\,\mathds{E}\,\Big[\left\lvert\, u^{i}\!\left(W^{i,\,\rho}_{T}\right)\,-\,u^{i}\!\left(W^{i,\,\rho_n}_{T}\right)\,\right\rvert\Big]\!\leq \limsup\limits_{n\,\rightarrow\,\infty}\,\mathds{E}\left[\,\left\lvert\, D_{1} u^{i}\!\left(W^{i,\,\rho}_{T}\right)\right\rvert \left\lvert\,W^{i,\,\rho}_{T}-\,W^{i,\,\rho_n}_{T}\,\right\rvert\,\right]\\
& \leq\limsup\limits_{n\,\rightarrow\,\infty}\,\mathds{E}\left[\,\left\lvert\ D_{1} u^{i}\!\left(W^{i,\,\rho}_{T}\right)\,\right\rvert^{2}\right]\mathds{E}\left[\, \left\lvert\,W^{i,\,\rho}_{T}-\,W^{i,\,\rho_n}_{T}\,\right\rvert^{2}\right]
\end{align*}

\par In view of the assumption on the functional form of the utility function, and by Cauchy\textendash Schwarz inequality, it follows that we can find a positive constant $C$ such that we obtain
\begin{multline*}
    \mathds{E}\left[\left\lvert\, D_{1} u^{i}\!\left(W^{i,\,\rho}_{T}\right)\,\right\rvert^{2}\!\right]\!\leq C\mathds{E}\left[\exp{\left(2\delta^{i}\theta^{i}\!\!\int_{s}^{T}\!\!\!\pi^{i,\,\rho}_{t}\,x^{i,\,\rho}_{t}dt+2\delta^{i}\theta^{-i}\!\!\int_{s}^{T}\!\!\!\pi^{i,\,\rho}_{t}\,x^{-i}_{t}dt -2\delta^{i}\!\!\int_{s}^{T}\!\!\!\pi^{i,\,\rho}_{t}\,\sigma\!\left(S_{t}\right)dB_{t}\right)}\right]\\ \leq C\mathds{E}\left[\exp{\left(4\delta^{i}\theta^{i}\!\!\int_{s}^{T}\!\!\!\pi^{i,\,\rho}_{t}\,x^{i,\,\rho}_{t}dt \right)}\right]\mathds{E}\left[\exp{\left(4\delta^{i}\theta^{-i}\!\!\int_{s}^{T}\!\!\!\pi^{i,\,\rho}_{t}\,x^{-i}_{t}dt \right)}\right]\mathds{E}\left[\exp{\left(-4\delta^{i}\!\!\int_{s}^{T}\!\!\!\pi^{i,\,\rho}_{t}\,\sigma\!\left(S_{t}\right)dB_{t}\right)}\right]
\end{multline*}

\par In view of the convergence result for the sequence of processes $\left\{X^{\,i,\,\rho_{n}}\right\}_{n \,\in\,\mathds{N}}$ established earlier, we only need to show that the three terms on the right-hand side of the equation above are finite. To this end, we consider the first term on the right-hand side of the equation above, and note that by way of H\"{o}lder's inequality, and the definition of $\pi^{i,\,\rho}_{t}$ we obtain
\begin{align*}
\begin{split}
    \mathds{E}\left[\exp{\left(4\delta^{i}\theta^{i}\!\!\int_{s}^{T}\!\!\!\!\pi^{i,\,\rho}_{t}\,x^{i,\,\rho}_{t}dt \right)}\right] & \leq \mathds{E}\left[\exp{\left(4\delta^{i}\theta^{i}\left(\,\sup\limits_{t\,\in\,\left[s,\,T\right]}\big\lvert\,\pi^{i,\,\rho}_{t}\,\big\rvert\,\right)\int_{s}^{T}\!\!\!\!\big\lvert\, x^{i,\,\rho}_{t}\,\big\rvert\,dt \right)}\right]\\
    & \leq\mathds{E}\left[\exp{\left(4\delta^{i}\theta^{i}\!\!\int_{s}^{T}\!\!\!\!\big\lvert\, x^{i,\,\rho}_{t}\,\big\rvert\,dt\int_{s}^{T}\!\!\!\!\big\lvert\, x^{i,\,\rho}_{t}\,\big\rvert\, dt \right)}\right]
\end{split}
\end{align*}

\par Further, since $X^{\,i,\,\rho} \in \mathcal{A}_{s}$ by definition, the term on the right-hand side of the inequality above is finite. In a similar vein, we obtain the following as an immediate consequence of H\"{o}lder's inequality and the definition of $\pi^{i,\,\rho}_{t}$
\begin{align*}
\begin{split}
    \mathds{E}\left[\exp{\left(4\delta^{i}\theta^{i}\!\!\int_{s}^{T}\!\!\!\!\pi^{i,\,\rho}_{t}\,x^{-i}_{t}dt \right)}\right] &\leq \mathds{E}\left[\exp{\left(4\delta^{i}\theta^{i}\left(\,\sup\limits_{t\,\in\,\left[s,\,T\right]}\big\lvert\,\pi^{i,\,\rho}_{t}\,\big\rvert\,\right)\int_{s}^{T}\!\!\!\!\big\lvert\, x^{-i}_{t}\,\big\rvert\, dt \right)}\right]\\ 
    & \leq \mathds{E}\left[\exp{\left(4\delta^{i}\theta^{i}\!\!\int_{s}^{T}\!\!\!\!\big\lvert\, x^{i,\,\rho}_{t}\,\big\rvert\,dt \int_{s}^{T}\!\!\!\!\big\lvert\,x^{-i}_{t}\,\big\rvert\,dt \right)}\right]
\end{split}
\end{align*}

\par Note that the term on the right-hand side of the inequality above is finite in view of the fact that $X^{i,\,\rho}$, $X^{-i} \in \mathcal{A}_{s}$. Next, we appeal to \cite[Theorem 43, Page 140]{protter2005stochastic} in conjunction with the fact that the local volatility function $\sigma$ satisfies \textcolor{azul-pesc}{\Cref{siglip}}, along with the definition of $\pi^{i,\,\rho}_{t}$ to ascertain the existence of a constant $C > 0$ such that we have
\begin{align*}
\begin{split}
    \mathds{E}\left[\exp{\left(-4\delta^{i}\!\!\int_{s}^{T}\!\!\!\!\pi^{i,\,\rho}_{t}\,\sigma\!\left(S_{t}\right)dB_{t}\right)}\right]&\leq\, \mathds{E}\left[\exp{\left(C\!\int_{s}^{T}\!\!\!\left(\pi^{i,\,\rho}_{t}\right)^{2}\!\! dt\right)}\right]\\
    &\leq\,\mathds{E}\left[\exp{\left(C\!\int_{s}^{T}\!\!\!\!\big\lvert\,x^{i,\,\rho}_{t}\,\big\rvert\, dt\int_{s}^{T}\!\!\!\!\big\lvert\, x^{i,\,\rho}_{t}\,\big\rvert\,dt \right)}\right]
\end{split}
\end{align*}

\par Given that $X^{\,i,\,\rho} \in \mathcal{A}_{s}$, the term on the right-hand side of the inequality above is finite. Thus, as an immediate consequence of the above and the convergence result established earlier we obtain the following
\begin{equation*}
    \lim\limits_{n\,\rightarrow\,\infty}\,U^{i}\!\left(X^{i,\,\rho_n},X^{-i}\right)\, =\, \lim\limits_{n\,\rightarrow\,\infty}\,\mathds{E}\left[u^{i}\!\left(W^{i,\,\rho_n}_{T}\right)\right]\, =\, \mathds{E}\left[u^{i}\!\left(W^{i,\,\rho}_{T}\right)\right]\, =\, U^{i}\!\left(X^{\rho},X^{-i}\right)
\end{equation*}

\par Moreover, upon combining the equation above with (\ref{jrho}), it then follows that we have
\begin{equation*}
J^{i}\!\left(s, y^{i}; X^{-i}\right) - \rho\, \leq\, U^{i}\!\left(X^{i,\,\rho},X^{-i}\right)\, = \, \lim\limits_{n\,\rightarrow\,\infty}\,U^{i}\!\left(X^{i,\,\rho_n},X^{-i}\right)
\end{equation*}
 
\par Further, recall that for $n,m \in \mathds{N}$, such that $n\leq m$, we have $\mathcal{A}^{n}_{s}\subseteq\mathcal{A}^{m}_{s}$, which in turn implies that $J^{i,\,n}\!\left(s, y^{i}; X^{-i}\right)\leq J^{i,\,m}\!\left(s, y^{i}; X^{-i}\right)$, and therefore we have
\begin{equation*}
U^{i}\!\left(X^{i,\,\rho_n}, X^{-i}\right)\, \leq\, \sup\limits_{m\,\geq\, n}\,J^{i,\,m}\!\left(s, y^{i}; X^{-i}\right)
\end{equation*}
 
\par Next, given the equation above, we consider the limit as $n\rightarrow\infty$, and in view of (\ref{jsup}) arrive at the following
\begin{equation*}
 J^{i}\!\left(s, y^{i}; X^{-i}\right) - \rho\, \leq\, \limsup\limits_{n\,\rightarrow\,\infty}\,U^{i}\!\left(X^{i,\,\rho_n},X^{-i}\right)\, \leq\, \limsup\limits_{n\,\rightarrow\,\infty}\,J^{i,\,n}\!\left(s,y^{i};X^{-i}\right)\, \leq\, J^{i}\!\left(s, y^{i}; X^{-i}\right)
\end{equation*}
 
\par Given that the choice of $\rho$ above was arbitrary the claim then follows from above.

\par \vspace{0.5em}

\par $\mathbf{(iii)}$ For a fixed $s\in [0,\,T)$, we show the continuity of $J^{i,n}\!\left(s,\,\cdot\,;\,X^{-i}\right)$. To this end, given $y^{i}$, we consider a sequence $\{y^{i,\,n}\}_{n\, \in\, \mathds{N}}$, where $y^{i},y^{i,\,n} \in \mathds{R}^{\,\lvert\, Y^{i}\,\rvert}$ for each $n \in \mathds{N}$, such that we have $\lim\limits_{n\,\rightarrow\,\infty}y^{i,\,n} = y^{i}$. It then follows in view of the definition of $J^{i,\,n}(\cdot)$ that we have
\begin{equation*}
\limsup\limits_{n\,\rightarrow\,\infty}\left\lvert\,J^{i,\,n}\!\left(s, y^{i}; X^{-i}\right) - J^{i,\,n}\!\left(s, y^{i,\,n}; X^{-i}\right)\,\right\rvert\ \leq\ \limsup\limits_{n\,\rightarrow\,\infty}\sup\limits_{X^{i}\,\in\,\mathcal{A}^{n}_{s}}\,\mathds{E}\,\Big[\left\lvert\, u^{i}\!\left(W^{\,i,\,y^{i}}_{T}\right)\,-\,u^{i}\!\left(W^{\,i,\,y^{i,\,n}}_{T}\right)\,\right\rvert\Big] 
\end{equation*}

\par In view of the above, it then follows that for $n$ large enough we have
\begin{equation*}
\limsup\limits_{n\,\rightarrow\,\infty}\left\lvert\,J^{i,\,n}\!\left(s, y^{i}; X^{-i}\right) - J^{i,\,n}\!\left(s, y^{i,\,n}; X^{-i}\right)\,\right\rvert\, \leq\, \sup\limits_{X^{i}\,\in\,\mathcal{A}^{n}_{s}}\limsup\limits_{n\,\rightarrow\,\infty}\,\mathds{E}\,\Big[\left\lvert\, u^{i}\!\left(W^{\,i,\,y^{i}}_{T}\right)\,-\,u^{i}\!\left(W^{\,i,\,y^{i,\,n}}_{T}\right)\,\right\rvert\Big]
\end{equation*}

\par Next, we use the fact that the Bernoulli utility function $u^{i}$ is strictly concave in conjunction with Cauchy\textendash Schwarz inequality in order to obtain
\begin{multline*}
\limsup\limits_{n\,\rightarrow\,\infty}\,\mathds{E}\left[\,\left\lvert u^{i}\!\left(W^{\,i,\,y^{i}}_{T}\right)\,-\,u^{i}\!\left(W^{\,i,\,y^{i,\,n}}_{T}\right)\right\rvert\,\right]\ \leq\ \limsup\limits_{n\,\rightarrow\,\infty}\,\mathds{E}\left[\,\left\lvert D_{1} u^{i}\!\left(W^{\,i,\,y^{i}}_{T}\right) \right\rvert\left\lvert W^{\,i,\,y^{i}}_{T}-\,W^{\,i,\,y^{i,\,n}}_{T}\right\rvert\,\right]\\ \leq \,\limsup\limits_{n\,\rightarrow\,\infty}\,\mathds{E}\left[\,\left\lvert D_{1} u^{i}\!\left(W^{\,i,\,y^{i}}_{T}\right)\right\rvert^{2}\right]\mathds{E}\left[\,\left\lvert W^{\,i,\,y^{i}}_{T}-\,W^{\,i,\,y^{i,\,n}}_{T}\right\rvert^{2}\right]
\end{multline*}

\par Note that given the equation above, we can show as in $\mathbf{(ii)}$ above that given $X^{-i}\,\in\,\mathcal{A}_{s}$, it follows that for every $X^{i}\,\in\,\mathcal{A}_{s}$ we have
\begin{equation*}
    0\, \leq\, \mathds{E}\left[\,\left\lvert D_{1} u^{i}\!\left(W^{\,i,\,y^{i}}_{T}\right)\right\rvert^{2}\right]\ <\ \infty
\end{equation*}

\par Moreover, it follows as an immediate consequence of \textcolor{azul-pesc}{\Cref{momin}(ii)} that we have
\begin{equation*}
    \mathds{E}\left[\,\left\lvert W^{\,i,\,y^{i}}_{T}-\,W^{i,\,y^{\,i,\,n}}_{T}\right\rvert^{2}\right] \,\leq\, \left\vert\, y^{i,\,n}-y^{i}\,\right\rvert^{2}
\end{equation*}

\par The continuity of $J^{i,\,n}\!\left(s,\, \cdot\,;\, X^{-i}\right)$ is then immediate from the above.

\par \vspace{0.5em}

\par $\mathbf{(iv)}$ Next, we prove the lower semi-continuity of $J^{i}\!\left(s,\, \cdot\,;\, X^{-i}\right)$. To this end, let $s\in [0,T)$ be given and note from $\mathbf{(ii)}$ above, that for every $\epsilon > 0$, we can find $m \in \mathds{N}$, such that for $n \geq m$ we have \begin{equation*}
J^{i}\!\left(s, y^{i}; X^{-i}\right)\,-\,J^{i,\,n}\!\left(s, y^{i}; X^{-i}\right)\,\leq\, \frac{\epsilon}{2}
\end{equation*}

\par Moreover, in view of $\mathbf{(iii)}$ above, we can find $\delta > 0$, such that for $\tilde{y}^{i}\,\in\, \mathcal{B}_{\delta}(y^{i})$ we have
\begin{equation*}
 J^{i,\,n}\!\left(s, y^{i}; X^{-i}\right)\,-\,\frac{\epsilon}{2}\,\leq\, J^{i,\,n}\!\left(s, \tilde{y}^{i}; X^{-i}\right)
\end{equation*}
 
\par Next, we combine the above with the fact that the function $J^{i,\,n}$ is bounded above by $J^{i}$ in order to obtain,
\begin{equation*}
J^{i}\!\left(s, \tilde{y}^{i}; X^{-i}\right)\,\geq\, J^{i,\,n}\!\left(s, \tilde{y}^{i}; X^{-i}\right)\,\geq\, J^{i,\,n}\!\left(s, y^{i}; X^{-i}\right)-\frac{\epsilon}{2}\,\geq\, J^{i}\!\left(s, y^{i}; X^{-i}\right)-\epsilon
\end{equation*}

\par The claim is then immediate from above.

\par \vspace{0.5em}

\par $\mathbf{(v)}$ First, we show that given $y^{i}\in \mathds{R}^{\lvert\,Y^{i} \rvert}$, $J^{i,\,n}\!\left(\,\cdot\, ,\,y^{i}\,;\,X^{-i}\right)$ is continuous at $T$. To this end, we consider an increasing sequence of deterministic times $\{s^{m}\}_{m\, \in\, \mathds{N}} \subseteq [s,T)$, such that we have $s^{m} \uparrow T$, as $m\,\rightarrow\,\infty$. It then follows in view of the definition of $J^{i}(\cdot)$ that we have
\begin{equation*}
\limsup\limits_{s^{m}\,\uparrow\ T}\,\left\lvert\,J^{i,\,n}\!\left(T, y^{i}; X^{-i}\right) - J^{i,\,n}\!\left(s^{n}, y^{i}; X^{-i}\right)\,\right\rvert\ \leq\ \limsup\limits_{s^{m}\,\uparrow\ T}\sup\limits_{X^{i}\,\in\,\mathcal{A}^{n}_{s}}\mathds{E}\left[\left\lvert\, u^{i}\!\left(W^{\,i,\,T}_{T}\right)\,-\,u^{i}\!\left(W^{\,i,\,s^{n}}_{T}\right)\right\rvert\right] 
\end{equation*}

\par In view of the above, it then follows that for $m$ large enough we have
\begin{equation*}
\limsup\limits_{s^{m}\,\uparrow\ T}\,\left\lvert\,J^{i,\,n}\!\left(T, y^{i}; X^{-i}\right) - J^{i,\,n}\!\left(s^{m}, y^{i}; X^{-i}\right)\,\right\rvert\, \leq\, \sup\limits_{X^{i}\,\in\,\mathcal{A}^{n}_{s}}\limsup\limits_{s^{m}\,\uparrow\ T}\,\mathds{E}\left[\left\lvert\, u^{i}\!\left(W^{\,i,\,T}_{T}\right)\,-\,u^{i}\!\left(W^{\,i,\,s^{m}}_{T}\right)\right\rvert\right]
\end{equation*}

\par Next, we use the fact that the Bernoulli utility function $u^{i}$ is strictly concave in conjunction with Cauchy\textendash Schwarz inequality in order to obtain
\begin{multline*}
\limsup\limits_{s^{m}\,\uparrow\ T}\,\mathds{E}\left[\left\lvert\, u^{i}\!\left(W^{\,i,\,T}_{T}\right)\,-\,u^{i}\!\left(W^{\,i,\,s^{m}}_{T}\right)\right\rvert\right]\ \leq\ \limsup\limits_{s^{m}\,\uparrow\ T}\,\mathds{E}\left[\left\lvert\ D_{1} u^{i}\!\left(W^{\,i,\,T}_{T}\right)\,\right\rvert\left\lvert\,W^{\,i,\,T}_{T}-\,W^{\,i,\,s^{m}}_{T}\,\right\rvert\right]\\ \leq \,\limsup\limits_{s^{m}\,\uparrow\ T}\,\mathds{E}\left[\,\left\lvert\, D_{1} u^{i}\!\left(W^{\,i,\,T}_{T}\right)\right\rvert^{2}\right]\mathds{E}\left[\,\left\lvert\,W^{\,i,\,T}_{T}-\,W^{\,i,\,s^{m}}_{T}\right\rvert^{2}\right]
\end{multline*}

\par Given the equation above, we exploit the fact that $W^{\,i,\,T}_{T}$ is deterministic, along with the fact that $u^{i} \in \mathds{C}^{2}$ to ascertain that
\begin{equation*}
    0\, \leq\, \mathds{E}\left[\,\left\lvert\, D_{1} u^{i}\!\left(W^{\,i,\,s}_{T}\right)\right\rvert^{2}\right]\ =\ \left\lvert\, D_{1} u^{i}\!\left(W^{\,i,\,s}_{T}\right)\right\rvert^{2} \ <\ \infty
\end{equation*}

\par Moreover, as an immediate consequence of \textcolor{azul-pesc}{\Cref{momin}(iii)} it follows that we have
\begin{equation*}
    \limsup\limits_{s^{m}\,\uparrow\ T}\,\mathds{E}\left[\,\left\lvert W^{\,i,\,T}_{T}-\,W^{\,i,\,s^{m}}_{T}\right\rvert^{2}\right] \,=\,0
\end{equation*}

\par The continuity of $J^{i,\,n}\!\left(\,\cdot,\, y^{i}\,;\, X^{-i}\right)$ at $T$ is immediate from the above. The lower semi-continuity of $J^{i}\!\left(\,\cdot,\, y^{i}\,;\, X^{-i}\right)$ at $T$ then follows from an argument identical to $\mathbf{(iv)}$ above.

\end{proof}

\subsection{Proof of \textcolor{azul-pesc}{\Cref{pcapp}}}

\begin{proof}
\par Note that we can assume without any loss of generality that given $\Pi_{1},\Pi_{2} \in\mathcal{A}^{a}_{s}$, we have $\mathds{P}$-almost surely $\left\lvert\pi_{1,\,t} - \pi_{2,\,t}\right\rvert<1$, for all $t\in\left[s,T\right]$. This statement follows in view of the fact that given the Euclidean metric $\lvert\,\cdot\,\rvert$ on $\mathds{R}$, we can define an equivalent metric $\rho$, such that $\rho\!\left(r\right)<1$, for all $r \in \mathds{R}$, as follows 
\begin{equation*}
    \rho\!\left(\pi_{1,\,t},\pi_{2,\,t}\right)\,=\,\frac{2}{\pi}\arctan\Big(\left\lvert\pi_{1,\,t}-\pi_{2,\,t}\right\rvert\Big)
\end{equation*}

\par To verify that $\rho$ is a metric, it suffices to establish the triangle inequality, since the identity of indiscernibles and the symmetry property for $\rho$ follow immediately from the definition. To this end, note that from the triangle inequality for the Euclidean metric and the strict monotonicity of the $\arctan$ function it follows that given $\Pi_{3} \in \mathcal{A}^{a}_{s}$  we have
\begin{equation*}
    \rho\!\left(\pi_{1,\,t},\,\pi_{2,\,t}\right)\,=\,\frac{2}{\pi}\arctan\Big(\left\lvert\pi_{1,\,t}-\pi_{3,\,t}+\pi_{3,\,t}-\pi_{2,\,t}\right\rvert\Big)\,\leq\,\frac{2}{\pi}\arctan\Big(\left\lvert\pi_{1,\,t}-\pi_{3,\,t}\right\rvert+\left\lvert\pi_{3,\,t}-\pi_{2,\,t}\right\rvert\Big)
\end{equation*}

\par Also, recall that the $\arctan$ function is strictly concave on the positive real line with $\arctan(0) = 0$, and hence it follows that the $\arctan$ function is subadditive on the positive real line, which then implies that the equation above leads us to
\begin{align*}
    \rho\!\left(\pi_{1,\,t},\,\pi_{2,\,t}\right)\,&\leq\,\frac{2}{\pi}\arctan\Big(\left\lvert\,\pi_{1,\,t}-\pi_{3,\,t}\right\rvert+\left\lvert\pi_{3,\,t}-\pi_{2,\,t}\right\rvert\Big)\\ &\leq\,\frac{2}{\pi}\arctan\Big(\left\lvert\,\pi_{1,\,t}-\pi_{3,\,t}\right\rvert\Big)+\frac{2}{\pi}\arctan\Big(\left\lvert\,\pi_{3,\,t}-\pi_{2,\,t}\right\rvert\Big)
\end{align*}

\par The triangle inequality for $\rho$ is immediate from the above. The proof of the claim in the statement of the lemma then follows \textit{mutatis mutandis} from the proof of \cite[Lemma 3.6]{krylov2008controlled}.
\end{proof}

\subsection{Proof of \textcolor{azul-pesc}{\Cref{pcsapp}}}

\begin{proof}
\par Given $X^{-i}\in\mathcal{A}_{s}$, and in view of the definition of the controlled state processes $Z^{i,\,\Pi_{n}}$ and $Z^{i,\,\Pi}$ it follows that for $t \in [0,T]$, we have
\begin{equation*}
    \left\lvert Z^{i,\,\Pi_{n}}_{t} - Z^{i,\,\Pi}_{t}\right\rvert \leq  \mathlarger{\int}_{s}^{t}\!\left\lvert\,\beta^{i}\!\left(\pi_{n,\,u},Z^{i,\,\Pi_{n}}_{u}\right) - \beta^{i}\!\left(\pi_{u},Z^{i,\,\Pi}_{u}\right)\right\rvert du\,
    +\mathlarger{\int}_{s}^{t}\!\left\lvert\,\nu^{i}\!\left(\pi_{n,\,u},Z^{i,\,\Pi_{n}}_{u}\right) - \nu^{i}\!\left(\pi_{u},Z^{i,\,\Pi}_{u}\right)\right\rvert dB_{u}
\end{equation*}

\par Further, recall that by definition we have
\begin{align*}
    &\beta^{i}\!\left(\pi^{i}_{t},Z^{i}_{t}\right) = \,x^{-i}_{t}\left[-\theta^{-i},-1,-\theta^{-i}\pi^{i}_{t},\theta^{i}\pi^{i}_{t}-\theta^{-i}\pi^{-i}_{t}\right]^{\mathrm{T}} \\
    &\nu^{i}\!\left(\pi^{i}_{t},Z^{i}_{t}\right) = \,\sigma\!\left(P^{i}_{t}+\theta^{i}\pi^{i}_{t}\right)\left[1,\,0,\,\pi^{i}_{t},\,\pi^{-i}_{t}\right]^{\mathrm{T}}
\end{align*}

\par In view of the above and given that the local volatility function $\sigma$ satisfies {\textcolor{azul-pesc}{\Cref{siglip}}}, there exists a positive constant $C$ such that we have
\begin{equation*}
    \left\lvert Z^{i,\,\Pi_{n}}_{t} - Z^{i,\,\Pi}_{t}\right\rvert\,\leq\, C\!\mathlarger{\int}_{s}^{t}\!\!\left\lvert\,x^{-i}_{u}\,\right\rvert\left\lvert\,\pi_{n,\,u} - \pi_{u}\,\right\rvert du\, +\, C\,\left\lvert\,\mathlarger{\int}_{s}^{t}\!\!\left(\pi_{n,\,u} - \pi_{u}\right) dB_{u}\,\right\rvert
    +\,C\,\left\lvert\,\mathlarger{\int}_{s}^{t}\!\!\left(Z^{i,\,\Pi_{n}}_{u} - Z^{i,\,\Pi}_{u}\right) dB_{u}\,\right\rvert
\end{equation*}

\par Next, given the equation above we can find a positive constant $C$ such that the following holds
\begin{align*}
    \mathds{E}\left[\sup\limits_{t\,\in\,[s,\,T]}\left\lvert Z^{i,\,\Pi_{n}}_{t} - Z^{i,\,\Pi}_{t}\right\rvert^{2}\right]\leq &\, C\mathds{E}\left[\mathlarger{\int}_{s}^{T}\!\!\!\!\left\lvert\,x^{-i}_{u}\,\right\rvert\left\lvert\,\pi_{n,\,u}-\pi_{u}\,\right\rvert du\right]^{2}\!+C\mathds{E}\left[\sup\limits_{t\,\in\,[s,\,T]}\left\lvert\,\mathlarger{\int}_{s}^{t}\!\!\!\!\left(\pi_{n,\,u} - \pi_{u}\right) dB_{u}\,\right\rvert\,\right]^{2}\\
    & +C\mathds{E}\left[\sup\limits_{t\,\in\,[s,\,T]} \left\lvert\,\mathlarger{\int}_{s}^{t}\!\!\left(Z^{i,\,\Pi_{n}}_{u} - Z^{i,\,\Pi}_{u}\right)dB_{u}\,\right\rvert\,\right]^{2}
\end{align*}

\par Note that since $\Pi \in \mathcal{A}^{a}_{s}$, it follows by definition that $\left\lvert\pi^{-i}_{t}\right\rvert$ is bounded above. Also, for an element of the sequence $\left\{\Pi_{n}\,\in\,\mathcal{A}^{a,\,pc}_{s}\!\left(I_{n}\right)\right\}_{n\,\in\,\mathds{N}}$, we have by definition
\begin{equation*}
   \left\lvert\pi_{n,\,t}\right\rvert \leq\max\Big\{\left\lvert\,\pi_{t_{0}}\,\right\rvert,\,\left\lvert\,\pi_{t_{1}}\,\right\rvert,...\,,\,\left\lvert\,\pi_{t_{n}}\,\right\rvert\Big\}
\end{equation*} 

\par Given the above, it then follows by way of Burkholder\textendash Davis\textendash Gundy inequality that there exists a positive constant $C$, such that we have
\begin{align*}
    \mathds{E}\left[\sup\limits_{t\,\in\,[s,\,T]}\,\left\lvert Z^{i,\,\Pi_{n}}_{t} - Z^{i,\,\Pi}_{t}\right\rvert^{2}\right]\leq\, &\, C\mathds{E}\left[\mathlarger{\int}_{s}^{T}\!\!\left\lvert\,x^{-i}_{u}\,\right\rvert\left\lvert\,\pi_{n,\,u}-\pi_{u}\,\right\rvert du\right]^{2}\!+ C\mathds{E}\left[\mathlarger{\int}_{s}^{T}\!\!\left\lvert\,\pi_{n,\,u}-\pi_{u}\,\right\rvert du\right]\\
    & +C\mathds{E}\left[\mathlarger{\int}_{s}^{T}\!\!\left\lvert Z^{i,\,\Pi_{n}}_{u} - Z^{i,\,\Pi}_{u}\right\rvert^{2} du\right]
\end{align*}

\par Additionally, in view of Tonelli's theorem, it is immediate from the above that we have
\begin{align}\label{auxsc}
    \begin{split}
    \mathds{E}\left[\sup\limits_{t\,\in\,[s,\,T]}\,\left\lvert Z^{i,\,\Pi_{n}}_{t} - Z^{i,\,\Pi}_{t} \right\rvert^{2}\right]\leq\, &\, C\mathds{E}\left[\mathlarger{\int}_{s}^{T}\!\!\left\lvert\,x^{-i}_{u}\,\right\rvert\left\lvert\,\pi_{n,\,u}-\pi_{u}\,\right\rvert du\right]^{2}\!+ C\mathds{E}\left[\mathlarger{\int}_{s}^{T}\!\!\left\lvert\,\pi_{n,\,u}-\pi_{u}\,\right\rvert du\right]\\
    & +\,C\!\mathlarger{\int}_{s}^{T}\!\!\mathds{E}\left[\sup\limits_{u\,\in\,[s,\,T]}\,\left\lvert Z^{i,\,\Pi_{n}}_{u} - Z^{i,\,\Pi}_{u}\right\rvert^{2}\right]du
    \end{split}
\end{align}

\par Further, as established earlier, both $\left\lvert\pi^{-i}_{t}\right\rvert$, $\left\lvert\pi_{n,\,t}\right\rvert$ are bounded above. Moreover, since $X^{-i}\in \mathcal{A}_{s}$, it follows from Dominated Convergence Theorem that we have
\begin{equation*}
   \lim\limits_{n\,\rightarrow\, \infty}\mathds{E}\left[\mathlarger{\int}_{s}^{T}\!\!\left\lvert\,x^{-i}_{u}\,\right\rvert\,\left\lvert\,\pi_{n,\,u}-\pi_{u}\,\right\rvert du\right]\, =\, \mathds{E}\left[\mathlarger{\int}_{s}^{T}\!\!\lim\limits_{n\,\rightarrow\, \infty}\,\left\lvert\,x^{-i}_{u}\,\right\rvert\left\lvert\,\pi_{n,\,u}-\pi_{u}\,\right\rvert du\right]
\end{equation*}

\par Given that $X^{-i}\in \mathcal{A}_{s}$, it follows that $\left\lvert x^{-i}_{t}\right\rvert < \infty$, $\mathds{P} \otimes \lambda_{[s,\,T]}$ almost everywhere, where $\lambda_{[s,\,T]}$ denotes the restriction of the Lebesgue measure defined over the real line $\mathds{R}$, on the interval $[s,T]$. Also, since $d(\Pi_{n}, \Pi)\rightarrow 0$, as $n\rightarrow \infty$, we can find a subsequence $\{n_{m}\}_{m \in \mathds{N}}$ such that $\left\lvert\pi_{n_{m},\,t}-\pi_{t}\right\rvert \rightarrow 0$, $\mathds{P}\otimes \lambda_{[s,\,T]}$ almost everywhere, as $m \rightarrow \infty$. This, in conjunction with the equation above, then implies that
\begin{equation*}
   \lim\limits_{m\,\rightarrow\, \infty}\mathds{E}\left[\mathlarger{\int}_{s}^{T}\!\!\left\lvert\,x^{-i}_{u}\,\right\rvert\left\lvert\,\pi_{n_{m},\,u}-\pi_{u}\,\right\rvert du\right]\, =\, \mathds{E}\left[\mathlarger{\int}_{s}^{T}\!\!\lim\limits_{m\,\rightarrow\, \infty}\,\left\lvert\,x^{-i}_{u}\,\right\rvert\left\lvert\,\pi_{n_{m},\,u}-\pi_{u}\,\right\rvert du\right]\ =\ 0
\end{equation*}

\par Moreover, note that by virtue of the fact that $X^{-i} \in \mathcal{A}_{s}$, we have
\[  \mathds{E}\left[\int_{s}^{T}\!\!\left\lvert\,x^{-i}_{u}\,\right\rvert du\,\right]^{2} < \infty\]

\par We employ the above in conjunction with the fact that both $\left\lvert\pi^{-i}_{t}\right\rvert$, $\left\lvert\pi_{n,\,t}\right\rvert$ are bounded above to invoke Dominated Convergence Theorem, which then leads us to
\begin{equation*}
   \lim\limits_{m\,\rightarrow\, \infty}\mathds{E}\left[\mathlarger{\int}_{s}^{T}\!\!\left\lvert\,x^{-i}_{u}\,\right\rvert\left\lvert\,\pi_{n_{m},\,u}-\pi_{u}\,\right\rvert du\right]^{2}\! =\, \mathds{E}\left[\,\lim\limits_{m\,\rightarrow\, \infty}\,\left(\mathlarger{\int}_{s}^{T}\!\!\left\lvert\,x^{-i}_{u}\,\right\rvert\left\lvert\,\pi_{n_{m},\,u}-\pi_{u}\,\right\rvert du\right)^{2}\,\right] =\, 0
\end{equation*}

\par In view of the equation above, it is immediate that the first term on the right-hand side of (\ref{auxsc}) converges to zero in the limit as $m \rightarrow \infty$. Also, note that by hypothesis the second term on the right-hand side of (\ref{auxsc}) again converges to zero in the limit as $m \rightarrow \infty$. Thus, given (\ref{auxsc}), we invoke Gronwall's lemma and consider the limit as $m \rightarrow \infty$, whereupon the claim follows immediately from above.
\end{proof}

\subsection{Proof of \textcolor{azul-pesc}{\Cref{pcval}}}

\begin{proof}
\par Given a fixed initial time $s \in [0,T]$, consider the deterministic strategy $\Pi^{0} = \{\pi^{0}_{t}\}_{t\in[s,T]}$, where $\pi^{0}_{t} = 0$, for $t \in [s,T]$. Note that $\Pi^{0} \in \mathcal{A}^{a,\,pc}_{s}\!(I_{n})$, for all $n \in \mathds{N}$. Given $X^{-i} \in \mathcal{A}_{s}$, we let $Z^{i,\,0}_{T}$ denote the terminal value of the auxiliary state process corresponding to strategy $\Pi^{0}$ for the $i$th investor. It then follows from the definition of $H^{i}$ that we have
\begin{equation*}
     H^{i}\!\left(\Pi^{0};X^{-i}\right)\, =\, \mathds{E}\left[u^{i}\!\left(\mathrm{w}^{i,\,0}_{T}\right)\right]\, =\, u^{i}\!\left(\mathrm{w}^{i,\,0}_{0}\right)\, >\, -\infty 
\end{equation*}

\par In view of the equation above, it then follows that for each $n \in \mathds{N}$, we have
\begin{equation*}
 \sup\limits_{\Pi_{n}\, \in\, \mathcal{A}^{a,\,pc}_{s}\!\left(I_{n}\right)} H^{i}\!\left(\Pi_{n};X^{-i}\right)\, >\, -\infty
\end{equation*}

\par Given that the sequence of partitions is nested, it follows that for $n,m \in \mathds{N}$ such that $n \leq m$ we have $\mathcal{A}^{a,\,pc}_{s}\!(I_{n}) \subseteq \mathcal{A}^{a,\,pc}_{s}\!(I_{m})$, which in turn implies that we have 
\begin{equation*}
 \sup\limits_{\Pi_{m}\, \in\, \mathcal{A}^{a,\,pc}_{s}\!\left(I_{m}\right)} H^{i}\!\left(\Pi_{m};X^{-i}\right)\ \geq\ \sup\limits_{\Pi_{n}\, \in\, \mathcal{A}^{a,\,pc}_{s}\!\left(I_{n}\right)} H^{i}\!\left(\Pi_{n};X^{-i}\right)
\end{equation*}

\par Further, note that $\mathcal{A}^{a,\,pc}_{s}\!(I_{n}) \subseteq \mathcal{A}^{a}_{s}$, which in conjunction with the definition of $V^{i}$ then implies that we have
\begin{equation}\label{vgpc}
V^{i}\!\left(s,z^{i};X^{-i}\right)\, \geq\, \limsup\limits_{n\,\rightarrow\,\infty} \left\{\,\sup\limits_{\Pi_{n}\, \in\, \mathcal{A}^{a,\,pc}_{s}\!\left(I_{n}\right)} H^{i}\!\left(\Pi_{n};X^{-i}\right)\,\right\}\, =\, \lim\limits_{n\,\rightarrow\,\infty} \left\{\,\sup\limits_{\Pi_{n}\, \in\, \mathcal{A}^{a,\,pc}_{s}\!\left(I_{n}\right)} H^{i}\!\left(\Pi_{n};X^{-i}\right)\,\right\}
\end{equation}

\par It then remains to show that the inequality above holds with equality. To this end, for a given $\rho > 0$, we consider a $\rho$-optimal control for the auxiliary control problem. That is, we consider $\Pi^{i,\,\rho} \in \mathcal{A}^{a}_{s}$ such that $\Pi^{i,\,\rho}$ satisfies the following inequality 
\begin{equation}\label{vrho}
V^{i}\!\left(s,z^{i};X^{-i}\right)-\rho\, \leq\, H^{i}\!\left(\Pi^{i,\,\rho},\,X^{-i}\right)
\end{equation}

\par The existence of $\Pi^{\,i,\,\rho} \in \mathcal{A}^{a}_{s}$ follows from the definition of $V^{i}$ and the fact that $V^{i}\!\left(s,z^{i};X^{-i}\right)\geq H^{i}\!\left(\Pi^{0};X^{-i}\right) > -\infty$. Moreover, it follows in view of \textcolor{azul-pesc}{\Cref{pcapp}} that we can find a sequence of controls $\left\{\Pi^{\,i,\,\rho_n} \in \mathcal{A}^{a,\,pc}_{s}\!\left(I_{n}\right)\right\}_{n\,\in\,\mathds{N}}$ such that $\mathrm{d}\!\left(\Pi^{\,i,\,\rho_n},\Pi^{\,i,\,\rho}\right) \rightarrow 0$, as $n \rightarrow \infty$.

\par Further, it follows from \textcolor{azul-pesc}{\Cref{pcsapp}} that we have $Z^{\,i,\,\rho}_{T} \rightarrow Z^{\,i,\,\rho_{n}}_{T}\, \mathds{P}$-almost surely as $n\rightarrow \infty$ (at least along a subsequence), where $Z^{\,i,\,\rho}$ and $Z^{\,i,\,\rho_{n}}$ denote the controlled state process corresponding to auxiliary controls $\Pi^{\,i,\,\rho}$ and $\Pi^{\,i,\,\rho_{n}}$ respectively. Thus, in view of the preceding arguments, in order to prove the claim in the statement of the lemma we only need to show that
\begin{equation*}
\lim\limits_{n\,\rightarrow\,\infty} H^{i}\!\left(\Pi^{\,i,\,\rho_{n}},\,X^{-i}\right)\,=\,H^{i}\!\left(\Pi^{\,i,\,\rho},\,X^{-i}\right)
\end{equation*}

\par To this end, recall that by definition we have we have $H^{i}(\Pi^{i};X^{-i}) = \mathds{E}\left[u^{i}(\mathrm{w}^{i}_{T})\right]$, given $\Pi^{i} \in \mathcal{A}^{a}_{s}$, $X^{-i} \in \mathcal{A}_{s}$. Further, since $-\infty<\mathds{E}[u^{i}(\mathrm{w}^{i,\,\rho}_{T})]\leq0$ as shown earlier, we exploit the fact that the Bernoulli utility function $u^{i}$ is strictly concave, in conjunction with Cauchy\textendash Schwarz inequality so as to obtain
\begin{align*}
0\,\leq\, &\,\limsup\limits_{n\,\rightarrow\,\infty}\,\mathds{E}\,\Big[\!\left\lvert\, u^{i}\!\left(\mathrm{w}^{\,i,\,\rho}_{T}\right)-u^{i}\!\left(\mathrm{w}^{\,i,\,\rho_n}_{T}\right)\right\rvert\!\Big]\,\leq\, \limsup\limits_{n\,\rightarrow\,\infty}\,\mathds{E}\left[\left\lvert\, D_{1} u^{i}\!\left(\mathrm{w}^{\,i,\,\rho}_{T}\right)\right\rvert \left\lvert\,\mathrm{w}^{\,i,\,\rho}_{T}-\mathrm{w}^{\,i,\,\rho_n}_{T}\,\right\rvert\right]\\
& \leq\,\limsup\limits_{n\,\rightarrow\,\infty}\,\mathds{E}\left[\,\left\lvert\ D_{1} u^{i}\!\left(\mathrm{w}^{\,i,\,\rho}_{T}\right)\,\right\rvert^{2}\right]\mathds{E}\left[\, \left\lvert\,\mathrm{w}^{\,i,\,\rho}_{T}-\mathrm{w}^{\,i,\,\rho_n}_{T}\,\right\rvert^{2}\right]
\end{align*}

\par In view of the assumption on the functional form of the Bernoulli utility function $u^{i}$, and by Cauchy\textendash Schwarz inequality, it follows that there exists a positive constant $C$ such that we have
\begin{multline*}
    \mathds{E}\left[\,\left\lvert D_{1} u^{i}\!\left(\mathrm{w}^{\,i,\,\rho}_{T}\right)\right\rvert^{2}\right]\, =\, \mathds{E}\left[\exp{\left(-2\delta^{i}\mathrm{w}^{i}_{s}+2\delta^{i}\theta^{-i}\!\!\int_{s}^{T}\!\!\pi^{i,\,\rho}_{t}x^{-i}_{t}dt -2\delta^{i}\!\!\int_{s}^{T}\!\!\pi^{i,\,\rho}_{t}\sigma\!\left(P^{i}_{t}+\theta^{i}\pi^{i,\,\rho}_{t}\right)dB_{t}\right)}\right]\\ \leq\,C\mathds{E}\left[\exp{\left(4\delta^{i}\theta^{-i}\!\!\int_{s}^{T}\!\!\pi^{i,\,\rho}_{t}x^{-i}_{t}dt \right)}\!\right]\mathds{E}\left[\exp{\left(-4\delta^{i}\!\!\int_{s}^{T}\!\!\pi^{i,\,\rho}_{t}\sigma\!\left(P^{i}_{t}+\theta^{i}\pi^{i,\,\rho}_{t}\right)dB_{t}\right)}\right]
\end{multline*}

\par Next, we consider the first term on the right-hand side of the inequality above, and note that by virtue of the fact that $\Pi^{\,i,\,\rho} \in \mathcal{A}^{a}_{s}$ and $X^{-i} \in \mathcal{A}_{s}$, we can find a positive constant $C$ (dependent on $T$) such that we have
\begin{equation*}
\mathds{E}\left[\exp{\left(4\delta^{i}\theta^{-i}\!\!\int_{s}^{T}\!\!\pi^{i,\,\rho}_{t}x^{-i}_{t}dt \right)}\right]\, \leq \, C\,<\,\infty
\end{equation*}

\par Moreover, we appeal to \cite[Theorem 43, Page 140]{protter2005stochastic} in conjunction with the fact that the local volatility function $\sigma$ satisfies \textcolor{azul-pesc}{\Cref{siglip}}, along with the fact that $\Pi^{\,i,\,\rho} \in \mathcal{A}^{a}_{s}$ to ascertain the existence of a constant $C > 0$ (dependent on $T$) such that we have
\begin{equation*}
\mathds{E}\left[\exp{\left(-4\delta^{i}\!\!\int_{s}^{T}\!\!\pi^{i,\,\rho}_{t}\,\sigma\!\left(P^{i}_{t}+\theta^{i}\pi^{i,\,\rho}_{t}\right)dB_{t}\right)}\right]\,\leq\,C\,<\,\infty
\end{equation*}

\par As an immediate consequence of the above and the convergence result established in \textcolor{azul-pesc}{\Cref{pcsapp}}, we thus obtain the following
\begin{equation*}
    \lim\limits_{n\,\rightarrow\,\infty}\,H^{i}\!\left(\Pi^{\,i,\,\rho_n},X^{-i}\right)\, =\, \lim\limits_{n\,\rightarrow\,\infty}\,\mathds{E}\left[u^{i}\!\left(\mathrm{w}^{i,\,\rho_n}_{T}\right)\right]\, =\, \mathds{E}\left[u^{i}\!\left(\mathrm{w}^{i,\,\rho}_{T}\right)\right]\, =\, H^{i}\!\left(\Pi^{\,\rho},X^{-i}\right)
\end{equation*}

\par Combining the equation above with (\ref{vrho}) it then follows that we have
\begin{equation*}
V^{i}\!\left(s, z^{i}; X^{-i}\right) - \rho\, \leq\, H^{i}\!\left(\Pi^{\,i,\,\rho},X^{-i}\right)\, = \, \lim\limits_{n\,\rightarrow\,\infty}\,H^{i}\!\left(\Pi^{\,i,\,\rho_n},X^{-i}\right)
\end{equation*}
 
\par Given the equation above, it is immediate that we have
\begin{equation*}
 V^{i}\!\left(s, z^{i}; X^{-i}\right) - \rho\, \leq\, \limsup\limits_{n\,\rightarrow\,\infty}\,H^{i}\!\left(\Pi^{\,i,\,\rho_n},X^{-i}\right)\, \leq\, \limsup\limits_{n\,\rightarrow\,\infty} \left\{\,\sup\limits_{\Pi_{n}\, \in\, \mathcal{A}^{a,\,pc}_{s}\!\left(I_{n}\right)} H^{i}\!\left(\Pi_{n};X^{-i}\right)\,\right\}
\end{equation*}
 
\par Note that since the choice of $\rho$ above was arbitrary and in view of (\ref{vgpc}), the claim in the statement of the lemma follows immediately from the above.
\end{proof}

\end{appendices}
\end{document}